\documentclass[12pt]{PisaPhdThesis}
\usepackage{amsmath}
\usepackage{amssymb}
\usepackage{amsthm}
\usepackage{bm}
\usepackage{algorithmic}
\usepackage{algorithm}
\usepackage{tikz}
\usepackage{subfigure}
\usepackage{setspace}
\usepackage{enumitem}
\usepackage{pgf,tikz}
\usepackage[pagebackref=true]{hyperref}
\usetikzlibrary{arrows}
\usetikzlibrary{snakes}

\newtheorem{theorem}{Theorem}[chapter]
\newtheorem{lemma}{Lemma}[chapter]
\newtheorem{corollary}[theorem]{Corollary}
\newtheorem{problem}[theorem]{Problem}
\newtheorem{invariant}{Invariant}

\newcommand{\edel}{\ensuremath{\mathtt{del}}}
\newcommand{\idel}{\ensuremath{\mathtt{undel}}}
\newcommand{\vdel}{\ensuremath{\mathtt{del}}}
\newcommand{\chooseedge}{\ensuremath{\mathtt{choose}}}
\newcommand{\kdfs}{\ensuremath{\mathtt{dfs}_k}}
\newcommand{\mkdfs}{\ensuremath{\mathtt{mdfs}_k}}
\newcommand{\treeoutput}{\ensuremath{\mathtt{output}}}
\newcommand{\listtrees}{\ensuremath{\mathtt{ListTrees}_{v_i}}}

\newcommand{\promote}{\ensuremath{\mathtt{promote}}}
\newcommand{\fdeg}{\ensuremath{\mathtt{deg}}}
\newcommand{\adj}{\ensuremath{\mathtt{adj}}}
\newcommand{\isunary}{\ensuremath{\mathtt{is\_unary}}}
\newcommand{\treecut}{\ensuremath{\mathtt{treecut}}}

\newcommand{\unary}{\ensuremath{\mathtt{unary}}}

\newcommand{\comb}{\ensuremath{\mathit{comb}}}

\newcommand{\setofisg}{\ensuremath{\mathcal{S}}}
\newcommand{\choosevertex}{\ensuremath{\mathtt{choose}}}
\newcommand{\certificate}{\ensuremath{\mathtt{certificate}}}
\newcommand{\cleft}{\ensuremath{\mathtt{update\_left}}}
\newcommand{\cright}{\ensuremath{\mathtt{update\_right}}}
\newcommand{\listsubgraphs}{\ensuremath{\mathtt{ListSubgraphs}_{v_i}}}
\newcommand{\restore}{\ensuremath{\mathtt{restore}}}
\newcommand{\removelastleaf}{\ensuremath{\mathtt{removelastleaf}}}

\newcommand{\setofcycles}{\ensuremath{\mathcal{C}}}
\newcommand{\setofpaths}{\ensuremath{\mathcal{P}}}
\newcommand{\beadstring}{\ensuremath{B_{u,t}}}
\newcommand{\sbeadstring}{\ensuremath{B_{s,t}}}
\newcommand{\vbeadstring}{\ensuremath{B_{v,t}}}
\newcommand{\head}{\ensuremath{H_u}}

\newcommand{\chead}{\ensuremath{H_X}}

\newcommand{\del}{\ensuremath{\mathtt{del}}}

\newcommand{\undel}{\ensuremath{\mathtt{undel}}}

\newcommand{\oracleleft}{\ensuremath{\mathtt{left\_update}}}
\newcommand{\oracleright}{\ensuremath{\mathtt{right\_update}}}
\newcommand{\undooracle}{\ensuremath{\mathtt{restore}}}

\newcommand{\bcc}{\textsc{bcc}}

\newcommand{\liststpaths}{\ensuremath{\mathtt{list\_paths}_{s,t}}}

\newcommand{\true}{\ensuremath{\mathtt{true}}}
\newcommand{\false}{\ensuremath{\mathtt{false}}}
\newcommand{\print}{\ensuremath{\mathtt{output}}}
\newcommand{\return}{\ensuremath{\mathtt{return}}}
\newcommand{\routput}{\ensuremath{\mathtt{output}}}

\makeindex

\begin{document}

\title{Efficiently Listing Combinatorial Patterns in Graphs}
\author{Rui Ferreira}
\supervisor{Roberto Grossi \and Romeo Rizzi}
\referee{Takeaki Uno \and St\'ephane Vialette} 
\date{May 2013\\
\medskip
Direttore della Scuola: Pierpaolo Degano\\
\medskip
\medskip
\medskip
\medskip
\medskip
SSD: INF/01 - Informatica}
\maketitle

\chapter*{Acknowledgments}

First and foremost, a heartfelt thank you to my advisers Roberto
Grossi and Romeo Rizzi \emph{-- this thesis would not have been possible
without your deep insights, support and encouragement}. My gratitude
to my co-authors Etienne Birmel\'e, Pierluigi Crescenzi, Roberto
Grossi, Vicent Lacroix, Andrea Marino, Nadia Pisanti, Romeo Rizzi,
Gustavo Sacomoto and Marie-France Sagot. I would also like to thank
everyone at the Department of Computer Science of the University of
Pisa, who made me feel at home in a distant country. It has been a
pleasure working surrounded by these great people.

\medskip

My deepest gratitude to my parents Fernando and Cila for their
encouragement and unconditional love. A big hug to my sister Ana,
brother-in-law Paulo, Andre \& Miguel. To Hugo, Jo\~ao, Tina, Ernesto
and all my family. \emph{-- I cannot express how sorry I am for being away
from you all.}

\medskip

Last but not least, to Desi for being by my side.

\medskip

\hfill --- Rui

\begin{abstract}

	Graphs are extremely versatile and ubiquitous mathematical
	structures with potential to model a wide range of domains.
	For this reason, graph problems have been of interest since
	the early days of computer science. Some of these problems
	consider substructures of a graph that have certain
	properties.  These substructures of interest, generally called
	patterns, are often meaningful in the domain being modeled.
	Classic examples of patterns include spanning trees, cycles
	and subgraphs.
	
	This thesis focuses on the topic of explicitly listing all the
	patterns existing in an input graph. One of the defining
	features of this problem is that the number of patterns is
	frequently exponential on the size of the input graph. Thus,
	the time complexity of listing algorithms is parameterized
	by the size of the output.

	The main contribution of this work is the presentation of
	optimal algorithms for four different problems of listing
	patterns in graphs. These algorithms are framed within the
	same generic approach, based in a recursive partition of the
	search space that divides the problem into subproblems. The
	key to an efficient implementation of this approach is to
	avoid recursing into subproblems that do not list any
	patterns. With this goal in sight, a dynamic data structure,
	called the certificate, is introduced and maintained
	throughout the recursion. Moreover, properties of the
	recursion tree and lower bounds on the number of patterns are
	used to amortize the cost of the algorithm on the size of the
	output.
	
	The first problem introduced is the listing of all
	$k$-subtrees: trees of fixed size~$k$ that are subgraphs of an
	undirected input graph. The solution is presented
	incrementally to illustrate the generic approach until an
	optimal output-sensitive algorithm is reached. This algorithm
	is optimal in the sense that it takes time proportional to the
	time necessarily required to read the input and write the
	output.

	The second problem is that of listing $k$-subgraphs: connected
	induced subgraphs of size $k$ in an undirected input
	graph. An optimal algorithm is presented, taking time
	proportional to the size of the input graph plus the
	edges in the $k$-subgraphs.

	The third and fourth problems are the listing of cycles and
	listing of paths between two vertices in an undirected input
	graph.  An optimality-preserving reduction from listing cycles
	to listing paths is presented. Both problems are solved
	optimally, in time proportional to the size of the input plus
	the size of the output.

	The algorithms presented improve previously known solutions
	and achieve optimal time bounds. The thesis concludes by
	pointing future directions for the topic.
\end{abstract}

\tableofcontents
\listoffigures

\chapter{Introduction}

This chapter presents the contributions and organization of this
thesis. Let us start by briefly framing the title \emph{``Efficiently
Listing Combinatorial Patterns in Graphs''}.

\medskip

{\bf Graphs.}
Graphs are an ubiquitous abstract model of both natural and man-made
structures. Since Leonhard Euler's famous use of graphs to solve the
Seven Bridges of K\"onigsberg
\cite{euler1956seven,gribkovskaia2007bridges}, graph models have been
applied in computer science, linguistics, chemistry, physics and
biology \cite{bondy1976graph,pirzada2008applications}.  Graphs, being
discrete structures, are posed for problems of combinatorial
nature~\cite{bondy1976graph} -- often easy to state and difficult to
solve.

\medskip

{\bf Combinatorial patterns.}
The term ``pattern'' is used as an umbrella of concepts to describe
substructures and attributes that are considered to have significance
in the domain being modeled. Searching, matching, enumerating and
indexing patterns in strings, trees, regular expressions or graphs
are widely researched areas of computer science and mathematics.  This
thesis is focused on patterns of combinatorial nature. In the
particular case of graphs, these patterns are subgraphs or
substructures that have certain properties. Examples of combinatorial
patterns in graphs include spanning trees, graph minors, subgraphs and paths
\cite{knuth2006art,ruskey2003combinatorial}.

\medskip

{\bf Listing.} The problem of graph enumeration, pioneered by
P\'olya, Caley and Redfield, focuses on counting the number of graphs
with certain properties as a function of the number of vertices and
edges in those graphs \cite{harary1973graphical}. Philippe Flajolet
and Robert Sedgewick have introduced techniques for deriving
generating functions of such objects \cite{flajolet2009analytic}. The
problem of counting patterns occurring in graphs (as opposed to the
enumeration of the graphs themselves) is more algorithmic in nature.
With the introduction of new theoretical tools, such as parameterized
complexity
theory~\cite{downey1999parameterized,flum2006parameterized}, several
new results have been achieved in the last
decade~\cite{guillemot2010finding,koutis2009limits}.

In the research presented, we tackle the closely related problem of
listing patterns in graphs: i.e.~explicitly outputting each pattern
found in an input graph. Although this problem can be seen as a type
of exhaustive enumeration (in fact, it is common in the literature
that the term enumeration is used), the nature of both problems is
considered different as the patterns have to be explicitly generated
\cite{knuth2006art,ruskey2003combinatorial}.

\medskip

{\bf Complexity.}
One of the defining features of the problem of listing combinatorial
patterns is that there frequently exists an exponential number of
patterns in the input graph. This implies that there are no
polynomial-time algorithms for this family of problems.  Nevertheless,
the time complexity of algorithms for the problem of listing patterns
in graphs can still be analyzed. The two most common approaches are:
(i) output-sensitive analysis, i.e.~analyzing the time complexity of
the algorithm as a function of the output and input size, and (ii)
time delay, i.e.~bounding the time it takes from the output of one
pattern to the next in function of the size of the input graph and the
pattern.

\section{Contribution}

This thesis presents a new approach for the problem of listing
combinatorial patterns in an input graph $G$ with uniquely labeled
vertices. At the basis of our technique
\cite{Ferreira11,ferreira2013}, we list the set of patterns in $G$ by
recursively using a binary partition: (i) listing the set of patterns
that include a vertex (or edge), and (ii) listing the set of patterns
that do not include that vertex (resp.~edge).

The core of the technique is to maintain a dynamic data structure that
efficiently tests if the set of patterns is empty, avoiding
branches of the recursion that do not output any patterns.  This
problem of dynamic data structures is very different from the
classical view of fully-dynamic or decrementally-dynamic data
structures. Traditionally, the cost of operations in a dynamic data
structure is amortized over a sequence of arbitrary operations. In
our case, the binary-partition method induces a well defined order in
which the operations are performed. This allows a much better
amortization of the cost.  Moreover, the existence of lower bounds
(even if very weak) on the number of the combinatorial patterns in
a graph allows a better amortization of the cost of the recursion and
the maintenance of the dynamic data structure.

\medskip

This approach is applied to the listing of the following four
different patterns in undirected graphs: $k$-subtrees\footnote{This
result has been published in \cite{Ferreira11}.}, $k$-subgraphs,
$st$-paths\footnotemark[2] and cycles\footnotemark[2]. In all four
cases, we obtain optimal output-sensitive algorithms, running in time
proportional to the size of the input graph plus the size of the
output patterns.

\footnotetext[2]{These results have been published in \cite{ferreira2013}.}

\subsection{Listing $k$-subtrees}

Given an undirected connected graph $G=(V,E)$, where $V$ is the set of
vertices and $E$ the set of edges, a $k$-subtree $T$ is a set of edges
$T \subseteq E$ that is acyclic, connected and contains $k$ vertices.
We present the first optimal output-sensitive algorithm for listing
the $k$-subtrees in input graph~$G$. When $s$ is the number of
$k$-subtrees in $G$, our algorithm solves the problem in~$O(s k)$
time.

We present our solution starting from the binary-partition method. We
divide the problem of listing the $k$-subtrees in $G$ into two
subproblems by taking an edge~$e \in E$: (i) we list the $k$-subtrees
that contain~$e$, and (ii) list those that do not contain~$e$. We
proceed recursively on these subproblems until there is just one
$k$-subtree to be listed. This method induces a binary recursion
tree, and all the $k$-subtrees are listed when reaching the leaves of
this recursion tree.

Although this first approach is output sensitive, it is not optimal.
One problem is that each adjacency list in $G$ can be of length
$O(n)$, but we cannot pay such a cost in each recursive call. Also, we
need to maintain a certificate throughout the recursive calls to
guarantee \emph{a priori} that at least one $k$-subtree
will be generated. By exploiting more refined structural properties of the
recursion tree, we present our algorithmic ideas until an optimal
output-sensitive listing is obtained.

\subsection{Listing $k$-subgraphs}

When considering an undirected connected graph $G$, a $k$-subgraph is
a connected subgraph of $G$ induced by a set of $k$ vertices. We
propose an algorithm to solve the problem of listing all the
$k$-subgraphs in $G$. Our algorithm is optimal: solving the problem in
time proportional to the size of the input graph $G$ plus the size of
the edges in the $k$-subgraphs to output.

We apply the binary-partition method, dividing the problem of listing
all the $k$-subgraphs in two subproblems by taking a vertex $v \in V$:
we list the $k$-subgraphs that contain $v$; and those that do not
contain $v$. We recurse on these subproblems until no $k$-subgraphs
remain to be listed. This method induces a binary recursion tree where
the leaves correspond to $k$-subgraphs.

In order to reach an optimal algorithm, we maintain a certificate that
allows us to determine efficiently if there exists at least one
$k$-subgraph in $G$ at any point in the recursion.  Furthermore, we
select the vertex $v$ (which divides the problem into two
subproblems), in a way that facilitates the maintenance of the
certificate.

\subsection{Listing cycles and $st$-paths}

 Listing all the simple cycles (hereafter just called cycles) in a
 graph is a classical problem whose efficient solutions date back to
 the early 70s. For a directed graph with $n$ vertices and $m$ edges,
 containing $\eta$ cycles, the best known solution in the literature
 is given by Johnson's algorithm \cite{Johnson1975} and takes
 $O((\eta+1)(m+n))$ time. This algorithm can be adapted to undirected
 graphs, maintaining the same time complexity. Surprisingly, this time
 complexity is not optimal for undirected graphs: to the best of our
 knowledge, no theoretically faster solutions have been proposed in
 almost 40 years.

  We present the first optimal solution to list all the cycles in an
  undirected graph $G$, improving the time bound of Johnson's
  algorithm.  Specifically, let $\setofcycles(G)$ denote the set of
  all these cycles, and observe that $|\setofcycles(G)| = \eta$. For a
  cycle $c \in \setofcycles(G)$, let $|c|$ denote the number of edges
  in~$c$. Our algorithm requires $O(m + \sum_{c \in
  \setofcycles(G)}{|c|})$ time and is asymptotically optimal: indeed,
  $\Omega(m)$ time is necessarily required to read the input $G$, and
  $\Omega(\sum_{c \in \setofcycles(G)}{|c|})$ time is required to list
  the output.

  We also present the first optimal solution to list all the simple
  paths from~$s$ to~$t$ (shortly, $st$-paths) in an undirected graph
  $G$. Let $\setofpaths_{st}(G)$ denote the set of $st$-paths in $G$
  and, for an $st$-path $\pi \in \setofpaths_{st}(G)$, let $|\pi|$ be
  the number of edges in $\pi$.  Our algorithm lists all the
  $st$-paths in~$G$ optimally in $O(m + \sum_{\pi \in
  \setofpaths_{st}(G)}{|\pi|})$ time, observing that $\Omega(\sum_{\pi
  \in \setofpaths_{st}(G)}{|\pi|})$ time is required to list the
  output.

  While the basic approach is simple, we use a number of non-trivial
  ideas to obtain our optimal algorithm for an undirected 
  graph $G$. We start by presenting an optimality-preserving reduction
  from the problem of listing cycles to the problem of listing
  $st$-paths. Focusing on listing $st$-paths, we consider the
  decomposition of the graph into biconnected components and use the
  property that $st$-paths pass in certain articulation points to
  restrict the problem to one biconnected component at a time. We then
  use the binary-partition method to list: (i) $st$-paths
  that contain an edge $e$, and (ii) those that do not contain
  $e$. We make use of a certificate of existence of at least one
  $st$-path and prove that the cost of maintaining the certificate
  throughout the recursion can be amortized. A key
  factor of the amortization is the existence of a lower bound on the
  number of $st$-paths in a biconnected component.

\section{Organization}

After this brief introduction, let us outline the organization of this
thesis.

\medskip

Chapter~\ref{chapter:Background} introduces background information on
the topic at hand. In Section~\ref{section:Graphs}, we start by
illustrating uses of graphs and reviewing basic concepts in graph
theory. Section~\ref{section:CombinatorialPatternsInGraphs} motivates
the results presented, by introducing combinatorial patterns and their
applications in diverse areas. We proceed into
Section~\ref{section:ListingAndEnumeration} where an overview of 
listing and enumeration of patterns is given and the state of the art
is reviewed. For a review of the state of the art on each problem
tackled, the reader is referred to the introductory text of the
respective chapter. We finalize the chapter with
Section~\ref{section:OverviewOfTheTechniquesUsed} which includes an
overview of the techniques used throughout the thesis and explains how
they are framed within the state of the art.

Chapters~\ref{chapter:ListingKSubtrees} to \ref{chapter:cycles} give a
detailed and formal view of the problems handled and results achieved.
In Chapter~\ref{chapter:ListingKSubtrees} we provide an incremental
introduction to the optimal output-sensitive algorithm to list
$k$-subtrees, starting from a vanilla version of the binary-partition
method and step by step introducing the certificate and amortization
techniques used.  Chapter~\ref{chapter:ListingKSubgraphs} directly
presents the optimal output-sensitive algorithm for the problem of
listing $k$-subgraphs.  Similarly, the results achieved on the
problems of listing cycles and $st$-paths are presented in
Chapter~\ref{chapter:cycles}.

We finalize this exposition with Chapter~\ref{chapter:Conclusion},
reviewing the main contributions and presenting some future directions
and improvements to the efficient listing of combinatorial patterns in
graphs.

\chapter{Background}
\label{chapter:Background}

\section{Graphs}
\label{section:Graphs}

\begin{figure}[t]
\centering
\subfigure[Undirected graph $G_1$ \label{fig:undirgraph}]{

\begin{tikzpicture}[shorten >=1pt,->,scale=1.3]
  \tikzstyle{vertex}=[shape=circle,draw,thick,fill=black,minimum size=2pt,inner sep=2pt]
  \node[vertex,label=right:$a$] (A) at (0,0) {};
  \node[vertex,label=left:$b$] (B) at (-1,-1)   {};
  \node[vertex,label=left:$c$] (C) at (-1,-2)  {};
  \node[vertex,label=right:$d$] (D) at (1,-2)  {};
  \node[vertex,label=right:$e$] (E) at (1,-1)  {};
  \draw (B) -- (C) -- (D) -- (E) -- cycle;
  \draw (B) -- (E) -- cycle;
  \draw (C) -- (E) -- cycle;
  \draw (A) -- (E) -- cycle;
\end{tikzpicture}
}
\subfigure[Directed graph $G_2$\label{fig:dirgraph}]{

\begin{tikzpicture}[shorten >=1pt,->,scale=1.3]
  \tikzstyle{vertex}=[shape=circle,draw,thick,fill=black,minimum size=2pt,inner sep=2pt]
  \node[vertex,label=right:$a$] (A) at (0,0) {};
  \node[vertex,label=left:$b$] (B) at (-1,-1)   {};
  \node[vertex,label=left:$c$] (C) at (-1,-2)  {};
  \node[vertex,label=right:$d$] (D) at (1,-2)  {};
  \node[vertex,label=right:$e$] (E) at (1,-1)  {};
  \draw (A) -- (B);
  \draw (B) -- (C);
  \draw (C) -- (D);
  \draw (E) -- (C);
  \draw (D) -- (E);
  \draw (A) -- (E);
  \draw (B) -- (E);
\end{tikzpicture}
}
\subfigure[Biconnected graph~$G_3$\label{fig:undirgraph2}]{

\begin{tikzpicture}[shorten >=1pt,->,scale=1.3]
  \tikzstyle{vertex}=[shape=circle,draw,thick,fill=black,minimum size=2pt,inner sep=2pt]
  \node[vertex,label=right:$a$] (A) at (0,0) {};
  \node[vertex,label=left:$b$] (B) at (-1,-1)   {};
  \node[vertex,label=left:$c$] (C) at (-1,-2)  {};
  \node[vertex,label=right:$d$] (D) at (1,-2)  {};
  \node[vertex,label=right:$e$] (E) at (1,-1)  {};
  \draw (A) -- (B) -- (C) -- (D) -- (E) -- (A) -- cycle;
  \draw (B) -- (E) -- cycle;
\end{tikzpicture}
}
\caption{Examples of graphs\label{fig:3examplegraphs}}
\end{figure}
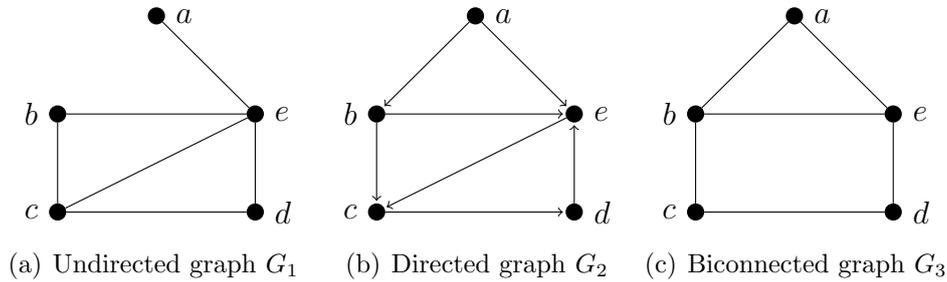

Graphs are abstract mathematical structures that consist of objects,
called \emph{vertices}, and links that connect pairs of vertices,
called \emph{edges}.  Although the term graph was only coined in 1876
by Sylvester, the study of graphs dates back to 1736 when Euler
published his famous paper on the bridges of K\"oningsberg
\cite{gribkovskaia2007bridges,euler1956seven}. With the contributions
of Caley, P\'olya, De Bruijn, Peterson, Robertson, Seymour, Erd\"os,
R\'enyi and many many others, the field of graph theory developed and
provided the tools for what is considered one of the most versatile
mathematical structures.

A classical example of a graph is the rail
network of a country: vertices represent the stations and there is an
edge between two vertices if they represent consecutive stations along
the same railroad. This is an example of an \emph{undirected graph}
since the notion of consecutive station is a symmetric relation.
Another example is provided by the World Wide Web: the vertices are
the websites and two are adjacent if there is a direct link from one
to the other. Graphs of this latter type are called \emph{directed
graphs}, since a link from one website to the other does not imply the
reverse. These edges are sometimes called \emph{directed edges} or
\emph{arcs}.

\subsection{Applications}
\label{subsection:GraphApplications}
There are vast practical uses of graphs. As an example, Stanford Large
Dataset Collection \cite{leskovec2012stanford} includes graphs from
domains raging from jazz musicians to metabolic networks. In this
section, we take a particular interest in graphs as a model of
biological functions and processes.  For further information about
different domains, \cite{pirzada2008applications,bondy1976graph} are
recommended as a starting point.

\medskip

{\bf Metabolic networks.} The physiology of living cells is controlled
by chemical reactions and regulatory interactions.  Metabolic pathways
model these individual processes. The collection of distinct pathways
co-existing within a cell is called a metabolic network.
Interestingly, with the development of the technology to sequence the
genome, it has been possible to link some of its regions to metabolic
pathways. This allows a better modeling of these networks and, through
simulation and reconstruction, it possible to have in-depth insight
into the key features of the metabolism of organisms.

\medskip

{\bf Protein-protein interaction networks.} Proteins have many
biological functions, taking part in processes such as DNA replication
and mediating signals from the outside to the inside of a cell. In
these processes, two or more proteins bind together. These
interactions play important roles in diseases (e.g.~cancer) and can be
modeled through graphs.

\medskip

{\bf Gene regulatory networks.} Segments of DNA present in a cell
interact with each other through their RNA and protein expression
products. There are also interactions occurring with the other
substances present in the cell. These processes govern the rates of
transcription of genes and can be modeled through a network.

\medskip

{\bf Signaling networks.} These networks integrate metabolic,
protein-protein interaction and gene regulatory networks to model how
signals are transduced within or between cells.

\subsection{Definitions}
\label{subsection:GraphDefinitions}

Let us now introduce some formalism and notation when dealing with
graphs. We will recall some of the concepts introduced here on the
preliminaries of
Chapters~\ref{chapter:ListingKSubtrees}-\ref{chapter:cycles}.

\medskip

{\bf Undirected graphs.} A graph~$G$ is an ordered
pair~$G=(V,E)$ where~$V$ is the set of vertices and~$E$ is the set of
edges. The \emph{order} of a graph is the number of vertices $|V|$ and
its \emph{size} the number of edges $|E|$.
In the case of \emph{undirected graphs}, an edge~$e \in E$ is an
unordered pair~$e=(u,v)$ with~$u,v \in V$. Graph~$G$ is said to be
simple if: (i) there are no \emph{loops}, i.e.~edges that start and
end in the same vertex, and (ii) there are no \emph{parallel edges},
i.e.~multiple edges between the same pair of vertices.  Otherwise, $E$
is a multiset and $G$ is called a \emph{multigraph}.
Figure~\ref{fig:undirgraph} shows the example of undirected
graph~$G_1$.

Given a simple undirected graph $G=(V,E)$, vertices $u,v
\in V$ are said to be \emph{adjacent} if there exists an edge $(u,v)
\in E$. The set of vertices adjacent to a vertex $u$ is called its
\emph{neighborhood} and is denoted by $N(u) = \{ v \mid (u,v) \in E
\}$. An edge $e \in E$ is incident in $u$ if $u \in e$. The
\emph{degree} $\deg(u)$ of a vertex $u \in V$ is the number of edges
incident in $u$, that is $\deg(u) = |N(u)|$. 

\medskip

{\bf Directed graphs.} In case of \emph{directed graph}
$G=(V,E)$, an edge~$e \in E$ has an orientation and therefore is an
\emph{ordered} pair of vertices $e=(u,v)$. Figure~\ref{fig:dirgraph}
shows the example of directed graph~$G_2$. The neighborhood $N(u)$ of
a vertex $u \in V$, is the union of the \emph{out-neighborhood}
$N^+(u)=\{ v \mid (u,v) \in E \}$ and \emph{in-neighborhood}
$N^-(u)=\{ v \mid (v,u) \in E \}$.  Likewise, the degree $deg(u)$ is
the union of the \emph{out-degree}~$deg^+(u) = |N^+(u)|$ and the
\emph{in-degree}~$deg^-(u) = |N^-(u)|$.

\medskip

{\bf Subgraphs.} A graph $G'=(V',E')$ is said to be a
\emph{subgraph} of $G=(E,V)$ if $V' \subseteq V$ and $E' \subseteq
E$. The subgraph $G'$ is said to be \emph{induced} if and only if $e
\in E'$ for every edge $e = (u,v) \in E$ with $u,v \in V'$. A subgraph
of $G$ induced by a vertex set $V'$ is denoted by $G[V']$.

\medskip

{\bf Biconnected graphs.} An undirected graph $G=(V,E)$ is said
to be \emph{biconnected} if it remains connected after the removal of
any vertex $v \in V$. The complete graph of two vertices if considered
biconnected.  Figure~\ref{fig:undirgraph2} shows the example of
biconnected graph~$G_3$. A \emph{biconnected component} is a maximal
biconnected subgraph. An \emph{articulation point} (or \emph{cut
vertex}) is any vertex whose removal increases the number of
biconnected components in $G$.  Note that any connected graph can be
decomposed into a tree of biconnected components, called the
\emph{block tree} of the graph.  The biconnected components in the
block tree are attached to each other by shared articulation points.

\section{Combinatorial patterns in graphs}
\label{section:CombinatorialPatternsInGraphs}

\begin{figure}[t]
\centering
\subfigure[4-subtree $T_1$ of $G_1$\label{fig:ksubtreeexample}]{

\begin{tikzpicture}[shorten >=1pt,->,scale=1.3]
  \tikzstyle{vertex}=[shape=circle,draw,thick,fill=black,minimum size=2pt,inner sep=2pt]
  \node[vertex,label=right:$a$] (A) at (0,0) {};
  \node[vertex,label=left:$c$] (C) at (-1,-2)  {};
  \node[vertex,label=right:$d$] (D) at (1,-2)  {};
  \node[vertex,label=right:$e$] (E) at (1,-1)  {};
  \draw (E) -- (D) -- cycle;
  \draw (E) -- (A) -- cycle;
  \draw (C) -- (E) -- cycle;
\end{tikzpicture}
}
\subfigure[4-subgraph $G_1{[\{b,c,d,e\}]}$\label{fig:ksubgraphexample}] {

\begin{tikzpicture}[shorten >=1pt,->,scale=1.3]
  \tikzstyle{vertex}=[shape=circle,draw,thick,fill=black,minimum size=2pt,inner sep=2pt]
  \node[vertex,label=left:$b$] (B) at (-1,-1)   {};
  \node[vertex,label=left:$c$] (C) at (-1,-2)  {};
  \node[vertex,label=right:$d$] (D) at (1,-2)  {};
  \node[vertex,label=right:$e$] (E) at (1,-1)  {};
  \draw (B) -- (C) -- (D) -- (E) -- cycle;
  \draw (B) -- (E) -- cycle;
  \draw (C) -- (E) -- cycle;
  \path (-2,-1) --  (2, -1);
\end{tikzpicture}
}
\subfigure[$ab$-path $\pi_1$ in $G_1$\label{fig:stpathexample}]{

\begin{tikzpicture}[shorten >=1pt,->,scale=1.3]
  \tikzstyle{vertex}=[shape=circle,draw,thick,fill=black,minimum size=2pt,inner sep=2pt]
  \node[vertex,label=right:$a$] (a) at (0,0) {};
  \node[vertex,label=left:$b$] (b) at (-1,-1)   {};
  \node[vertex,label=left:$c$] (c) at (-1,-2)  {};
  \node[vertex,label=right:$d$] (d) at (1,-2)  {};
  \node[vertex,label=right:$e$] (e) at (1,-1)  {};
  \draw (a) -- (e) -- (d) -- (c) -- (b) -- cycle;
\end{tikzpicture}
}

\subfigure[cycle $c_1$ in $G_1$\label{fig:cycleexample}]{

\begin{tikzpicture}[shorten >=1pt,->,scale=1.1]
  \tikzstyle{vertex}=[shape=circle,draw,thick,fill=black,minimum size=2pt,inner sep=2pt]
  \node[vertex,label=right:$a$] (a) at (0,0) {};
  \node[vertex,label=left:$b$] (b) at (-1,-1)   {};
  \node[vertex,label=left:$c$] (c) at (-1,-2)  {};
  \node[vertex,label=right:$d$] (d) at (1,-2)  {};
  \node[vertex,label=right:$e$] (e) at (1,-1)  {};
  \draw (a) -- (e) -- (d) -- (c) -- (b) -- (a) -- cycle;
  \path (-2,-1) --  (2, -1);
\end{tikzpicture}
}
\subfigure[induced path $\pi_2$ in $G_3$\label{fig:inducedpathexample}]{

\begin{tikzpicture}[shorten >=1pt,->,scale=1.1]
  \tikzstyle{vertex}=[shape=circle,draw,thick,fill=black,minimum size=2pt,inner sep=2pt]
  \node[vertex,label=right:$a$] (a) at (0,0) {};
  \node[vertex,label=left:$c$] (c) at (-1,-2)  {};
  \node[vertex,label=right:$d$] (d) at (1,-2)  {};
  \node[vertex,label=right:$e$] (e) at (1,-1)  {};
  \draw (a) -- (e) -- (d) -- (c) -- cycle;
  \path (-2,-1) --  (2, -1);
\end{tikzpicture}
}

\subfigure[chordless cycle $c_2$ in $G_3$\label{fig:chordlesscycleexample}]{

\begin{tikzpicture}[shorten >=1pt,->,scale=1.1]
  \tikzstyle{vertex}=[shape=circle,draw,thick,fill=black,minimum size=2pt,inner sep=2pt]
  \node[vertex,label=left:$b$] (B) at (-1,-1)   {};
  \node[vertex,label=left:$c$] (C) at (-1,-2)  {};
  \node[vertex,label=right:$d$] (D) at (1,-2)  {};
  \node[vertex,label=right:$e$] (E) at (1,-1)  {};
  \draw (B) -- (E) -- (D) -- (C) -- (B) -- cycle;
  \path (-2,-1) --  (2, -1);
\end{tikzpicture}
}
\subfigure[$ac$-bubble $b_1$ in $G_2$\label{fig:bubbleexample}]{

\begin{tikzpicture}[shorten >=1pt,->,scale=1.1]
  \tikzstyle{vertex}=[shape=circle,draw,thick,fill=black,minimum size=2pt,inner sep=2pt]
  \node[vertex,label=right:$a$] (A) at (0,0) {};
  \node[vertex,label=left:$b$] (B) at (-1,-1)   {};
  \node[vertex,label=left:$c$] (C) at (-1,-2)  {};
  \node[vertex,label=right:$e$] (E) at (1,-1)  {};
  \draw (A) -- (B);
  \draw (B) -- (C);
  \draw (E) -- (C);
  \draw (A) -- (E);
\end{tikzpicture}
}
\caption{Examples of patterns\label{fig:examplepatterns}}
\end{figure}
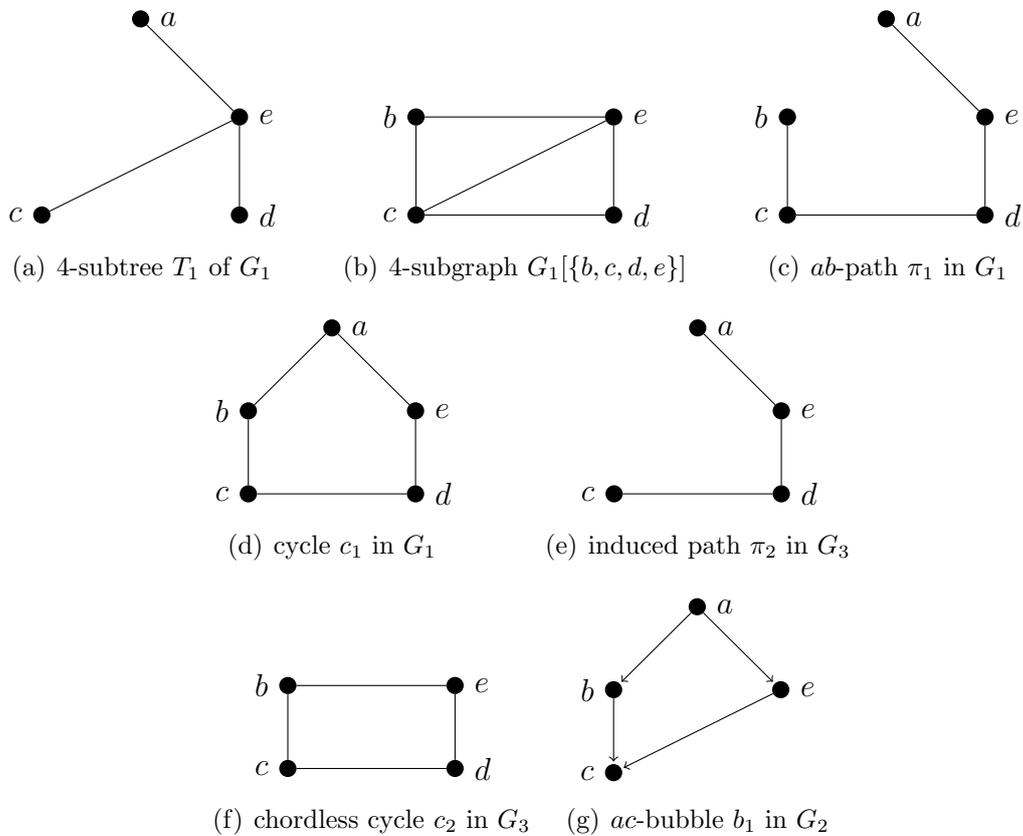

Combinatorial patterns in an input graph are constrained substructures
of interest for the domain being modeled. As an example, a cycle in a
graph modeling the railway network of a country could represent a
service that can be efficiently performed by a single train. In
different domains, there exists a myriad of different structures with
meaningful interpretations. We focus on generic patterns such as
subgraphs, trees, cycles and paths, which have broad applications in
different domains. In addition to the patterns studied in this thesis,
we also introduce other patterns where the techniques presented are
likely to have applications.

\subsection{$k$-subtrees}
\label{subsection:kSubtrees}

Given a simple (without loops or parallel edges), undirected and
connected graph $G = (V,E)$ with $n:=|V|$ and an integer $k \in
[2,n]$, a \emph{$k$-subtree} $T$ of $G$ is an acyclic connected
subgraph of $G$ with $k$ vertices. Formally, $T \subseteq E$, $|T| =
k-1$, and $T$ is an unrooted, unordered, free tree.
Figure~\ref{fig:ksubtreeexample} shows the example of the $4$-subtree
$T_1$ of graph $G_1$ from Figure~\ref{fig:undirgraph}.

\medskip

Trees that are a subgraph of an input graph $G=(V,E)$ have been
considered significant since the early days of graph theory. Examples
of such trees that have been deeply studied include spanning trees
\cite{graham1982history} and Steiner trees \cite{hwang1992steiner}.
Interestingly, when $k=|V|$, the $k$-subtrees of $G$ are its spanning
trees. Thus, $k$-subtrees can be considered a generalization of
spanning trees.

In several domains, e.g.~biological networks, researchers are
interested in local
structures~\cite{girvan2002community,boccaletti2006complex} that
represent direct interactions of the components of the network (see
Section~\ref{subsection:GraphApplications}). In
these domains, $k$-subtrees model acyclic interactions between those
components. In \cite{wasa2012}, the authors present related
applications of $k$-subtrees.

\subsection{$k$-subgraphs}
\label{subsection:InducedSubgraphs}
Given an undirected graph $G=(V,E)$, a set of vertices $V' \subseteq
V$ induces a subgraph $G[V']=(V',E')$ where $E' = \{ (u,v) \mid u,v
\in V' \text{ and } (u,v) \in E \}$. A $k$-subgraph is a connected
induced subgraph $G[V']$ with $k$ vertices.
Figure~\ref{fig:ksubgraphexample} shows the example of the
$4$-subgraph $G_1[V']$ with $V'=\{b,c,d,e\}$ (graph $G_1$ is available
on Figure~\ref{fig:undirgraph}).

\medskip

The subgraphs of an input graph have the been subject of study
\cite{ullmann1976algorithm,mcgregor1982backtrack,barnard1993substructure,strogatz2001exploring}
of many researchers. We highlight the recent interest of the
bioinformatics community in network motifs \cite{milo}. Network motifs
are $k$-subgraphs that appear more frequently in an input graph than
what would be expected in random networks. These motifs are likely to
be components in the function of the network. In fact, motifs
extracted from gene regulatory networks
\cite{lee2002transcriptional,milo,shen2002network} have been shown to
perform definite functional roles.

\subsection{$st$-paths and cycles}
\label{subsection:PathsAndCycles}

Given a directed or undirected graph $G=(V,E)$, a
\emph{path} in $G$ is a subgraph $\pi = (V',E') \subseteq
G$ of the form:
$$
V' = \{u_1,u_2, \ldots, u_k \}
\quad\quad
E' = \{ (u_1,u_2), (u_2,u_3), \ldots, (u_{k-1},u_{k}) \}
$$
where all the $u_i \in V'$ are distinct (and therefore a path is
simple by definition). We refer to a path $\pi$ by its natural
sequence of vertices or edges. A path $\pi$ from $s$ to $t$, or
$st$-\emph{path}, is denoted by $\pi = s \leadsto t$.
Figure~\ref{fig:stpathexample} shows the example of the
$ab$-path~$\pi_1$ in graph $G_1$ from Figure~\ref{fig:undirgraph}.

When $\pi = u_1, u_2, \ldots, u_k$ is a path, $k \ge 3$ and edge
$(u_k,u_1) \in E$ then $c = \pi + (u_k,u_1)$ is a cycle in $G$.
Figure~\ref{fig:cycleexample} shows the example of
the cycle~$c_1$ in graph $G_1$ from Figure~\ref{fig:undirgraph}.
We denote the number of edges in a path $\pi$ by $|\pi|$ and in a
cycle~$c$ by $|c|$.

\medskip

Cycles have broad applications in many fields. Following our
bioinformatics theme, cycles in metabolic networks represent cyclic
metabolic pathways which are known to often be functionally important
(e.g. Krebs cycle \cite{melendez1996puzzle}). Other domains where
cycles and $st$-paths are considered important include symbolic model
checking \cite{mcmillan1992symbolic}, telecommunication networks
\cite{gavish1982topological,schwartz1987telecommunication} and circuit
design~\cite{barahona1988application}.

\subsection{Induced paths, chordless cycles and holes}
\label{subsection:InducedPathsCyclesAndHoles}
Given an undirected graph $G=(V,E)$ and vertices $s,t \in V$, an
\emph{induced path} $\pi = s \leadsto t$ is an induced subgraph of $G$
that is a path from $s$ to $t$. By definition, there exists an edge
between each pair of vertices adjacent in $\pi$ and there are no edges
that connect two non-adjacent vertices.
Figure~\ref{fig:inducedpathexample} shows the example of the induced
path $\pi_2$ in graph $G_3$ from Figure~\ref{fig:undirgraph2}. Note
that $\pi_1$ in Figure~\ref{fig:inducedpathexample} is not an induced
path in $G_3$ due to the existence of edges $(a,b)$ and $(b,e)$.

Similarly, a \emph{chordless cycle} $c$ is an induced subgraph of $G$
that is a cycle. Chordless cycles are also sometimes called
\emph{induced cycles} or, when $|c|>4$, \emph{holes}.
Figure~\ref{fig:chordlesscycleexample} shows the example of the
chordless cycle $c_2$ in graph $G_3$ from Figure~\ref{fig:undirgraph2}.

\medskip

Many important graph families are characterized in terms of induced
paths and induced cycles. One example are chordal graphs: graphs with
no holes. Other examples include distance hereditary graphs
\cite{blair1991introduction}, trivially perfect graphs
\cite{golumbic1978trivially} and block graphs
\cite{harary1963characterization}.

\subsection{$st$-bubbles}
\label{subsection:Bubbles}
Given a directed graph $G=(V,E)$, and two vertices $s,t \in V$, a
$st$\emph{-bubble} $b$ is an unordered pair $b=(P,Q)$ of $st$-paths
$P, Q$ whose inner-vertices are disjoint (i.e.: $P \cap Q = \{s,t\}$).
The term \emph{bubble} (also called a \emph{mouth}) refers to
$st$-bubbles without fixing the source and target (i.e. $\forall s,t
\in V$). Figure~\ref{fig:bubbleexample} shows the example of the
$ac$-bubble $b_1$ in directed graph $G_2$ from Figure~\ref{fig:dirgraph}.

\medskip

Bubbles represent polymorphisms in models of the DNA. Specifically, in
De Bruijn graphs (a directed graph) generated by the reads of a DNA
sequencing project, bubbles can represent two different traits
occurring in the same species or population
\cite{pevzner2004novo,robertson2010novo}. Moreover, bubbles can can
also represents sequencing errors in the DNA sequencing project
\cite{simpson2009abyss,zerbino2008velvet,birmele2012}.

\section{Listing}
\label{section:ListingAndEnumeration}

Informally, given an input graph $G$ and a definition of pattern $P$,
a listing problem asks to output all substructures of graph $G$ that
satisfy the properties of pattern $P$.

Listing combinatorial patterns in graphs has been a long-time problem
of interest. In his 1970 book \cite{moon1970counting}, Moon remarks
``Many papers have been written giving algorithms, of varying degrees
of usefulness, for listing the spanning trees of a
graph''\footnote{This citation was found in
\cite{ruskey2003combinatorial}}. Among others, he cites
\cite{hakimi1961trees,duffin1959analysis,wang1934new,feussner2,feussner}
-- some of these early papers date back to the beginning of the XX
century. More recently, in the 1960s, Minty proposed an algorithm to
list all spanning trees \cite{minty}. Other results from Welch,
Tiernan and Tarjan for this and other problems soon followed
\cite{Welch66,Tiernan70,Tarjan73}.

It is not easy to find a reference book or survey with the state of
the art in this area,
\cite{harary1973graphical,Bezem87,ruskey2003combinatorial,knuth2006art}
partially cover the material. In this section we present a brief
overview of the theory, techniques and problems of the area. For a
review of the state of the art related with the problems tackled, we
refer the reader to the introductory text of each chapter.

\subsection{Complexity}

One defining characteristic of the problem of listing combinatorial
patterns in graphs is that the number of patterns present in the input
is often exponential in the input size. Thus, the number of patterns
and the output size have to be taken into account when analyzing the
time complexity of these algorithms. Several notions of
output-sensitive efficiency have been proposed:

\begin{enumerate}
\item \emph{Polynomial total time.} Introduced by Tarjan in
	\cite{Tarjan73}, a listing algorithm runs in polynomial total
	time if its time complexity is polynomial in the input and
	output size.

\item \emph{P-enumerability.} Introduced by Valiant in
	\cite{valiant1979permanent,valiant1979complexity}, an
	algorithm P-enumerates the patterns of a listing problem if
	its running time is $O(p(n)s)$, where $p(n)$ is a
	polynomial in the input size $n$ and $s$ is the number of
	solutions to output. When this algorithm uses space polynomial
	in the input size only, Fukuda introduced the term
	\emph{strong P-enumerability} \cite{fukuda1997analysis}.

\item \emph{Polynomial delay.} An algorithm has \emph{delay} $D$ if:
	(i) the time taken to output the first solution is $O(D)$, and
	(ii) the time taken between the output of any two consecutive
	solutions is $O(D)$. Introduced by Johnson, Yannakakis and
	Papadimitriou \cite{johnson1988generating}, an algorithm
	has polynomial (resp. linear, quadratic, \ldots) delay if $D$
	is polynomial (resp. linear, quadratic, \ldots) in the size of
	the input.
\end{enumerate}

In \cite{rospocher2006dit}, Rospocher proposed a hierarchy of
complexity classes that take these concepts into account. He
introduces a notion of reduction between these classes and some
listing problems were proven to be complete for the
class~$\mathcal{L}\texttt{P}$: the listing analogue of
class~$\texttt{\#P}$ for counting problems.

\medskip

We define an algorithm for a listing problem to be \emph{optimally
output sensitive} if the running time of the algorithm is $O(n+q)$,
where $n$ is the input size and $q$ is the size of the output.
Although this is a notion of optimality for when the output has to be
explicitly represented, it is possible that the output can be encoded
in the form of the differences between consecutive patterns in the
sequence of patterns to be listed.

\subsection{Techniques}

Since combinatorial patterns in graphs have different properties and
constraints, it is complex to have generic algorithmic techniques that
work for large classes of problems. Some of the ideas used for the
listing of combinatorial structures
\cite{goldberg2009efficient,kreher1998combinatorial,wilf1989combinatorial}
can be adapted to the listing of combinatorial patterns in those
structures. Let us present some of the known approaches.

\medskip

{\bf Backtrack search.} According to this approach, a
backtracking algorithm finds the solutions for a listing problem by
exploring the search space and abandoning a partial solution (thus,
the name ``backtracking'') that cannot be completed to a valid one.
For further information see~\cite{Read75}.

\medskip

{\bf Binary partition.} An algorithm that follows this approach
recursively divides the search space into two parts. In the case of
graphs, this is generally done by taking an edge (or a vertex) and:
(i) searching for all solutions that include that edge (resp.~vertex),
and (ii) searching for all solutions that do not include that edge
(resp.~vertex). This method is presented with further detail in
Section~\ref{subsection:BinaryPartitionMethod}.

\medskip

{\bf Contraction--deletion.} Similarly to the binary-partition
approach, this technique is characterized by recursively dividing the
search space in two. By taking an edge of the graph, the search space
is explored by: (i) contraction of the edge, i.e.~merging the
endpoints of the edge and their adjacency list, and (ii) deletion of
the edge. For further information the reader is referred
to~\cite{bollobas2000contraction}

\medskip

{\bf Gray codes.} According to this technique, the space of
solutions is encoded in such a way that consecutive solutions differ
by a constant number of modifications. Although not every pattern has
properties that allow such encoding, this technique leads to very
efficient algorithms. For further information see \cite{savage}.

\medskip

{\bf Reverse search} This is a general technique to explore the
space of solutions by reversing a local search algorithm. Considering
a problem with a unique maximum objective value, there exist local
search algorithms that reach the objective value using simple
operations. One such example is the Simplex algorithm. The idea of
reverse search is to start from this objective value and
\emph{reverse} the direction of the search. This approach implicitly
generates a tree of the search space that is traversed by the reverse
search algorithm. One of the properties of this tree is that it has
bounded height, a useful fact for proving the time complexity of the
algorithm. Avis and Fukuda introduced this idea in \cite{avis}.

\medskip

Although there is some literature on techniques for enumeration
problems \cite{uno1998new,takeaki2003two}, many more techniques and
``tricks'' have been introduced when attacking particular problems.
For a deep understanding of the topic, the reader is recommended to
review the work of David Avis, Komei Fukuda, Shin-ichi Nakano, Takeaki
Uno.

\subsection{Problems}

As a contribution for researchers starting to explore the topic of
listing patterns in graphs, we present a table with the most relevant
settings of problems and a list of state of the art references. For
the problems tackled in this thesis, a more detailed review is
presented on the introductory text of
Chapters~\ref{chapter:ListingKSubtrees},
\ref{chapter:ListingKSubgraphs} and \ref{chapter:cycles}.

\begin{center}
\begin{tabular}{|l|l|}
\hline
Cycles and Paths & See \cite{ferreira2013} and \cite{Johnson1975} \\ \hline
Spanning trees & See \cite{shioura} \\ \hline
Subgraphs & See \cite{avis}, \cite{koch2001enumerating} \\ \hline
Matchings & See \cite{uno1997algorithms}, \cite{propp1999enumeration}, \cite{fukuda1994finding} and \cite{uno2001fast} \\ \hline
Cut-sets & See \cite{tsukiyama1980algorithm}, \cite{arunkumar1979enumeration} \\ \hline
Independent sets & See \cite{johnson1988generating} \\ \hline
Induced paths, cycles, holes & See \cite{uno2003output} \\ \hline
Cliques, pseudo-cliques & See \cite{modani2008large}, \cite{makino2004new} and \cite{uno2007efficient}\\ \hline
Colorings & See \cite{matsui2007enumeration} and \cite{matsui1996enumeration}\\ \hline
\end{tabular}
\end{center}

\section{Overview of the techniques used}
\label{section:OverviewOfTheTechniquesUsed}

In this section, we present an overview of the approach we have applied
to problems of listing combinatorial patterns in graphs. Let us start
by describing the binary partition method and then introduce the
concept of certificate.

\subsection{Binary partition method}
\label{subsection:BinaryPartitionMethod}

Binary partition is a method for recursively subdividing the search
space of the problem. In the case of a graph $G=(V,E)$, we can take an
edge $e \in E$ (or a vertex $v \in V$) and list: (i) all the patterns
that include $e$, and (ii) all the patterns that do not include edge
$e$.

Formally, let $\mathcal{S}(G)$ denote the set of patterns in $G$.  For
each pattern $p \in \mathcal{S}(G)$, a \emph{prefix} $p'$ of~$p$ is a
substructure of pattern $p$, denoted by $p' \subseteq p$.  Let
$\mathcal{S}(p',G)$ be the set of patterns in $G$ that have $p'$ as a
substructure. By taking an edge $e \in E$, we can write the
following recurrence:

\begin{equation}
\label{eq:genbinpart}
\mathcal{S}\left(p',G\right) \quad = \quad \mathcal{S} \left( p' + \{e\},G\right) \quad \cup \quad \mathcal{S}\left(p',G - \{e\} \right)
\end{equation}

Noting that the left side is disjoint with the right side of the
recurrence, this can be a good starting point to enumerate the
patterns in $G$.

\medskip

In order to implement this idea efficiently, one key point is to
maintain $\mathcal{S}(p',G) \neq \emptyset$ as an invariant. Generally
speaking, maintaining the invariant on the left side of the recurrence
is easier than on the right side. Since we can take any edge $e \in E$
to partition the space, we can use the definition of the pattern to
select an edge that maintains the invariant. For example, if the
pattern is connected, we augment the prefix with an edge that is
connected to it. In order to maintain this invariant efficiently, we
introduce the concept of a certificate.

\subsection{Certificates}

When implementing an algorithm using recurrence~\ref{eq:genbinpart},
we maintain a certificate that guaranties that $\mathcal{S}(p',G) \neq
\emptyset$. An example of a certificate is a pattern $p \supseteq p'$.
In the recursive step, we select an edge $e$ that interacts as little
as possible with the certificate and thus facilitates its maintenance.
In the ideal case, a certificate $C$ is both the certificate of
$\mathcal{S}(p',G)$ and of $\mathcal{S}(p',G-\{e\})$. Although this is
not always possible, we are able to amortize the cost of modifying
the certificate.

During the recursion we are able to avoid copying the certificate and
maintain it instead. We use data structures that preserve previous
versions of itself when they are modified. These data structures,
usually called \emph{persistent}, allow us to be efficient in terms of
the space.

\subsection{Adaptive amortized analysis}

The problem of maintaining a certificate is very different than the
traditional view of dynamic data structures. When considering
fully-dynamic or partially-dynamic data structures, the cost of
operations is usually amortized after a sequence of \emph{arbitrary}
insertions and deletions. In the case of maintaining the
certificate throughout the recursion, the sequence of operations is not
arbitrary. As an example, when we remove an edge from the graph, the
edge is added back again on the callback of the recursion.
Furthermore, we have some control over the recursion as we can select
the edge $e$ in a way that facilitates the maintenance of the
certificate (and its analysis).

Another key idea is that a certificate can provide a lower bound on the
size of set $\mathcal{S}(p',G)$. We have designed certificates whose
time to maintain is proportional to the lower bound on the number of
patterns they provide.

\medskip

Summing up, the body of work presented in this thesis revolves around
shaping the recursion tree, designing and efficiently implementing the
certificate while using graph theoretical properties to amortize the
costs of it all.

\chapter{Listing k-subtrees}
\label{chapter:ListingKSubtrees}

Given an undirected connected graph $G=(V,E)$ with $n$ vertices and
$m$ edges, we consider the problem of listing all the $k$-subtrees in
$G$. Recall that we define a $k$-subtree $T$ as an edge subset $T
\subseteq E$ that is acyclic and connected, and contains $k$ vertices.
Refer to Section~\ref{subsection:kSubtrees} for applications of this
combinatorial pattern.  We denote by $s$ the number of $k$-subtrees in
$G$.  For example, there are $s=9$ $k$-subtrees in the graph of
Fig.~\ref{fig:ksubtreesexample}, where $k=3$. Originally published
in~\cite{Ferreira11}, we present the \emph{first optimal
output-sensitive} algorithm for listing all the $k$-subtrees in $O(s k)$
time, using $O(m)$ space.

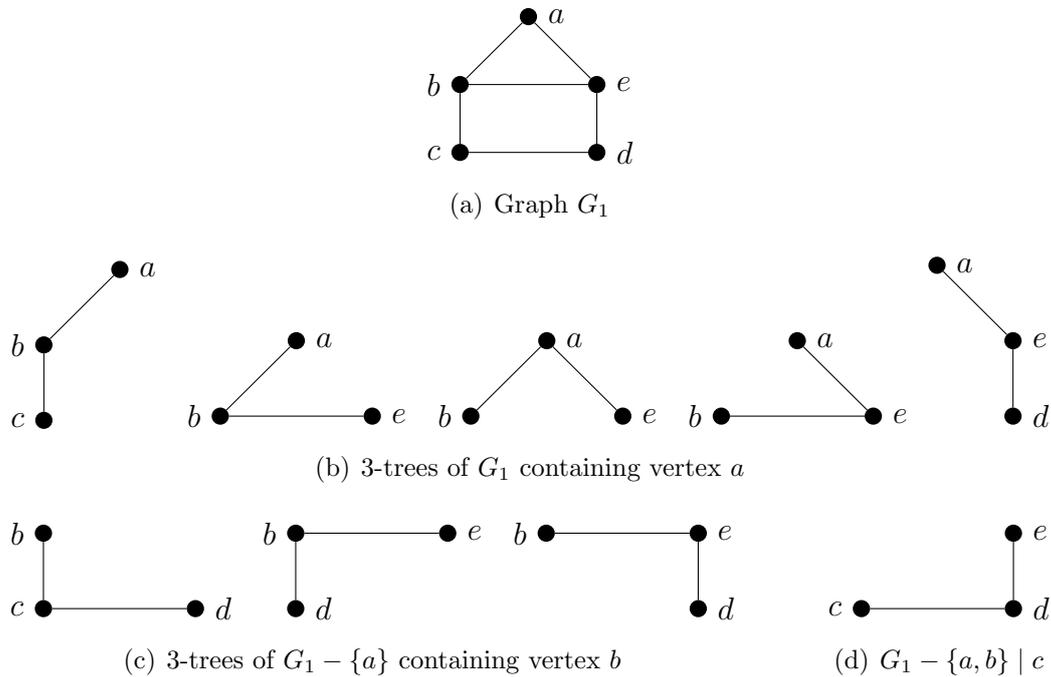
\begin{figure}[t]
\centering
\subfigure[Graph $G_1$]{
\begin{tikzpicture}[shorten >=1pt,->,scale=0.9]
  \tikzstyle{vertex}=[shape=circle,draw,thick,fill=black,minimum size=2pt,inner sep=2pt]
  \node[vertex,label=right:$a$] (A) at (0,0) {};
  \node[vertex,label=left:$b$] (B) at (-1,-1)   {};
  \node[vertex,label=left:$c$] (C) at (-1,-2)  {};
  \node[vertex,label=right:$d$] (D) at (1,-2)  {};
  \node[vertex,label=right:$e$] (E) at (1,-1)  {};
  \draw (A) -- (B) -- (C) -- (D) -- (E) -- (A) -- cycle;
  \draw (B) -- (E) -- cycle;
\end{tikzpicture}
} 

\subfigure[3-trees of $G_1$ containing vertex $a$]{
\begin{tikzpicture}[shorten >=1pt,->,scale=1.0]
  \tikzstyle{vertex}=[shape=circle,draw,thick,fill=black,minimum size=2pt,inner sep=2pt]
  \node[vertex,label=right:$a$] (A) at (0,0) {};
  \node[vertex,label=left:$b$] (B) at (-1,-1)   {};
  \node[vertex,label=left:$c$] (C) at (-1,-2)  {};
  \draw (A) -- (B) -- (C) -- cycle;
\end{tikzpicture}
\begin{tikzpicture}[shorten >=1pt,->,scale=1.0]
  \tikzstyle{vertex}=[shape=circle,draw,thick,fill=black,minimum size=2pt,inner sep=2pt]
  \node[vertex,label=right:$a$] (A) at (0,0) {};
  \node[vertex,label=left:$b$] (B) at (-1,-1)   {};
  \node[vertex,label=right:$e$] (E) at (1,-1)  {};
  \draw (A) -- (B) -- (E) -- cycle;
\end{tikzpicture}
\begin{tikzpicture}[shorten >=1pt,->,scale=1.0]
  \tikzstyle{vertex}=[shape=circle,draw,thick,fill=black,minimum size=2pt,inner sep=2pt]
  \node[vertex,label=right:$a$] (A) at (0,0) {};
  \node[vertex,label=left:$b$] (B) at (-1,-1)   {};
  \node[vertex,label=right:$e$] (E) at (1,-1)  {};
  \draw (A) -- (B) -- cycle;
  \draw (A) -- (E) -- cycle;
\end{tikzpicture}
\begin{tikzpicture}[shorten >=1pt,->,scale=1.0]
  \tikzstyle{vertex}=[shape=circle,draw,thick,fill=black,minimum size=2pt,inner sep=2pt]
  \node[vertex,label=right:$a$] (A) at (0,0) {};
  \node[vertex,label=left:$b$] (B) at (-1,-1)   {};
  \node[vertex,label=right:$e$] (E) at (1,-1)  {};
  \draw (A) -- (E) -- (B) -- cycle;
\end{tikzpicture}
\begin{tikzpicture}[shorten >=1pt,->,scale=1.0]
  \tikzstyle{vertex}=[shape=circle,draw,thick,fill=black,minimum size=2pt,inner sep=2pt]
  \node[vertex,label=right:$a$] (A) at (0,0) {};
  \node[vertex,label=right:$e$] (E) at (1,-1)  {};
  \node[vertex,label=right:$d$] (D) at (1,-2)   {};
  \draw (A) -- (E) -- (D) -- cycle;
\end{tikzpicture}
}

\subfigure[$3$-trees of $G_1-\{a\}$ containing vertex $b$]{
\begin{tikzpicture}[shorten >=1pt,->,scale=1.0]
  \tikzstyle{vertex}=[shape=circle,draw,thick,fill=black,minimum size=2pt,inner sep=2pt]
  \node[vertex,label=left:$b$] (B) at (-1,-1)   {};
  \node[vertex,label=left:$c$] (C) at (-1,-2)  {};
  \node[vertex,label=right:$d$] (D) at (1,-2)  {};
  \draw (B) -- (C) -- (D) -- cycle;
\end{tikzpicture}
\begin{tikzpicture}[shorten >=1pt,->,scale=1.0]
  \tikzstyle{vertex}=[shape=circle,draw,thick,fill=black,minimum size=2pt,inner sep=2pt]
  \node[vertex,label=left:$b$] (B) at (-1,-1)   {};
  \node[vertex,label=right:$d$] (C) at (-1,-2)  {};
  \node[vertex,label=right:$e$] (E) at (1,-1)  {};
  \draw (C) -- (B) -- (E) -- cycle;
\end{tikzpicture}
\begin{tikzpicture}[shorten >=1pt,->,scale=1.0]
  \tikzstyle{vertex}=[shape=circle,draw,thick,fill=black,minimum size=2pt,inner sep=2pt]
  \node[vertex,label=left:$b$] (B) at (-1,-1)   {};
  \node[vertex,label=right:$d$] (D) at (1,-2)  {};
  \node[vertex,label=right:$e$] (E) at (1,-1)  {};
  \draw (B) -- (E) -- (D) -- cycle;
\end{tikzpicture}
}
\hspace{1.0em}
\subfigure[$G_1-\{a,b\} \mid c$]{
\begin{tikzpicture}[shorten >=1pt,->,scale=1.0]
  \tikzstyle{vertex}=[shape=circle,draw,thick,fill=black,minimum size=2pt,inner sep=2pt]
  \node[vertex,label=left:$c$] (C) at (-1,-2)  {};
  \node[vertex,label=right:$d$] (D) at (1,-2)  {};
  \node[vertex,label=right:$e$] (E) at (1,-1)  {};
  \draw (C) -- (D) -- (E) -- cycle;
\end{tikzpicture}
}

\caption{Example graph $G_1$ and its 3-trees}
\label{fig:ksubtreesexample}

\end{figure}

\medskip

As a special case, this problem also models the classical problem of
listing the spanning trees in $G$, which has been largely investigated
(here $k=n$ and $s$ is the number of spanning trees in $G$). The first
algorithmic solutions appeared in the 60's~\cite{minty}, and the
combinatorial papers even much earlier~\cite{moon1970counting}. Read
and Tarjan presented an output-sensitive algorithm in $O(sm)$ time and
$O(m)$ space~\cite{Read75}. Gabow and Myers proposed the first
algorithm~\cite{gabow} which is optimal when the spanning trees are
explicitly listed. When the spanning trees are implicitly enumerated,
Kapoor and Ramesh~\cite{kapoor} showed that an elegant incremental
representation is possible by storing just the $O(1)$ information
needed to reconstruct a spanning tree from the previously enumerated
one, giving $O(m+s)$ time and $O(mn)$ space~\cite{kapoor}, later
reduced to $O(m)$ space by Shioura et al.~\cite{shioura}.  After we
introduced the problem in~\cite{Ferreira11}, a constant-time algorithm
for the enumeration of $k$-subtrees in the particular case of
trees~\cite{wasa2012} was introduced by Wasa, Uno and Arimura.  We are
not aware of any other non-trivial output-sensitive solution for the
problem of listing the $k$-subtrees in the general case.

\medskip

We present our solution starting from the binary partition method
(Section~\ref{subsection:BinaryPartitionMethod}). We divide the
problem of listing all the $k$-subtrees in two subproblems by choosing
an edge $e \in E$: we list the $k$-subtrees that contain~$e$ and those
that do not contain~$e$.  We proceed recursively on these subproblems
until there is just one $k$-subtree to be listed. This method induces
a binary recursion tree, and all the $k$-subtrees can be listed when
reaching the leaves of the recursion tree.

Although output sensitive, this simple method is not optimal. In fact,
a more careful analysis shows that it takes $O((s+1)(m+n))$ time. One
problem is that the adjacency lists of $G$ can be of length $O(n)$
each, but we cannot pay such a cost in each recursive call. Also, we
need a \emph{certificate} that should be easily maintained through the
recursive calls to guarantee \emph{a priori} that there will be at
least one $k$-subtree generated. By exploiting more refined structural
properties of the recursion tree, we present our algorithmic ideas
until an optimal output-sensitive listing is obtained, i.e.\mbox{}
$O(sk)$ time. Our presentation follows an incremental approach to
introduce each idea, so as to evaluate its impact in the complexity of
the corresponding algorithms.

\section{Preliminaries}
\label{section:kSubtreesPreliminaries}
Given a simple (without self-loops or parallel edges), undirected and
connected graph $G = (V,E)$, with $n=|V|$ and $m=|E|$, and an integer
$k \in [2,n]$, a \emph{$k$-subtree} $T$ is an acyclic connected
subgraph of $G$ with $k$ vertices.  We denote the total number of
$k$-subtrees in $G$ by $s$, where $sk \geq m \geq n-1$ since $G$ is
connected. Let us formulate the problem of listing $k$-subtrees.

\begin{problem}
  \label{problem:ksubtreelisting}
  Given an input graph $G$ and an integer $k$, list all the $k$-subtrees of $G$.
\end{problem}

We say that an algorithm that solves Problem~\ref{problem:ksubtreelisting} is
\emph{optimal} if it takes $O(sk)$ time, since the latter is
proportional to the time taken to explicitly list the output, namely,
the $k-1$ edges in each of the $s$ listed $k$-subtrees. We also say that
the algorithm has \emph{delay} $t(k)$ if it takes $O( t(k) )$ time to
list a $k$-subtree after having listed the previous one.

We adopt the standard representation of graphs using adjacency lists
$\adj(v)$ for each vertex $v \in V$. We maintain a counter for each $v
\in V$, denoted by $|\adj(v)|$, with the number of edges in the
adjacency list of $v$.  Additionally, as the graph is undirected,
$(u,v)$ and $(v,u)$ are equivalent.

Let $X \subseteq E$ be a \emph{connected edge} set. We denote by $V[X]
= \{ u \mid (u,v) \in X \}$ the set of its endpoints, and its
\emph{ordered vertex list} $\hat V(X)$ recursively as follows: $\hat
V( \{ (\cdot,u_0)\} ) = \langle u_0 \rangle$ and $\hat V(X+(u,v)) =
\hat V(X)+\langle v \rangle$ where $u \in V[X]$, $v \notin V[X]$, and
$+$ denotes list concatenation. We also use the shorthand $E[X] \equiv
\{ (u,v) \in E \mid u,v \in V[X] \}$ for the induced edges.  In
general, $G[X]=(V[X],E[X])$ denotes the subgraph of $G$ \emph{induced}
by $X$, which is equivalently defined as the subgraph of $G$ induced
by the vertices in $V[X]$.

The \emph{cutset} of $X$ is the set of edges $C(X)
\subseteq E$ such that $(u,v) \in C(X)$ if and only if $u \in
V[X]$ and $v \in V-V[X]$. Note that when $V[X]=V$, the cutset is empty.
Similarly, the  \emph{ordered cutlist}~$\hat C(X)$
contains the edges in $C(X)$ ordered by the rank of their endpoints
in $\hat V(X)$. If two edges have the same endpoint
vertex $v \in \hat V(S)$, we use the order in which they appear in
$\adj(v)$ to break the tie.

Throughout the chapter we represent an unordered $k'$-subtree $T=\langle
e_1, e_2, \ldots, e_{k'} \rangle$ with $k' \leq k$ as an
\emph{ordered}, connected and acyclic list of $k'$ edges, where we use
a \emph{dummy} edge $e_1=(\cdot , v_i)$ having a vertex
$v_i$ of $T$ as endpoint. The order is the one by which we discover the edges
$e_1, e_2, \ldots, e_k'$.  Nevertheless, we do not generate two
different orderings for the same $T$. 

\section{Basic approach: recursion tree}
\label{sec:initial-approach}

We begin by presenting a simple algorithm that solves
Problem~\ref{problem:ksubtreelisting} in $O(s k^3)$ time, while using
$O(mk)$ space.  Note that the algorithm is not optimal yet: we will
show in Sections~\ref{sec:improved}--\ref{sec:output-sensitive} how to
improve it to obtain an optimal solution with $O(m)$ space and delay
$t(k) = k^2$.

\subsection{Top level}
\label{sub:top-level}

We use the standard idea of fixing an ordering of the vertices in $V=
\langle v_{1},v_{2},\ldots, v_{n} \rangle$. For each $v_{i}\in V$, we
list the $k$-subtrees that include $v_{i}$ and do not include any
previous vertex $v_{j}\in V$ ($j<i$). After reporting the
corresponding $k$-subtrees, we remove $v_{i}$ and its incident edges
from our graph~$G$. We then repeat the process, as summarized in
Algorithm~\ref{alg:ListAllTrees}.
Here, $S$ denotes a $k'$-subtree with $k' \leq k$, and we use the dummy
edge $(\cdot, v_i)$ as a start-up point, so that the ordered vertex
list is $\hat V(S) = \langle v_i \rangle$.  Then, we find a $k$-subtree
by performing a DFS starting from $v_i$: when we meet the $k$th
vertex, we are sure that there exists at least one $k$-subtree for $v_i$
and execute the binary partition method with \listtrees;
otherwise, if there is no such $k$-subtree, we can skip~$v_i$ safely.  We
exploit some properties on the recursion tree and an efficient
implementation of the following operations on~$G$:

\begin{itemize}
\item \label{op:vdel} $\vdel(u)$ deletes a vertex $u \in V$ and all
  its incident edges.

\item \label{op:edel} $\edel(e)$ deletes an edge $e=(u,v) \in E$.
The inverse operation
is denoted by $\idel(e)$.
Note that $|\adj(v)|$ and $|\adj(u)|$ are updated.

\item \label{op:chooseedge} $\chooseedge(S)$, for a $k'$-subtree $S$ with
  $k' \leq k$, returns an edge $e \in C(S)$: 
  $e^-$ the vertex in $e$ that belongs to~$V[S]$ and by $e^+$ the one
  s.t. $e^+ \in V - V[S]$.

\item \label{op:kdfs} $\kdfs(S)$ returns the list of the 
  tree edges obtained by a \emph{truncated DFS}, where conceptually $S$ is
  treated as a \emph{single} (collapsed vertex) source whose adjacency
  list is the cutset $C(S)$. The DFS is truncated when it finds $k$
  tree edges (or less if there are not so many). The resulting list is
  a $k$-subtree (or smaller) that includes all the edges in $S$.
  Its purpose is to check if there exists a connected component of
  size at least $k$ that contains $S$.
\end{itemize}

\begin{lemma} 
  \label{lem:graph-ops}
  Given a graph $G$ and a $k'$-subtree $S$, we can implement the
  following operations: $\vdel(u)$ for a vertex $u$ in time
  proportional to $u$'s degree; $\edel(e)$ and $\idel(e)$ for an edge
  $e$ in $O(1)$ time; $\chooseedge(S)$ and $\kdfs(S)$ in $O(k^2)$ time.
\end{lemma}
\begin{proof}
  We represent $G$ using standard adjacency lists where, for each edge
  $(u,v)$, we keep references to its occurrence in the adjacency list
  of $u$ and $v$, respectively. We also update $|\adj(u)|$ and
  $|\adj(v)|$ in constant time. This immediately gives the complexity
  of the operations $\vdel(u)$, $\edel(e)$, and $\idel(e)$.

  As for $\chooseedge(S)$, recall that $S$ is a $k'$-subtree with $k' <
  k$.  Recall that there are at most ${k' \choose 2} = O(k^2)$ edges
  in between the vertices in $V[S]$. Hence, it takes $O(k^2)$ to
  discover an edge that belongs to the cutset $C(S)$: it is the first
  one found in the adjacency lists of $V[S]$ having an endpoint in
  $V-V[S]$. This can be done in $O(k^2)$ time using a standard trick:
  start from a binary array $B$ of $n$ elements set to~0. For each $x
  \in V[S]$, set $B[x] := 1$. At this point, scan the adjacency lists
  of the vertices in $V[S]$ and find a vertex $y$ within them with
  $B[y] = 0$. After that, clean it up setting $B[x] := 0$ for each $x
  \in V[S]$.

  Consider now $\kdfs(S)$. Starting from the $k'$-subtree $S$, we first
  need to generate the list $L$ of $k-k'$ vertices in $V-V[S]$ (or less
  if there are not so many) from the cutset $C(S)$. This is a simple
  variation of the method above, and it takes $O(k^2)$ time. At this
  point, we start a regular, truncated DFS-tree from a source
  (conceptually representing the collapsed $S$) whose adjacency list
  is $L$. During the DFS traversal, when we explore an edge $(v,u)$
  from the current vertex $v$, either $u$ has already been visited or
  $u$ has not yet been visited.  Recall that we stop when new $k-k'$
  distinct vertices are visited. Since there are at most ${k \choose
    2} = O(k^2)$ edges in between $k$ vertices, we traverse $O(k^2)$
  edges before obtaining the truncated DFS-tree of size $k$.
\end{proof}

\begin{algorithm}[t]
\begin{algorithmic}[1]
	\FOR{$i=1,2,\ldots,n-1$}
		\STATE $S:= \langle (\cdot,v_i) \rangle$
		\IF{$|\kdfs(S)|=k$}
			\STATE $\listtrees(S)$
		\ENDIF
		\STATE $\vdel(v_i)$
	\ENDFOR
\end{algorithmic}
\caption{$\mathtt{ListAllTrees}$( $G=(V,E)$, $k$ )\label{alg:ListAllTrees}}
\end{algorithm}

\begin{algorithm}[t]
\begin{algorithmic}[1]
	\IF{$|S|=k$}
        	\STATE $\treeoutput(S)$
        	\STATE $\mathtt{return}$
	\ENDIF
\STATE $e:=\chooseedge(S)$
\STATE $\listtrees(S+ \langle e \rangle)$
\STATE $\edel(e)$
\IF {$|\kdfs(S)|=k$}
	\STATE $\listtrees(S)$
\ENDIF
\STATE $\idel(e)$
\end{algorithmic}
\caption{ $\listtrees(S)$ \label{alg:ListTrees}}
\end{algorithm}

\subsection{Recursion tree and analysis}
\label{sub:analysis-suboptimal}

The recursive binary partition method in Algorithm~\ref{alg:ListTrees}
is quite simple, and takes a $k'$-subtree $S$ with $k' \leq k$ as
input. The purpose is that of listing all $k$-subtrees that include all
the edges in $S$ (excluding those with endpoints $v_1, v_2, \ldots,
v_{i-1}$).  The precondition is that we recursively explore $S$ if and
only if there is at least a $k$-subtree to be listed.  The corresponding
recursion tree has some interesting properties that we exploit during
the analysis of its complexity. The root of this binary tree is
associated with $S = \langle (\cdot, v_i) \rangle$. Let $S$ be the
$k'$-subtree associated with a node in the recursion tree. Then, left
branching occurs by taking an edge $e \in C(S)$ using \chooseedge, so
that the \emph{left child} is $S + \langle e \rangle$. Right branching
occurs when $e$ is deleted using \edel, and the \emph{right child} is
still $S$ but on the reduced graph $G := (V, E - \{e\})$. Returning
from recursion, restore $G$ using $\idel(e)$.  Note that we do
\emph{not} generate different permutations of the same $k'$-subtree's
edges as we either take an edge $e$ as part of $S$ or remove it from
the graph by the binary partition method. 

\begin{lemma}
  \label{lem:correctness}
  Algorithm \ref{alg:ListTrees} lists each $k$-subtree containing vertex
  $v_i$ and no vertex $v_j$ with $j<i$, once and only once.
\end{lemma}
\begin{proof}
  \listtrees\ outputs all the wanted $k$-subtrees according to a simple
  rule: first list all the $k$-subtrees that include edge $e$ and then
  those not including $e$: an edge $e$ must exist because of the
  precondition, and we can choose \emph{any} edge $e$ incident to the
  partial solution $S$. As \chooseedge\ returns an edge $e$ from the
  cutset $C(S)$, this edge is incident in $S$ and does not introduce a
  cycle. Note that if $\kdfs(S)$ has size $k$, then
  there is a connected component of size $k$. Hence, there must be a
  $k$-subtree to be listed, as the spanning tree of the component is a valid
  $k$-subtree. Additionally, if there is a $k$-subtree to be listed, a
  connected component of size $k$ exists.  As a result, we
  conceptually partition all the $k$-subtrees in two disjoint sets: the
  $k$-subtrees that include $S+\langle e \rangle$, and the $k$-subtrees that
  include $S$ and do not include $e$. We stop when $|S|=k$ and we
  start the partial solution $S$ with a dummy edge that connects to
  $v_i$, ensuring that all trees of size $k$ connected to $v_i$ are
  listed.  Since all the vertices $v_j$ with $j<i$ are removed from
  $G$, $k$-subtrees that include $v_j$ are not listed twice. Therefore,
  each tree is listed at most one time.
\end{proof}

\begin{figure}[t]

\centering

\begin{tikzpicture}[style=thick,scale=0.8,rotate=-90] 
\tikzstyle{vertex}=[circle,inner sep=2pt,draw] 
\tikzstyle{level 1}=[sibling distance=60mm]
\tikzstyle{level 2}=[sibling distance=20mm]
\tikzstyle{level 3}=[sibling distance=15mm]
\node[vertex]{\null}[grow'=right] 
child{
	node[vertex]{\null}
	child{node [] {\null} edge from parent[draw=none]} 
	child{
		node[vertex]{}
		child{
			node[vertex]{}
			child{node [] {\null} edge from parent[draw=none]} 
			child{
				node[vertex]{$T_5$}
				node{\null} edge from parent node[left] {$ed$}
			} 
			node{\null} edge from parent node[right] {$\lnot be$}
		} 
		child{
			node[vertex]{$T_4$}
			node{\null} edge from parent node[left] {$be$}
		} 
		node{\null} edge from parent node[left] {$ae$}
	} 
	node{\null} edge from parent node[right] {$\lnot ab$}
} 
child{
	node[vertex]{\null}
	child{
		node[vertex]{\null}
		child{
			node[vertex]{\null}
			child{node [] {\null} edge from parent[draw=none]} 
			child{
				node[vertex]{$T_3$}
				node{\null} edge from parent node[left] {$ae$}
			} 
			node{\null} edge from parent node[right] {$\lnot be$}
		} 
		child{
			node[vertex]{$T_2$}
			node{\null} edge from parent node[left] {$be$}
		} 
		node{\null} edge from parent node[right] {$\lnot bc$}
	} 
	child{
		node[vertex]{$T_1$}
		node{\null}  edge from parent node[left] {$bc$}
	}
	node{\null}  edge from parent node[left] {$ab$}
}; 
\end{tikzpicture}

\caption{Recursion tree of \listtrees{} for graph $G_1$ in
Fig.~\ref{fig:ksubtreeexample} and $v_i=a$}
\label{fig:recursion-tree}
\end{figure}

A closer look at the recursion tree (e.g.\mbox{}
Figure~\ref{fig:recursion-tree}), reveals that it is
\emph{$k$-left-bounded}: namely, each root-to-leaf path has exactly
$k-1$ \emph{left branches}. Since there is a one-to-one correspondence
between the leaves and the $k$-subtrees, we are guaranteed that leftward
branching occurs less than $k$ times to output a $k$-subtree.

What if we consider rightward branching?  Note that the height of the
tree is less than $m$, so we might have to branch rightward $O(m)$
times in the worst case. Fortunately, we can prove in
Lemma~\ref{lem:kleftbounded} that for each internal node $S$ of the
recursion tree that has a right child, $S$ has always its left child
(which leads to one $k$-subtree). This is subtle but very
useful in our analysis in the rest of the chapter.

\begin{lemma}
  \label{lem:kleftbounded}
  At each node $S$ of the recursion tree, if there exists a $k$-subtree
  (descending from $S$'s right child) that does not include edge $e$,
  then there is a $k$-subtree (descending from $S$'s left child) that
  includes~$e$.
\end{lemma}
\begin{proof}
  Consider a $k$-subtree $T$ that does not include $e=(u,v)$, which was
  opted out at a certain node~$S$ during the recursion. Either $e$ is
  only incident to one vertex in $T$, so there is at least one
  $k$-subtree $T'$ that includes $e$ and does not include an edge  of $T$.
  Or, when $e =(u,v)$ is incident to two different nodes $u,v \in
  V[T]$, there is a valid $k$-subtree $T'$ that includes $e$ and does not
  include one edge on the path that connects $u$ and $v$ using edges
  from $T$. Note that both $T$ and $T'$ are ``rooted'' at $v_i$ and
  are found in two descending leaves from $S$.
\end{proof}

Note that the symmetric situation for Lemma~\ref{lem:kleftbounded}
does not necessarily hold. We can find nodes having just the left
child: for these nodes, the chosen edge cannot be removed since this
gives rise to a connected component of size smaller than $k$. We can
now state how many nodes there are in the recursion tree.

\begin{corollary}
  \label{cor:right-branch}
  Let $s_i$ be the number of $k$-subtrees reported by \listtrees. Then,
  its recursion tree is binary and contains $s_i$ leaves and at most
  $s_i \, k$ internal nodes. Among the internal nodes, there are
  $s_i-1$ of them having two children.
\end{corollary}
\begin{proof}
  The number $s_i$ of leaves derives from the one-to-one
  correspondence with the $k$-subtrees found by \listtrees.  To give an
  upper bound on the number of internal nodes, consider a generic node
  $S$ and apply Lemma~\ref{lem:kleftbounded} to it.  If $S$ has a
  single child, it is a left child that leads to one or more $k$-subtrees
  in its descending leaves. So, we can charge one token (corresponding
  to~$S$) to the leftmost of these leaves. Hence, the total number of
  tokens over all the $s_i$ leaves is at most $s_i \, (k-1)$ since the
  recursion tree is $k$-left-bounded. The other option is that $S$ has
  two children: in this case, the number of these kind of nodes cannot
  exceed the number $s_i$ of leaves. Summing up, we have a total of
  $s_i \, k$ internal nodes in the recursion tree. Consider the
  compacted recursion tree, where each maximal path of unary nodes
  (i.e.\mbox{} internal nodes having one child) is replaced by a
  single edge. We obtain a binary tree with $s_i$ leaves and all the other
  nodes having two children: we fall within a classical situation, for
  which the number of nodes with two children is one less than the
  number of leaves, hence, $s_i-1$.
\end{proof}

\begin{lemma}
  \label{lem:listtrees}
  Algorithm~\ref{alg:ListTrees} takes $O(s_i \, k^3)$ time and $O(mk)$
  space, where $s_i$ is the number of $k$-subtrees reported by
  \listtrees.
\end{lemma}
\begin{proof}
  Each call to \listtrees\ invokes operations $\edel$, $\idel$,
  $\chooseedge$, and $\kdfs$ once. By Lemma~\ref{lem:graph-ops}, the
  execution time of the call is therefore $O(k^2)$.  Since there are
  $O(s_i \, k)$ calls by Corollary~\ref{cor:right-branch}, the total
  running time of Algorithm~\ref{alg:ListTrees} is $O(s_i \,
  k^3)$. The total space of $O(mk)$ derives from the fact that the
  height of the recursion tree is at most $m$. On each node in the
  recursion path we keep a copy of $S$ and $D$, taking $O(k)$ space.
  As we modify and restore the graph incrementally, this totals
  $O(mk)$ space.
\end{proof}

\begin{theorem}
  \label{the:main}
  Algorithm~\ref{alg:ListAllTrees} can solve Problem~\ref{problem:ksubtreelisting}
  in $O(n k^2 + s k^3) = O(sk^3)$ time
  and $O(mk)$ space.
\end{theorem}
\begin{proof}
  The correctness of Algorithm~\ref{alg:ListAllTrees} easily derives
  from Lemma~\ref{lem:correctness}, so it outputs each $k$-subtree once
  and only once. Its cost is upper bounded by the sum (over all $v_i
  \in V$) of the costs of $\vdel(v_i)$ and $\kdfs(S)$ plus
  the cost of Algorithm~\ref{alg:ListTrees} (see
  Lemma~\ref{lem:listtrees}). The costs for $\vdel(v_i)$'s sum to
  $O(m)$, while those of $\kdfs(S)$'s sum to $O((n-k)
  k^2)$. Observing that $\sum_{i=1}^n s_i = s$, we obtain that the
  cumulative cost for Algorithm~\ref{alg:ListTrees} is $O(s \,
  k^3)$. Hence, the total running time is $O(m+(n-k)k^2+sk^3)$, which is
  $O(m+sk^3)$ since it can be proved by a simple induction that $s
  \geq n-k+1$ in a connected graph (adding a vertex increases the
  number of $k$ trees by at least one). Space usage of $O(mk)$ derives
  from Lemma~\ref{lem:listtrees}. As for the delay $t(k)$, we observe
  that for any two consecutive leaves in the preorder of the recursion
  tree, their distance (traversed nodes in the recursion tree) never
  exceeds $2k$. Since we need $O(k^2)$ time per node in the recursion
  tree, we have $t(k) = O(k^3)$.
\end{proof}

\section{Improved approach: certificates}
\label{sec:improved}


A way to improve the running time of \listtrees\ to $O(s_i k^2)$ is
indirectly suggested by Corollary~\ref{cor:right-branch}. Since there
are $O(s_i)$ binary nodes and $O(s_i k)$ unary nodes in the recursion
tree, we can pay $O(k^2)$ time for binary nodes and $O(1)$ for unary
nodes (i.e.\mbox{} reduce the cost of \chooseedge\ and \kdfs\ to
$O(1)$ time when we are in a \emph{unary} node). This way, the total
running time is $O(sk^2)$.

The idea is to maintain a certificate that can tell us if we are
in a unary node in $O(1)$ time and that can be updated in $O(1)$ time in
such a case, or can be completely rebuilt in $O(k^2)$ time otherwise
(i.e.\mbox{} for binary nodes).  This will guarantee a total cost
of $O(s_i k^2)$ time for \listtrees, and lay out the path to the
wanted optimal output-sensitive solution of Section~\ref{sec:output-sensitive}.
 
\subsection{Introducing certificates}
\label{sub:certificates}

We impose an ``unequivocal behavior'' to $\kdfs(S)$, obtaining a
variation denoted $\mkdfs(S)$ and called \emph{multi-source truncated
  DFS}.  During its execution, \mkdfs\ takes the order of the edges in
$S$ into account (whereas an order is not strictly necessary in
$\kdfs$). Specifically, given a $k'$-subtree $S = \langle e_1, e_2,
\ldots, e_{k'} \rangle$, the returned $k$-subtree $D = \mkdfs(S)$
contains $S$, which is conceptually treated as a collapsed vertex: the
main difference is that $S$'s ``adjacency list'' is now the
\emph{ordered cutlist} $\hat C(S)$, rather than $C(S)$ employed for
\kdfs.

Equivalently, since $\hat C(S)$ is induced from $C(S)$ by using the
ordering in $\hat V(S)$, we can see $\mkdfs(S)$ as the execution of
multiple standard DFSes from the vertices in $\hat V(S)$, in that
order. Also, all the vertices in $V[S]$ are conceptually marked as
visited at the beginning of \mkdfs, so $u_j$ is never part of the DFS
tree starting from $u_i$ for any two distinct $u_i, u_j \in
V[S]$. Hence the adopted terminology of multi-source.  Clearly, $\mkdfs(S)$
is a feasible solution to $\kdfs(S)$ while the vice versa is not true.

We use the notation $S \sqsubseteq D$ to indicate that $D=\mkdfs(S)$,
and so $D$ is a \emph{certificate} for $S$: it guarantees that node
$S$ in the recursion tree has at least one descending leaf whose
corresponding $k$-subtree has not been listed so far. Since the behavior
of \mkdfs\ is non-ambiguous, relation $\sqsubseteq$ is well
defined. We preserve the following invariant on \listtrees, which now
has two arguments.

\begin{invariant}
  For each call to $\listtrees(S,D)$, we have $S \sqsubseteq D$. 
\end{invariant}

Before showing how to keep the invariant, we detail how to represent
the certificate $D$ in a way that it can be efficiently updated. We
maintain it as a partition $D = S \cup L \cup F$, where $S$ is the
given list of edges, whose endpoints are kept in order as $\hat V(S) =
\langle u_1, u_2, \ldots, u_{k'} \rangle$. Moreover, $L = D \cap C(S)$
are the tree edges of $D$ in the cutset $C(S)$, and $F$ is the forest
storing the edges of $D$ whose both endpoints are in $V[D]-V[S]$.

\begin{enumerate}
\renewcommand{\theenumi}{(\roman{enumi})}
\renewcommand{\labelenumi}{\theenumi}
\item \label{item:S} We store the $k'$-subtree $S$ as a sorted
  doubly-linked list of $k'$ edges $\langle e_1, e_2, \ldots, e_{k'}
  \rangle$, where $e_1:=(\cdot,v_i)$. We also keep the sorted
  doubly-linked list of vertices $\hat V(S) = \langle u_1, u_2,
  \ldots, u_{k'} \rangle$ associated with $S$, where $u_1 := v_i$. For
  $1 \leq j \leq k'$, we keep the number of tree edges in the cutset
  that are incident to $u_j$, namely $\eta[u_j] = | \{ (u_j,x) \in L \}|$.
\item \label{item:L} We keep $L = D \cap C(S)$ as an ordered
  doubly-linked list of edges in $\hat C(S)$'s order: it can be easily
  obtained by maintaining the parent edge connecting a root in $F$ to
  its parent in $\hat V(S)$.
\item \label{item:F} We store the forest $F$ as a sorted doubly-linked
  list of the roots of the trees in $F$. The order of this list is
  that induced by $\hat C(S)$: a root $r$ precedes a root~$t$ if the
  (unique) edge in $L$ incident to~$r$ appears before the (unique) edge of
  $L$ incident to~$t$.  For each node $x$ of a tree $T \in F$, we also
  keep its number $\fdeg(x)$ of children in $T$, and its predecessor
  and successor sibling in $T$.
\item \label{item:isunary} We maintain a flag \isunary\ that is true
  if and only if
  $|\adj(u_i)| = \eta[u_i] + \sigma(u_i) \mbox{\ for\ all\ } 1 \leq i \leq k'$, 
  where $\sigma(u_i) = |\{ (u_i,u_j) \in E \mid i \neq j \}|$ is the number of
  internal edges, namely, having both endpoints in $V[S]$.
\end{enumerate}

Throughout the chapter, we identify $D$ with both \emph{(1)}~the set of
$k$ edges forming it as a $k$-subtree and \emph{(2)}~its representation
above as a certificate. We also support the following operations on
$D$, under the requirement that \isunary\ is true (i.e.\mbox{} all the
edges in the cutset $C(S)$ are tree edges), otherwise they are
undefined:

\begin{itemize}
\item $\treecut(D)$ returns the last edge in $L$.
\item $\promote(r,D)$, where root $r$ is the last in the doubly-linked
  list for $F$: remove~$r$ from $F$ and replace $r$ with its children
  $r_1, r_2, \ldots, r_c$ (if any) in that list, so they become the
  new roots (and so $L$ is updated). 
\end{itemize}

\begin{lemma}
  \label{lem:certificate-timespace}
  The representation of certificate $D = S \cup L \cup F$ requires
  $O(|D|) = O(k)$ memory words, and $\mkdfs(S)$ can build $D$ in
  $O(k^2)$ time.  Moreover, each of the operations \treecut\ and
  \promote\ can be supported in $O(1)$ time.
\end{lemma}
\begin{proof}
  Recall that $D = S \cup L \cup F$ and that $|S|+ |L| + |F| = k$ when
  they are seen as sets of edges. Hence, $D$ requires $O(k)$ memory
  words, since the representation of $S$ in the certificate takes
  $O(|S|)$ space, $L$ takes $O(|L|)$ space, $F$ takes $O(|F|)$ space,
  and \isunary\ takes $O(1)$ space.

  Building the representation of $D$ takes $O(k^2)$ using
  $\mkdfs$. Note that the algorithm behind $\mkdfs$ is very similar to
  $\kdfs$ (see the proof of Lemma~\ref{op:kdfs}) except that the order
  in which the adjacency lists for the vertices in $V[S]$ are explored
  is that given by $\hat V(S)$. After that the edges in $D$ are found
  in $O(k^2)$ time, it takes no more than $O(k^2)$ time to build the lists
  in points~\ref{item:S}--\ref{item:F} and check the condition in
  point~\ref{item:isunary} of Section~\ref{sub:certificates}.

  Operation \treecut\ is simply implemented in constant time by
  returning the last edge in the list for $L$ (point~\ref{item:L} of
  Section~\ref{sub:certificates}).  

  As for $\promote(r,D)$, the (unique) incident edge $(x,r)$ that
  belongs to $L$ is removed from $L$ and added to $S$ (updating the
  information in point~\ref{item:S} of
  Section~\ref{sub:certificates}), and edges $(r,r_1)$,~\dots,
  $(r,r_c)$ are appended at the end of the list for $L$ to preserve
  the cutlist order (updating the information in~\ref{item:L}). This
  is easily done using the sibling list of $r$'s children: only a
  constant number of elements need to be modified.  Since we know that
  \isunary\ is true before executing \promote, we just need to check
  if appending $r$ to $\hat V(S)$ adds non-tree edges to the cutlist
  $\hat C(S)$. Recalling that there are $\fdeg(r) + 1$ tree edges
  incident to~$r$, this is only the case when $|\adj(r)| > \fdeg(r) +
  1$, which can be easily checked in $O(1)$ time: if so, we set
  \isunary\ to false, otherwise we leave \isunary\ true. Finally, we
  correctly set $\eta[r] := \fdeg(r)$, and decrease $\eta[x]$ by one,
  since \isunary\ was true before executing \promote.
\end{proof}

\subsection{Maintaining the invariant}
\label{sub:new-listtrees}

We now define \chooseedge\ in a more refined way to facilitate the
task of maintaining the invariant $S \sqsubseteq D$ introduced in
Section~\ref{sub:certificates}. As an intuition, \chooseedge\ selects
an edge $e=(e^-,e^+)$ from the \emph{cutlist}~$\hat C(S)$ that
interferes as least as possible with the certificate~$D$. Recalling
that $e^- \in V[S]$ and $e^+ \in V - V[S]$ by definition of cutlist,
we consider the following case analysis:

\usetikzlibrary{snakes}
\tikzstyle{vertex}=[circle,fill=black,minimum size=5pt,inner sep=0pt]
\tikzstyle{edge} = [draw,thin,-]
\tikzstyle{triangle} = [triangle]

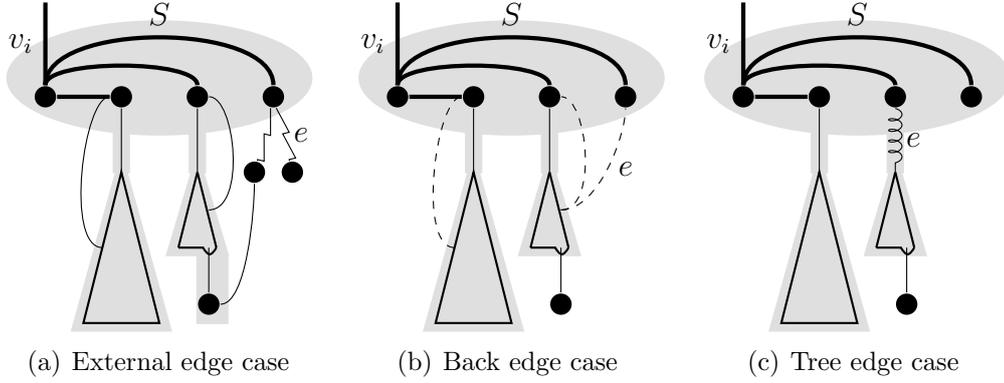
\begin{figure}[t]

\centering

\subfigure[External edge case]{
\label{fig:ksubtreeexternal}

\begin{tikzpicture}[style=thick,scale=0.5] 

\draw[color=gray!25,fill=gray!25] (-3,-1.5) ellipse (4 and 1.5);
\draw[color=gray!25,fill=gray!25] (-4,-3.5) -- (-5.3,-8.2) -- (-2.7,-8.2) -- cycle;
\draw[color=gray!25,fill=gray!25] (-2,-3.5) -- (-2.8,-6.2) -- (-1.2,-6.2) -- cycle;
\draw[color=gray!25,fill=gray!25] (-4.2,-2) -- (-4.2,-4) -- (-3.8,-4) -- (-3.8,-2) -- cycle;
\draw[color=gray!25,fill=gray!25] (-2.2,-2) -- (-2.2,-4) -- (-1.8,-4) -- (-1.8,-2) -- cycle;
\draw[color=gray!25,fill=gray!25] (-2.0,-6.2) -- (-2.0,-8) -- (-1.21,-8) -- (-1.21,-6.2) -- cycle;

\node[vertex] (a) at (-6,-2) {$a$};
\node[vertex] (b) at (0,-2) {$a$};
\node[vertex] (c) at (-4,-2) {$a$};
\node[vertex] (g) at (-2,-2) {$a$};

\path[edge,ultra thick] (-6,0.5) --  node [left] {$v_i$} (a);

\path[edge,ultra thick] (a) .. controls +(up:2cm) and +(up:2cm) .. node [auto] {$S$} (b);
\path[edge,ultra thick] (a) -- (c);
\path[edge,ultra thick] (a) .. controls +(up:1cm) and +(up:1cm) .. (g);

\draw (-4,-4) -- (-5,-8) -- (-3,-8) -- cycle;

\draw (-2,-4) -- (-2.5,-6);
\draw[snake=coil,segment aspect=0.0] (-1.5,-6) -- (-2.5,-6);
\draw (-1.5,-6) -- (-2,-4);

\node[vertex] (d) at (-0.5,-4) {$a$};
\node[vertex] (e) at (0.5,-4) {$a$};
\node[vertex] (f) at (-1.7,-7.5) {$a$};

\path[edge,snake=saw] (b) -- (d);
\path[edge,snake=saw] (b) -- node [right] {$e$} (e);
\path[edge] (d) .. controls +(down:1cm) and +(right:1cm) .. (f);
\path[edge] (c) -- (-4,-4);
\path[edge] (g) -- (-2,-4);

\path[edge] (-1.7,-6) -- (f);
\path[edge] (-1.7,-5) .. controls +(right:1cm) and +(right:1cm) .. (g);
\path[edge] (-4.5,-6) .. controls +(left:1cm) and +(left:1cm) .. (c);
\end{tikzpicture} 
}
%
\subfigure[Back edge case]{
\label{fig:ksubtreebackedge}

\begin{tikzpicture}[style=thick,scale=0.5] 

\draw[color=gray!25,fill=gray!25] (-3,-1.5) ellipse (4 and 1.5);
\draw[color=gray!25,fill=gray!25] (-4,-3.5) -- (-5.3,-8.2) -- (-2.7,-8.2) -- cycle;
\draw[color=gray!25,fill=gray!25] (-2,-3.5) -- (-2.8,-6.2) -- (-1.2,-6.2) -- cycle;
\draw[color=gray!25,fill=gray!25] (-4.2,-2) -- (-4.2,-4) -- (-3.8,-4) -- (-3.8,-2) -- cycle;
\draw[color=gray!25,fill=gray!25] (-2.2,-2) -- (-2.2,-4) -- (-1.8,-4) -- (-1.8,-2) -- cycle;

\node[vertex] (a) at (-6,-2) {$a$};
\node[vertex] (b) at (0,-2) {$a$};
\node[vertex] (c) at (-4,-2) {$a$};
\node[vertex] (g) at (-2,-2) {$a$};

\path[edge,ultra thick] (-6,0.5) --  node [left] {$v_i$} (a);

\path[edge,ultra thick] (a) .. controls +(up:2cm) and +(up:2cm) .. node [auto] {$S$} (b);
\path[edge,ultra thick] (a) -- (c);
\path[edge,ultra thick] (a) .. controls +(up:1cm) and +(up:1cm) .. (g);

\draw (-4,-4) -- (-5,-8) -- (-3,-8) -- cycle;

\draw (-2,-4) -- (-2.5,-6);
\draw[snake=coil,segment aspect=0.0] (-1.5,-6) -- (-2.5,-6);
\draw (-1.5,-6) -- (-2,-4);

\node[vertex] (f) at (-1.7,-7.5) {$a$};

\path[edge] (c) -- (-4,-4);
\path[edge] (g) -- (-2,-4);

\path[edge] (-1.7,-6) -- (f);
\path[edge,dashed] (-1.7,-5) .. controls +(right:1cm) and +(right:1cm) .. (g);
\path[edge,dashed] (-4.5,-6) .. controls +(left:1cm) and +(left:1cm) .. (c);
\path[edge,dashed] (-1.7,-5) .. controls +(right:1cm) and +(down:1cm) .. node [right] {$e$} (b);

\end{tikzpicture} 
}
\subfigure[Tree edge case]{
\label{fig:treeedge}

\begin{tikzpicture}[style=thick,scale=0.5] 

\draw[color=gray!25,fill=gray!25] (-3,-1.5) ellipse (4 and 1.5);
\draw[color=gray!25,fill=gray!25] (-4,-3.5) -- (-5.3,-8.2) -- (-2.7,-8.2) -- cycle;
\draw[color=gray!25,fill=gray!25] (-2,-3.5) -- (-2.8,-6.2) -- (-1.2,-6.2) -- cycle;
\draw[color=gray!25,fill=gray!25] (-4.2,-2) -- (-4.2,-4) -- (-3.8,-4) -- (-3.8,-2) -- cycle;
\draw[color=gray!25,fill=gray!25] (-2.2,-2) -- (-2.2,-4) -- (-1.8,-4) -- (-1.8,-2) -- cycle;

\node[vertex] (a) at (-6,-2) {$a$};
\node[vertex] (b) at (0,-2) {$a$};
\node[vertex] (c) at (-4,-2) {$a$};
\node[vertex] (g) at (-2,-2) {$a$};

\path[edge,ultra thick] (-6,0.5) --  node [left] {$v_i$} (a);

\path[edge,ultra thick] (a) .. controls +(up:2cm) and +(up:2cm) .. node [auto] {$S$} (b);
\path[edge,ultra thick] (a) -- (c);
\path[edge,ultra thick] (a) .. controls +(up:1cm) and +(up:1cm) .. (g);

\draw (-4,-4) -- (-5,-8) -- (-3,-8) -- cycle;

\draw (-2,-4) -- (-2.5,-6);
\draw[snake=coil,segment aspect=0.0] (-1.5,-6) -- (-2.5,-6);
\draw (-1.5,-6) -- (-2,-4);

\node[vertex] (f) at (-1.7,-7.5) {$a$};

\path[edge] (c) -- (-4,-4);
\path[edge,snake=coil,segment length=4pt] (g) -- node [right] {$e$} (-2,-4);

\path[edge] (-1.7,-6) -- (f);
\end{tikzpicture} 
}

\caption{Choosing edge $e \in C(S)$. The certificate $D$ is shadowed.}

\label{fig:chooseedge}
\end{figure}

\begin{enumerate}
\renewcommand{\theenumi}{(\alph{enumi})}
\renewcommand{\labelenumi}{\theenumi}
\item \label{item:external} [\textsf{external edge}] Check if there
  exists an edge $e \in \hat C(S)$ such that $e \not \in D$ and $e^+
  \not \in V[D]$. If so, return $e$, shown as a saw in
  Figure~\ref{fig:ksubtreeexternal}.
\item \label{item:back} [\textsf{back edge}] Otherwise, check if there
  exists an edge $e \in \hat C(S)$ such that $e \not \in D$ and $e^+
  \in V[D]$. If so, return $e$, shown dashed in Figure~\ref{fig:ksubtreebackedge}.
\item \label{item:tree} [\textsf{tree edge}] As a last resort, every
  $e \in \hat C(S)$ must be also $e \in D$ (i.e.\mbox{} all edges in the
  cutlist are tree edges). Return $e := \treecut(D)$, the last edge
  from $\hat C(S)$, shown as a coil in Fig.~\ref{fig:treeedge}.
\end{enumerate}

\begin{lemma}
  \label{lemma:choose_cases}
  For a given $k'$-subtree $S$, consider its corresponding node in the
  recursion tree. Then, this node is
  \emph{binary} when \chooseedge\ returns an external or back edge
  (cases~\ref{item:external}--\ref{item:back}) and is \emph{unary}
  when \chooseedge\ returns a tree edge (case~\ref{item:tree}).
\end{lemma}
\begin{proof}
  Consider a node $S$ in the recursion tree and the corresponding
  certificate $D$. For a given edge $e$ returned by $\chooseedge(S,D)$
  note that: if $e$ is an external or back edge
  (cases~\ref{item:external}--\ref{item:back}), $e$ does not belong to
  the $k$-subtree $D$ and therefore there exists a $k$-subtree that does not
  include $e$. By Lemma~\ref{lem:kleftbounded}, a $k$-subtree that
  includes $e$ must exist and hence the node is binary.  We are left
  with case~\ref{item:tree}, where $e$ is an edge of $D$ that belongs
  to the cutlist $\hat C(S)$.  Recall that, by the way \chooseedge\
  proceeds, all edges in the cutlist $\hat C(S)$ belong to $D$ (see
  Figure \ref{fig:treeedge}).  There are no further $k$-subtrees that do
  not include edge $e$ as $e$ is the last edge from $\hat C(S)$ in the
  order of the truncated DFS tree, and so the node in the recursion
  tree is unary. The existence of a $k$-subtree that does not include $e$
  would imply that $D$ is not the valid $\mkdfs(S)$: at least $k$
  vertices would be reachable using the previous branches of $\hat
  C(S)$ which is a contradiction with the fact that we traverse
  vertices in the DFS order of $\hat V(S)$.
\end{proof}

\begin{algorithm}[t]
\begin{algorithmic}[1]
	\FOR{$v_i \in V$}
		\STATE $S:=\, \langle (\cdot,v_i) \rangle$
		\STATE $D:=\mkdfs(S)$
		\IF{$|D|<k$}
			\FOR {$u \in V[D]$}
			\STATE $\vdel(u)$
			\ENDFOR
		\ELSE
			\STATE $\listtrees(S,D)$ 
                	\STATE $\vdel(v_i)$
		\ENDIF
	\ENDFOR
\end{algorithmic}
\caption{$\mathtt{ListAllTrees}$( $G=(V,E)$, $k$ ) \label{alg:ListAllTrees2}}
\end{algorithm}

\begin{algorithm}[t]
\begin{algorithmic}[1]
	\IF{$|S|=k$}
		\STATE $\treeoutput(S)$
		\STATE $\mathtt{return}$
	\ENDIF
\STATE $e:=\chooseedge(S,D)$
\IF {\isunary}
        \STATE $D' := \promote(e^+,D)$
        \STATE $\listtrees(S+ \langle e \rangle,D')$
\ELSE
        \STATE $D' := \mkdfs(S+ \langle e \rangle)$
        \STATE $\listtrees(S+ \langle e \rangle,D')$
        \STATE $\edel(e)$
        \STATE $D'' := \mkdfs(S)$
        \STATE $\listtrees(S,D'')$
        \STATE $\idel(e)$
\ENDIF
\end{algorithmic}
\caption{ $\listtrees(S,D)$ \hfill \textsl{\{Invariant: $S \sqsubseteq D$\}} \label{alg:ListTrees2}}
\end{algorithm} 

We now present the new listing approach in
Algorithm~\ref{alg:ListAllTrees2}. If the connected component of
vertex $v_i$ in the residual graph is smaller than~$k$, we delete its
vertices since they cannot provide $k$-subtrees, and so we skip them in
this way.  Otherwise, we launch the new version of \listtrees, shown
in Algorithm~\ref{alg:ListTrees2}. In comparison with the previous
version (Algorithm~\ref{alg:ListTrees}), \emph{we produce the new
  certificate $D'$ from the current $D$ in $O(1)$ time when we are in
  a unary node}.  On the other hand, we completely rebuild the
certificate twice when we are in a binary nodes (since either child
could be unary at the next recursion level).

\begin{lemma}
  \label{fig:ksubtreeinvariant}
  Algorithm \ref{alg:ListTrees2} correctly maintains the invariant
  $S\sqsubseteq D$.
\end{lemma}
\begin{proof}
  The base case is when $|S|=k$, as before. Hence, we discuss the
  recursion, where we suppose that $S$ is a $k'$-subtree with $k' < k$.
  Let $e=(e^-,e^+)$ denote the edge returned by $\chooseedge(S,D)$.

  \medskip 

  \emph{First:} Consider the situation in which we want to output all
  the $k$-subtrees that contain $S' = S + \langle e \rangle$.  Now, from
  certificate $D = S \cup L \cup F$ we obtain a new $D' = S' \cup L'
  \cup F'$ according to the three cases behind \chooseedge, for which
  we have to prove that $S' \sqsubseteq D'$.

  \smallskip

  \ref{item:external}~[\textsf{external edge}] We simply recompute $D'
  := \mkdfs(S')$. So $S' \sqsubseteq D'$ by definition of
  relation~$\sqsubseteq$.

  \smallskip

  \ref{item:back}~[\textsf{back edge}] Same as in the case above, we
  recompute $D' := \mkdfs(S')$ and therefore $S' \sqsubseteq D'$.

  \smallskip

  \ref{item:tree}~[\textsf{tree edge}] In this case, the set of edges
  of the certificate does not change ($D' = D$ seen as sets), but the
  internal representation described in Section~\ref{sub:certificates}
  changes partially since $S' = S + \langle e\rangle$. The flag
  \isunary\ is true, and so \treecut\ and \promote\ can be
  invoked. The former is done by \chooseedge, which correctly returns
  $e=(e^-,e^+)$ as the last tree edge in the cutlist $\hat C(S)$. The
  latter is done to promote the children $r_1, r_2, \ldots, r_c$ of
  $e^+$.

  To show that $S' \sqsubseteq D'$ for this case, we need to prove
  that the resulting certificate $D'$ is the same as the one returned
  by an explicit call to $\mkdfs(S')$, which we clearly want to avoid
  calling.  Let $D_0 = S_0 \cup L_0 \cup F_0$ be the output of the
  call to $\mkdfs(S')$, and let $D' = S' \cup L' \cup F'$ be what we
  obtain in Algorithm~\ref{alg:ListTrees2}. 

  First of all, note that $S' = S_0 = S + \langle e \rangle$ by
  definition of \mkdfs. This means that the sorted lists for $S'$ and
  $\hat V(S')$ are ``equal'' to those for $S_0$ and $\hat V(S_0)$
  (elements are the same and in the same order). Hence, $S' = S_ 0 =
  \langle e_1, e_2, \ldots, e_{k'}, e \rangle$ and $\hat V(S') = \hat
  V(S_0) = \langle u_1, u_2, \ldots, u_{k'}, e^+
  \rangle$. 

  Consequently, the cutsets $C(S') = C(S_0)$: when considering the
  corresponding cutlists $\hat C(S')$ and $\hat C(S_0)$, recall that
  \mkdfs\ performs a multi-source truncated DFS from the vertices
  $u_1$, $u_2$,~\dots, $u_{k'}$, $e^+$ in this order (where all of
  them are initially marked as already visited).  When \mkdfs\ starts
  from $ e^- \equiv\, u_j$ (for some $1 \leq j \leq k'$), it does not
  explore $e^+$ through edge $e$. Moreover, $r_1, r_2, \ldots, r_c$
  are not explored as well, since otherwise there would be back edges
  and \isunary\ would be false. Since $e^+$ is the last in $S'$, when
  \mkdfs\ starts from $e^+$, observe that $e^-$ has been totally
  explored, and $r_1, r_2, \ldots, r_c$ are discovered now from
  $e^+$. Since $e$ is the last tree edge in the cutlist $\hat C(S)$,
  we have that the ordering in the new cutlists $\hat C(S')$ and $\hat
  C(S_0)$ must be the same.

  Consider now $L'$ and $L_0$. We show that $L' = L_0$ using the fact
  that $\hat C(S') = \hat C(S_0)$. Operation $\promote(e^+,D)$ removes
  $e$ from $L$ and adds tree edges $(e^+, r_i)$ for its children $r_1,
  r_2, \ldots, r_c$ to form $L'$. Note these edges are added in the
  same order as they were discovered by $\mkdfs(S)$ and $\mkdfs(S')$
  since \isunary\ is true and $\hat C(S') = \hat C(S_0)$. Since $L_0$
  does not contain $e$, we have that $L'=L_0$.
  
  It remains to show that $F' = F_0$. This is easy since $\hat C(S') =
  \hat C(S_0)$: the subtree at $r_i$ is totally explored before that
  at $r_j$ for $i < j$. Hence, when $r_1, r_2, \ldots, r_c$ are
  promoted as roots in $F'$, their corresponding subtrees do not
  change. Also, their ordering in the sublist for $F'$ and $F_0$ is
  the same because $\hat C(S') = \hat C(S_0)$.

  Finally, since there are no back edges, the update of $\eta$ 
  only involves $\eta[e^-]$ and $\eta[e^+]$ as discussed in the
  implementation of \promote. For the same reason, the only case in
  which \isunary\ can becomes false is when $|\adj(e^+)| > \fdeg(e^+)
  + 1$. This completes the proof that $S' \sqsubseteq D'$ for
  case~\ref{item:tree}.

  \medskip

  \emph{Second:} Consider the situation in which we want to list all
  the $k$-subtrees that contain $S$ but do \emph{not} contain $e$. This
  is equivalent to list all the $k$-subtrees that contain $S$ in
  $G-\{e\}$. Hence, we remove $e$ from $G$ and recomputed the
  certificate from scratch before each of the two recursive calls.
  Consider the three cases behind \chooseedge.
  Cases~\ref{item:external} and~\ref{item:back} are trivial, since we
  recompute $D'' := \mkdfs(S)$ and so $S \sqsubseteq D''$.
  Case~\ref{item:tree} cannot arise by Lemma~\ref{lemma:choose_cases}.
\end{proof}

\subsection{Analysis}
\label{sub:listtree2-analysis}

We implement $\chooseedge(S,D)$ so that it can now exploit the
information in $D$. At each node $S$ of the recursion tree, when it
selects an edge $e$ that belongs to the cutset $C(S)$, it first
considers the edges in $C(S)$ that are external or back
(cases~\ref{item:external}--\ref{item:back}) before the edges in $D$
(case~\ref{item:tree}).

\begin{lemma} 
  \label{lem:choosetime_k}
  There is an implementation of \chooseedge\ in $O(1)$ for unary
  nodes in the recursion tree and $O(k^2)$ for binary nodes.
\end{lemma}
\begin{proof}
  Given $D$, we can check if the current node $S$ in the recursion
  tree is unary by checking the flag \isunary. If this is the case we
  simply return the edge indicated by $\treecut(D)$ in $O(1)$
  time. Otherwise, the node $S$ is binary, and so there exists at
  least an external edge or a back edge. We visit the first $2k$ edges
  in each $\adj(u)$ for every $u \in S$. Note that less than $k$ edges
  can connect $u$ to vertices in $V[S]$ and less than $k$ edges can
  connect $u$ to vertices in $V[D]$: if an external edge exists, we
  can find it in $O(k^2)$ time. Otherwise, no external edge exists, so
  there must be a back edge to be returned since the node is
  binary. We visit the first $k$ edges in each $\adj(u)$ for every $u
  \in S$, and surely find one back edge in $O(k^2)$.
\end{proof}

\begin{lemma}
  \label{lem:listtrees2}
  Algorithm~\ref{alg:ListTrees2} takes $O(s_i \, k^2)$ time and $O(mk)$
  space, where $s_i$ is the number of $k$-subtrees reported by \listtrees.
\end{lemma}
\begin{proof}
  We report the breakdown of the costs for a call to \listtrees\
  according the cases, using Lemmas~\ref{lem:certificate-timespace}
  and~\ref{lem:choosetime_k}:
  \begin{itemize}
  \item[\ref{item:external}] External edge: $O(k^2)$ for \chooseedge\
    and \mkdfs, $O(1)$ for $\edel$, $\idel$.
  \item[\ref{item:back}] Back edge: $O(k^2)$ for \chooseedge\ and
    \mkdfs, and $O(1)$ for $\edel$ and $\idel$.
  \item[\ref{item:tree}] Tree edge: $O(1)$ for \chooseedge\ and
    \promote.
  \end{itemize}
  Hence, binary nodes take $O(k^2)$ time and unary nodes take $O(1)$
  time. By Corollary~\ref{cor:right-branch}, there are $O(s_i)$ binary
  nodes and $O(s_i k)$ unary nodes, and so
  Algorithm~\ref{alg:ListTrees2} takes $O(s_i \, k^2)$ time. The space
  analysis is left unchanged, namely, $O(mk)$ space.
\end{proof}

\begin{theorem}
  \label{the:main2}
  Algorithm~\ref{alg:ListAllTrees2} solves
  Problem~\ref{problem:ksubtreelisting} in $O(sk^2)$ time and $O(mk)$ space.
\end{theorem}
\begin{proof}
  The vertices belonging to the connected components of size less than
  $k$ in the residual graph, now contribute with $O(m)$ total time
  rather than $O(n k^2)$. The rest of the complexity follows from
  Lemma~\ref{lem:listtrees2}.
\end{proof}

\section{Optimal approach: amortization}
\label{sec:output-sensitive}

In this section, we discuss how to adapt
Algorithm~\ref{alg:ListTrees2} so that a more careful analysis can
show that it takes $O(sk)$ time to list the $k$-subtrees. Considering
$\listtrees$, observe that each of the $O(s_i k)$ unary nodes requires
a cost of $O(1)$ time and therefore they are not much of problem. On
the contrary, each of the $O(s_i)$ binary nodes takes $O(k^2)$ time:
our goal is to improve this case.

Consider the operations on a binary node $S$ of the recursion tree
that take $O(k^2)$ time, namely: $(I)$~$e := \chooseedge(S,D)$;
$(II)$~$D' := \mkdfs(S')$, where $S' \equiv S+ \langle e \rangle$; and
$(III)$~$D'' := \mkdfs(S)$ in $G-\{e\}$. In all these operations,
while scanning the adjacency lists of vertices in $V[S]$, we visit
some edges $e'=(u,v)$, named \emph{internal}, such that $e' \not \in
S$ with $u,v \in V[S]$. These internal edges of $V[S]$ can be
visited even if they were previously visited on an ancestor node to
$S$. In Section~\ref{sub:internaledges}, we show how to amortize the
cost induced by the internal edges. In Section~\ref{sub:amortization},
we show how to amortize the cost induced by the remaining edges and
obtain a delay of $t(k) = k^2$ in our optimal
output-sensitive algorithm.

\subsection{Internal edges}
\label{sub:internaledges}

To avoid visiting the internal edges of $V[S]$ several times
throughout the recursion tree, we remove these edges from the graph
$G$ on the fly, and introduce a global data structure, which we call
\emph{parking lists}, to store them temporarily. Indeed, out of the
possible $O(n)$ incident edges in vertex $u \in V[S]$, less than~$k$
are internal: it is simply too costly removing these internal edges by
a complete scan of $\adj(u)$. Therefore we remove them as they appear
while executing \chooseedge\ and \mkdfs\ operations.

Formally, we define \emph{parking lists} as a global array $P$ of $n$
pointers to lists of edges, where $P[u]$ is the list of internal edges
discovered for $u \in V[S]$. When $u \not \in V[S]$, $P[u]$ is null.
On the implementation level, we introduce a slight modification of the
\chooseedge\ and \mkdfs\ algorithms such that, when they meet for the
first (and only) time an internal edge $e' =(u,v)$ with $u,v \in
V[S]$, they perform $\edel(e')$ and add $e'$ at the end of both
parking lists $P[u]$ and $P[v]$. We also keep a cross reference to the
occurrences of $e'$ in these two lists.

Additionally, we perform a small modification in algorithm \listtrees\
by adding a fifth step in Algorithm~\ref{alg:ListTrees2} just before
it returns to the caller. Recall that on the recursion node $S+
\langle e \rangle$ with $e = (e^-,e^+)$, we added the vertex $e^+$ to
$V[S]$. Therefore, when we return from the call, all the internal
edges incident to $e^+$ are no longer internal edges (and are the only
internal edges to change status). On this new fifth step, we scan $P[e^+]$
and for each edge $e'= (e^+,x)$ in it, we remove $e'$ from both $P[e^+]$
and $P[x]$ in $O(1)$ time using the cross reference. Note that when
the node is unary there are no internal edges incident to $e^+$, so
$P[e^+]$ is empty and the total cost is $O(1)$. When the node is
binary, there are at most $k-1$ edges in $P[e^+]$, so the cost is
$O(k)$.

\begin{lemma}
  \label{lem:internal_edges}
  The operations over internal edges done in \listtrees\ have a total
  cost of $O(s_i k)$ time.
\end{lemma}
\begin{proof}
For each $\edel(e')$ in  \chooseedge\ or \mkdfs, there
  is exactly one matching $\idel(e')$ in \listtrees\ done in a binary
  node. Since there are $O(s_i)$ binary nodes and $O(k)$ {\idel} calls per
  binary node, the total contribution of internal edges to the
  complexity can be overall bounded by $O(s_i k)$ time.
\end{proof}

\subsection{Amortization}
\label{sub:amortization}

Let us now focus on the contribution given by the remaining edges,
which are not internal for the current $V[S]$.  Given the results in
Section~\ref{sub:internaledges}, for the rest of this section \emph{we
  can assume wlog that there are no internal edges} in $V[S]$, namely,
$E[S] = S$. We introduce two metrics that help us to parameterize the
time complexity of the operations done in binary nodes of the
recursion tree.

The first metric we introduce is helpful when analyzing the operation
\chooseedge. For connected edge sets $S$ and $X$ with $S \sqsubseteq
X$, define the \emph{cut number}~$\gamma_X$ as the number of edges in
the induced (connected) subgraph $G[X] = (V[X], E[X])$ that are in the
cutset $C(S)$ (i.e. tree edges plus back edges): $\gamma_X = |\,E[X]
\cap C(S)\,| $.

\begin{lemma}
  \label{lem:choose_improved}
  For a binary node $S$ with certificate $D$, $\chooseedge(S,D)$ takes
  $O(k + \gamma_D)$ time.
\end{lemma}
\begin{proof}
  Consider that we are in a binary node $S$ with associated
  certificate $D$, otherwise just return $\treecut(D)$ in $O(1)$. We
  examine each adjacency list $\adj(u)$ for $u \in V[S]$ until we find an
  external edge. If we scan all these lists completely without finding
  an external edge, we have visited less than $k$ edges that belong to
  $S$ and $\gamma_{D}$ edges that belong to $C(S)$ (the tree and back
  edges of $D$). When we know that no external edges exist, a back edge
  is found in $O(k)$ time. This totals $O(k + \gamma_{D})$ time.
\end{proof}

For connected edge sets $S$ and $X$ with $S \sqsubseteq X$, the second
metric is the \emph{cyclomatic number} $\nu_X$ (also known as circuit
rank, nullity, or dimension of cycle space) as the smallest number of
edges which must be removed from $G[X]$ so that no cycle remains in
it: $\nu_X = |\,E[X]\,| - |\,V[X]\,| + 1$.


Using the cyclomatic number of a certificate $D$ (ignoring the
internal edges of $V[S]$), we obtain a lower bound on the number of
$k$-subtrees that are output in the leaves descending from a node $S$ in the
recursion tree.

\begin{lemma}
  \label{lem:cyclomatic}
  Considering the cyclomatic number $\nu_D$
  and the fact that $|V[D]|=k$, we have that $G[D]$ contains at least $\nu_D$ $k$-subtrees.
\end{lemma}
\begin{proof}
As $G[D]$ is connected, it contains at least one spanning tree. Take any
spanning tree $T$ of $G[D]$. Note that there are $\nu_D$ edges in $G[D]$ that
do not belong to this spanning tree. For each of these $\nu_D$ edges
$e_1, \ldots, e_{\nu}$ one can construct a spanning tree $T_i$ by
adding $e_i$ to $T$ and breaking the cycle introduced. These $\nu_D$
spanning trees are all different $k$-subtrees as they include a different edge
$e_i$ and have $k$ different edges.
\end{proof}

\begin{lemma}
  \label{lem:mkdfs_improved}
  For a node $S$ with certificate $D$, computing $D' =
  \mkdfs(S)$ takes $O(k + \nu_{D'})$ time.
\end{lemma}
\begin{proof}
	Consider the certificate $k$-subtree $D'$ returned by
	$\mkdfs(S)$. To calculate $D'$, the edges in $D'$ are visited
	and, in the worst case, all the $\nu_{D'}$ edges in between
	vertices of $V[D']$ in G are also visited. No other edge
	$e'=(u,v)$ with $u\in D'$ and $v \not \in D'$ is visited, as
	$v$ would belong to $D'$ since we visit vertices depth first.
\end{proof}

Recalling that the steps done on a binary node $S$ with certificate
$D$ are: $(I)$~$e := \chooseedge(S,D)$; $(II)$~$D' := \mkdfs(S')$,
where $S' \equiv S+ \langle e \rangle$; and $(III)$~$D'' := \mkdfs(S)$
in $G-\{e\}$, they take a total time of $O(k + \gamma_D + \nu_{D'} +
\nu_{D''})$. We want to pay $O(k)$ time on the recursion node $S$ and
amortize the \emph{rest} of the cost to some suitable nodes descending
from its \emph{left child} $S'$ (with certificate $D'$). To do this we
are to relate $\gamma_D$ with $\nu_{D'}$ and avoid performing
step~$(III)$ in $G-\{e\}$ by maintaining $D''$ from $D'$. We exploit
the property that the cost $O(k + \nu_{D'})$ for a node $S$ in the
recursion tree can be amortized using the following lemma:

\begin{lemma}
\label{lem:amortization}
Let $S'$ be the left child (with certificate $D'$) of a generic node
$S$ in the recursion tree. The sum of $O(\nu_{D'})$ work, over all
left children $S'$ in the recursion tree is upper bounded by
$\sum_{S'} \nu_{D'} = O( s_i k )$.
\end{lemma}
\begin{proof}
  By Lemma~\ref{lem:cyclomatic}, $S'$ has at least $\nu_{D'}$
  descending leaves. Charge $O(1)$ to each leaf descending from $S'$
  in the recursion tree. Since $S'$ is a left child and we know that
  the recursion tree is $k$-left-bounded by
  Lemma~\ref{lem:kleftbounded}, each of the $s_i$ leaves can be
  charged at most $k$ times, so $\sum_{S'} \nu_{D'} = O(s_i k)$ for
  all such $S'$.
\end{proof}

We now show how to amortize the $O(k + \gamma_D)$ cost of
step~$(I)$. Let us define $\comb(S')$ for a left child $S'$ in the
recursion tree as its maximal path to the right (its right spine) and
the left child of each node in such a path. Then, $|\comb(S')|$ is the
number of such left children.

\begin{lemma}
\label{lem:comb}
On a node $S$ in the recursion tree, the cost of
$\chooseedge(S,D)$ is $O(k + \gamma_D) = O(k + \nu_{D'} + |\comb(S')|)$.
\end{lemma}
\begin{proof}
  Consider the set $E'$ of $\gamma_D$ edges in $E[D] \cap C(S)$. Take
  $D'$, which is obtained from $S' = S + \langle e \rangle$, and
  classify the edges in $E'$ accordingly. Given $e' \in E'$, one of
  three possible situations may arise: either $e'$ becomes a tree edge
  part of $D'$ (and so it contributes to the term $k$), or $e'$
  becomes a back edge in $G[D']$ (and so it contributes to the term
  $\nu_{D'}$), or $e'$ becomes an external edge for $D'$. In the
  latter case, $e'$ will be chosen in one of the subsequent recursive
  calls, specifically one in $\comb(S')$ since $e'$ is still part of
  $C(S')$ and will surely give rise to another $k$-subtree in a
  descending leaf of $\comb(S')$.
\end{proof}

While the $O(\nu_{D'})$ cost over the leaves of $S'$ can be amortized
by Lemma~\ref{lem:mkdfs_improved}, we need to show how to amortize the
cost of $|\comb(S')|$ using the following:

\begin{lemma}
\label{lem:comb_amortize}
$\sum_{S'} |\comb(S')| = O( s_i k ) $ over all left children
$S'$ in the recursion. 
\end{lemma}
\begin{proof}
	Given a left child $S^*$ in the recursion tree, there is
	always a unique node $S'$ that is a left child such that its
	$\comb(S')$ contains $S^*$. Hence, $\sum_{S'} |\comb(S')|$ is
	upper bounded by the number of left children in the recursion
	tree, which is $O(s_i k)$.
\end{proof}

At this point we are left with the cost of computing the two
$\mkdfs$'s. Note that the cost of step~$(II)$ is $O(k + \nu_{D'})$,
and so is already expressed in terms of the cyclomatic number of its
left child, $\nu_{D'}$ (so we use Lemma~\ref{lem:cyclomatic}). The
cost of step~$(III)$ is $O(k + \nu_{D''})$, expressed with the
cyclomatic number of the certificate of its \emph{right} child. This
cost is not as easy to amortize since, when the edge $e$ returned by
\chooseedge\ is a back edge, $D'$ of node $S+\langle e \rangle$ can
change \emph{heavily} causing $D'$ to have just $S$ in common with
$D''$. This shows that $\nu_{D''}$ and $\nu_{D'}$ are not easily
related.

Nevertheless, note that $D$ and $D''$ are the same certificate since
we only remove from $G$ an edge $e \not \in D$. The only thing that
can change by removing edge $e=(e^-,e^+)$ is that the right child of
node $S'$ is no longer binary (i.e.\mbox{} we removed the last back
edge). The question is if we can check quickly whether it is unary in
$O(k)$ time: observe that $|\adj(e^-)|$ is no longer the same,
invalidating the flag $\isunary$ (item~\ref{item:isunary} of
Section~\ref{sub:certificates}). Our idea is the following: instead of
recomputing the certificate $D''$ in $O(k+\nu_{D''})$ time, we update
the $\isunary$ flag in just $O(k)$ time. We thus introduce a new operation
$D'' = \unary(D)$, a valid replacement for $D'' = \mkdfs(D)$ in $G -
\{e\}$: it maintains the certificate while recomputing the flag
$\isunary$ in $O(k)$ time.

\begin{lemma}
\label{lem:unary}
Operation $\unary(D)$ takes $O(k)$ time and correctly computes
$D''$.
\end{lemma}
\begin{proof}
  Consider the two certificates $D$ and $D''$: we want to compute
  directly $D''$ from $D$. Note that the only difference in the two
  certificate is the flag \isunary. Hence, it suffices to show how to
  recompute \isunary\ from scratch in $O(k)$ time. This is
  equivalent to checking if the cutset $C(S)$ in $G - \{e\}$
  contains only tree edges from $D$. It suffices to examine the each
  adjacency list $\adj(u)$ for $u \in S$, until either we find an
  external or a back edge (so \isunary\ is false) or we scan all these
  lists (so \isunary\ is true): in the latter case, there are at most
  $2k$ edges in the adjacency lists of the vertices in $V[S]$ as these
  are the edges in $S$ or in $D$. All other edges are external or back
  edges.
\end{proof}

Since there is no modification or impact on unary nodes of the
recursion tree, we finalize the analysis.

\begin{lemma}
\label{lem:binary_nodes_parameterized}
The cost of $\listtrees(S,D)$ on a binary node $S$ is
  $O(k + \nu_{D'} + |\comb(S')|)$.
\end{lemma}
\begin{proof}
	Consider the operations done in binary nodes. By Lemmas
	\ref{lem:comb}, \ref{lem:mkdfs_improved} and \ref{lem:unary}
	we have the following breakdown of the costs:
\begin{itemize}
  \item Operation $e := \chooseedge(S,D)$ takes $O(k + \nu_{D'} +
	  |\comb(S')|)$;
  \item Operation $D':= \mkdfs(D)$ takes $O(k + \nu_{D'})$;
  \item Operation $D'' := \unary(D)$ takes $O(k)$.
\end{itemize}

This sums up to $O(k + \nu_{D'} + |\comb(S')|)$ time per binary
node.
\end{proof}

\begin{lemma}
\label{lem:amortized_listtrees}
The algorithm \listtrees takes $O(s_i k)$ time and $O(m k)$ space.
\end{lemma}
\begin{proof}
	By Lemma~\ref{lem:binary_nodes_parameterized} we have a cost
	of $O(k + \nu_{D'} + |\comb(S')|)$ per binary node on the
	recursion tree. Given that there are $O(s_i)$ binary nodes by
	Corollary~\ref{cor:right-branch} and given the amortization of
	$\nu_{D'}$ and $|\comb(S')|$ in each left children $S'$ over
	the leafs of the recursion tree in
	Lemmas~\ref{lem:amortization} and \ref{lem:comb_amortize}, the
	sums of $\sum_{S}{k} + \sum_{S'}{\nu_{D'}} +
	\sum_{S'}{|\comb(S')|} = s_i k$.  Additionally, by
	Lemma~\ref{lem:listtrees2}, the $O(s_i k)$ unary nodes give a
	total contribution $O(s_i k)$ to the cost. This is sums up to
	$O(s_i k)$ total time.
\end{proof}

Note that our data structures are lists and array, so it is not
difficult to replace them with \emph{persistent} arrays and lists, a
classical trick in data structures. As a result, we just need $O(1)$
space per pending recursive call, plus the space of the parking lists,
which makes a total of $O(m)$ space.  

\begin{theorem}
\label{thm:final}
Algorithm~\ref{alg:ListAllTrees2} takes a total of $O(s k)$ time,
being therefore optimal, and $O(m)$ space.
\end{theorem}
\begin{proof}
Algorithm~\ref{alg:ListAllTrees2} deletes the connected component of
size smaller than $k$ and containing $v_i$ in the residual graph. So
this cost sums up to $O(m)$. Otherwise, we can proceed as in the
analysis of Algorithm~\ref{alg:ListAllTrees2}. Note that our data
structures are lists and array, so it is not difficult to replace
them with \emph{persistent} arrays and lists, a classical trick in
data structures. As a result, we just need $O(1)$ space per pending
recursive call, plus the space of the parking lists, which makes a
total of $O(m)$ space.
\end{proof}

We finally show how to obtain an
efficient delay. We exploit the following property on the recursion
tree, which allows to associate a unique leaf with an internal node
\emph{before} exploring the subtree of that node (recall that we are
in a recursion tree). Note that only the rightmost leaf descending
from the root is not associated in this way, but we can easily handle
this special case.

\begin{lemma}
  \label{lemma:rightmost_leaf}
  For a binary node $S$ in the recursion tree, $\listtrees(S,D)$
  outputs the $k$-subtree $D$ in the rightmost leaf
  descending from its left child $S'$.
\end{lemma}
\begin{proof}
  Follow the rightmost spine in the subtree rooted at $S'$. An edge of
  $D$ is chosen after that all the external and back edges are
  deleted. Now, adding that edge to the partial solution, may bring
  new external and back edges. But they are again deleted along the
  path that follows the right spine. In other words, the path from $S'$
  to its rightmost leaf includes only edges from $D$ when branching
  to the left, and removes the external and back edges when branching
  to the right.
\end{proof}

Nakano and Uno
\cite{nakano} have introduced this nice trick.
Classify a binary node $S$ in the recursion tree as even (resp., odd)
if it has an even (resp., odd) number of ancestor nodes that are
binary.  Consider the simple modification to \listtrees\ when
$S$ is binary: if $S$ is even then output $D$ \emph{immediately
  before} the two recursive calls; otherwise ($S$ is odd), output $D$
\emph{immediately after} the two recursive calls. 

\begin{theorem}
\label{thm:final-delay}
Algorithm~\ref{alg:ListAllTrees2} can be implemented with delay $t(k) = k^2$.
\end{theorem}
\begin{proof}
Between any two $k$-subtrees that are listed one after the other, we
traverse $O(1)$ binary nodes and $O(k)$ unary nodes, so the total cost
is $O(k^2)$ time in the worst case.
\end{proof}

\chapter{Listing k-subgraphs}
\label{chapter:ListingKSubgraphs}

When considering an undirected connected graph $G$ with $n$ vertices
and $m$ edges, we solve the problem of listing all the connected
induced subgraphs of $G$ with $k$ vertices.
Figure~\ref{fig:inducedsubgraphs} shows an example graph $G_1$ and its
$k$-subgraphs when $k=3$.  We solve this problem optimally, in time
proportional to the size of the input graph $G$ plus the size of the
edges in the $k$-subgraphs to output.

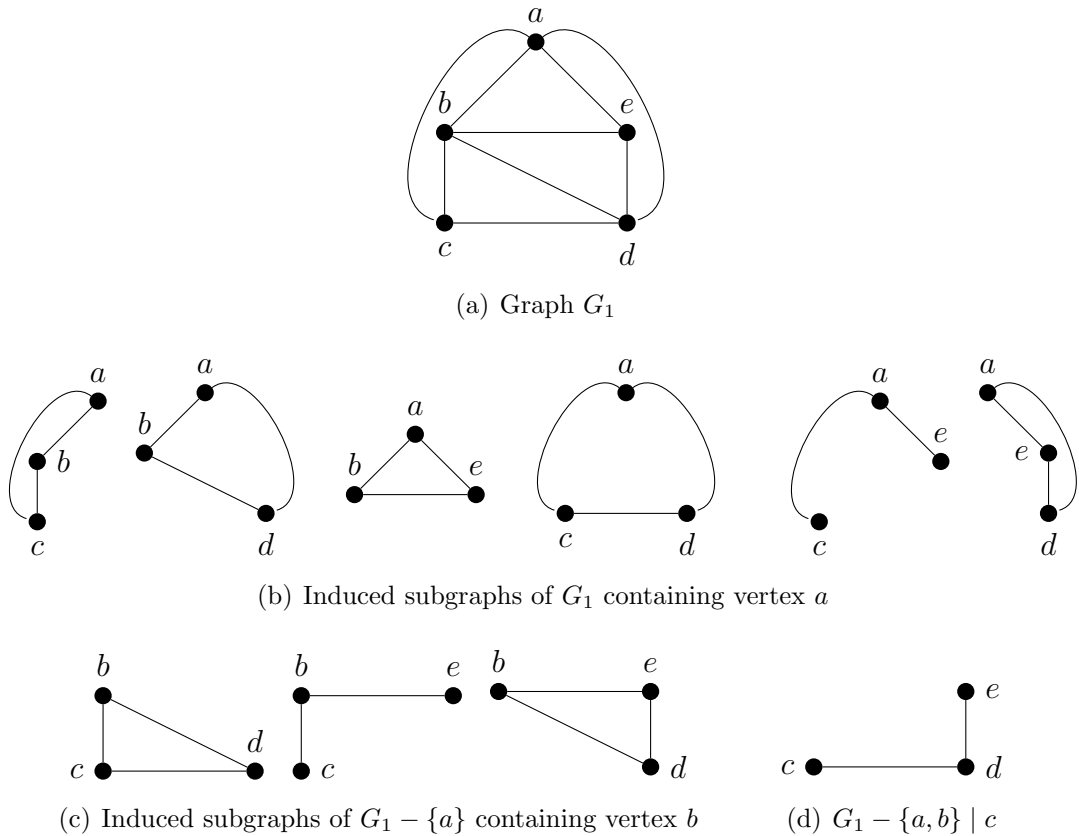
\begin{figure}[t]
\centering
\subfigure[Graph $G_1$]{
\centering
\begin{tikzpicture}[shorten >=1pt,->,scale=1.2]
  \tikzstyle{vertex}=[shape=circle,draw,thick,fill=black,minimum size=2pt,inner sep=2pt]
  \node[vertex,label=above:$a$] (A) at (0,0) {};
  \node[vertex,label=above:$b$] (B) at (-1,-1)   {};
  \node[vertex,label=below:$c$] (C) at (-1,-2)  {};
  \node[vertex,label=below:$d$] (D) at (1,-2)  {};
  \node[vertex,label=above:$e$] (E) at (1,-1)  {};
  \draw (A) -- (B) -- (C) -- (D) -- (E) -- (A) -- cycle;
  \draw (B) -- (E) -- cycle;
  \draw (B) -- (D) -- cycle;
  \path (A) edge [-,bend right=100] (C);
  \path (A) edge [-,bend left=100] (D);
\end{tikzpicture}
} 

\subfigure[Induced subgraphs of $G_1$ containing vertex $a$]{
\begin{tikzpicture}[shorten >=1pt,->,scale=0.8]
  \tikzstyle{vertex}=[shape=circle,draw,thick,fill=black,minimum size=2pt,inner sep=2pt]
  \node[vertex,label=above:$a$] (A) at (0,0) {};
  \node[vertex,label=right:$b$] (B) at (-1,-1)   {};
  \node[vertex,label=below:$c$] (C) at (-1,-2)  {};
  \draw (A) -- (B) -- (C) -- cycle;
  \path (A) edge [-,bend right=100] (C);
\end{tikzpicture}
\begin{tikzpicture}[shorten >=1pt,->,scale=0.8]
  \tikzstyle{vertex}=[shape=circle,draw,thick,fill=black,minimum size=2pt,inner sep=2pt]
  \node[vertex,label=above:$a$] (A) at (0,0) {};
  \node[vertex,label=above:$b$] (B) at (-1,-1)   {};
  \node[vertex,label=below:$d$] (D) at (1,-2)  {};
  \draw (A) -- (B) -- cycle;
  \draw (B) -- (D) -- cycle;
  \path (A) edge [-,bend left=100] (D);
\end{tikzpicture}
\begin{tikzpicture}[shorten >=1pt,->,scale=0.8]
  \tikzstyle{vertex}=[shape=circle,draw,thick,fill=black,minimum size=2pt,inner sep=2pt]
  \node[vertex,label=above:$a$] (A) at (0,0) {};
  \node[vertex,label=above:$b$] (B) at (-1,-1)   {};
  \node[vertex,label=above:$e$] (E) at (1,-1)  {};
  \node[color=white] (D) at (1,-2)  {};
  \draw (A) -- (B) -- cycle;
  \draw (B) -- (E) -- cycle;
  \draw (E) -- (A) -- cycle;
\end{tikzpicture}
\begin{tikzpicture}[shorten >=1pt,->,scale=0.8]
  \tikzstyle{vertex}=[shape=circle,draw,thick,fill=black,minimum size=2pt,inner sep=2pt]
  \node[vertex,label=above:$a$] (A) at (0,0) {};
  \node[vertex,label=below:$c$] (C) at (-1,-2)  {};
  \node[vertex,label=below:$d$] (D) at (1,-2)  {};
  \draw (C) -- (D) -- cycle;
  \path (A) edge [-,bend right=100] (C);
  \path (A) edge [-,bend left=100] (D);
\end{tikzpicture}
\begin{tikzpicture}[shorten >=1pt,->,scale=0.8]
  \tikzstyle{vertex}=[shape=circle,draw,thick,fill=black,minimum size=2pt,inner sep=2pt]
  \node[vertex,label=above:$a$] (A) at (0,0) {};
  \node[vertex,label=below:$c$] (C) at (-1,-2)  {};
  \node[vertex,label=above:$e$] (E) at (1,-1)  {};
  \draw (E) -- (A) -- cycle;
  \path (A) edge [-,bend right=100] (C);
\end{tikzpicture}
\begin{tikzpicture}[shorten >=1pt,->,scale=0.8]
  \tikzstyle{vertex}=[shape=circle,draw,thick,fill=black,minimum size=2pt,inner sep=2pt]
  \node[vertex,label=above:$a$] (A) at (0,0) {};
  \node[vertex,label=below:$d$] (D) at (1,-2)  {};
  \node[vertex,label=left:$e$] (E) at (1,-1)  {};
  \draw (D) -- (E) -- (A) -- cycle;
  \path (A) edge [-,bend left=100] (D);
\end{tikzpicture}
}

\subfigure[Induced subgraphs of $G_1-\{a\}$ containing vertex $b$]{
\begin{tikzpicture}[shorten >=1pt,->,scale=1.0]
  \tikzstyle{vertex}=[shape=circle,draw,thick,fill=black,minimum size=2pt,inner sep=2pt]
  \node[vertex,label=above:$b$] (B) at (-1,-1)   {};
  \node[vertex,label=left:$c$] (C) at (-1,-2)  {};
  \node[vertex,label=above:$d$] (D) at (1,-2)  {};
  \draw (B) -- (C) -- (D) -- cycle;
  \draw (B) -- (D) -- cycle;
\end{tikzpicture}
\begin{tikzpicture}[shorten >=1pt,->,scale=1.0]
  \tikzstyle{vertex}=[shape=circle,draw,thick,fill=black,minimum size=2pt,inner sep=2pt]
  \node[vertex,label=above:$b$] (B) at (-1,-1)   {};
  \node[vertex,label=right:$c$] (C) at (-1,-2)  {};
  \node[vertex,label=above:$e$] (E) at (1,-1)  {};
  \draw (B) -- (C) -- cycle;
  \draw (B) -- (E) -- cycle;
\end{tikzpicture}
\begin{tikzpicture}[shorten >=1pt,->,scale=1.0]
  \tikzstyle{vertex}=[shape=circle,draw,thick,fill=black,minimum size=2pt,inner sep=2pt]
  \node[vertex,label=above:$b$] (B) at (-1,-1)   {};
  \node[vertex,label=right:$d$] (D) at (1,-2)  {};
  \node[vertex,label=above:$e$] (E) at (1,-1)  {};
  \draw (D) -- (E) -- cycle;
  \draw (B) -- (E) -- cycle;
  \draw (B) -- (D) -- cycle;
\end{tikzpicture}
}
\hspace{1.0em}
\subfigure[$G_1-\{a,b\} \mid c$]{
\begin{tikzpicture}[shorten >=1pt,->,scale=1.0]
  \tikzstyle{vertex}=[shape=circle,draw,thick,fill=black,minimum size=2pt,inner sep=2pt]
  \node[vertex,label=left:$c$] (C) at (-1,-2)  {};
  \node[vertex,label=right:$d$] (D) at (1,-2)  {};
  \node[vertex,label=right:$e$] (E) at (1,-1)  {};
  \draw (C) -- (D) -- (E) -- cycle;
\end{tikzpicture}
}

\caption{Example graph $G_1$ and its $3$-subgraphs}
\label{fig:inducedsubgraphs}

\end{figure}

\medskip

There exist vast practical applications of using $k$-subgraphs as
mining pattern in diverse fields. In general, determining subgraphs
that appear frequently in one or more networks is considered relevant
in the study of biological \cite{alm2003,alon2003,junker2008,milo},
social \cite{hay2008resisting,hogg2004enhancing,wei2006method} and
other complex networks~\cite{arenas2008motif,boccaletti2006complex,
costa2007characterization,myers2003software}.  These patterns,
frequently called motifs, can provide insight to the function and
design principles of these networks.

In the particular case of biological networks, Wernicke
\cite{wernicke2006efficient} has introduced the ESU algorithm as a
base to find motifs. This algorithm solves the problem of listing the
$k$-subgraphs of a network. The author further modifies the algorithm
to sample $k$-subgraphs and is able to efficiently find network
motifs. Nevertheless, a theoretical analysis of the algorithm is not
performed and it does not seem to be possible to implement the
algorithm in a optimal output-sensitive sensitive way.

Researchers focusing on finding motifs in graphs face the problem of
counting or estimating the occurrences of each structure and thus face
the graph isomorphism problem. For this reason, there appears to be
little research done on efficiently listing the patterns themselves. 

In \cite{avis}, Avis and Fukuda propose the application of the general
reverse search method and obtain an algorithm taking $O(m n \eta)$
time, where $\eta$ is the number of $k$-subgraphs in the input graph.

\medskip

Equivalent to the problem of listing $k$-subgraphs is to list the
\emph{vertices} in the $k$-subgraphs of $G$. Although it is trivial to
reduce one problem to the other, note that an optimal output-sensitive
solution for the problem of listing $k$-subgraphs does not imply an
optimal solution for the problem of listing its vertices, as in the
latter case is sufficient to output the vertices in the $k$-subgraph
(instead of the edges).

Interestingly, for the particular case of when graph $G$ is a tree,
the problem of listing $k$-subgraphs is equivalent to the problem of
listing $k$-subtrees. We achieve the same time bound as achieved in
Chapter~\ref{chapter:ListingKSubtrees}. For this particular case, the
algorithm proposed in \cite{wasa2012} is able to enumerate
$k$-subgraphs in $O(1)$ per $k$-subgraph.

\medskip

We present our solution based on the binary partition method (Section
\ref{subsection:BinaryPartitionMethod}). We divide the problem of
listing all the $k$-subgraphs in two subproblems by choosing a vertex
$v \in V$: (i) we list the $k$-subgraphs that contain $v$, and (ii) those that
do not contain $v$. We proceed recursively on these subproblems until
there is just one $k$-subgraph to be listed. This method induces a
binary recursion tree, and all the $k$-subgraphs are listed when
reaching the leaves of the recursion tree.

As previously mentioned, in order to reach an output-sensitive
algorithm, we maintain a certificate that allows us to determine
efficiently if there exists at least one $k$-subgraph to be listed at
any point in the recursion.  Furthermore we select the vertex $v$,
that divides the problem into two subproblems, in a way that
facilitates the maintenance of the certificate.

\newpage

\section{Preliminaries}

Given a graph $G = (V,E)$ and $V' \subseteq V$, $G[V']=(V',E[V'])$
denotes the subgraph \emph{induced} by the vertices $V'$. Note that we
use the shorthand $E[V'] = \{ (u,v) \in E \mid u,v \in V'\}$. A
$k$-subgraph is a connected induced subgraph $G[V']$ such that
$|V'|=k$ and $G[V']$ is connected. We denote by $\setofisg_k(G)$ the
set of $k$-subgraphs in $G$. Formally: $$\setofisg_k(G) = \{ G[V']
\mid V'\subseteq V,~|V'|=k \text{ and } G[V'] \text{ is connected}
\}$$ For a vertex $u \in V$, $N(u)$ denotes the \emph{neighborhood}
of $u$. For a vertex set $S$, $N(S)$ is the union of neighborhoods of
the vertices in $S$ while $N^*(S)=N(S)\setminus{S}$. Additionally, for
a vertex set $S$, we define the cutset $C(S) \subseteq E$ such that
$(u,v) \in C(S)$ iif $u \in S$ and $v \in V[G]\setminus{S}$.

In this chapter we focus on the problem of efficiently listing
$k$-subgraphs.

\begin{problem}
	\label{prob:listinducedsubgraphs}
	Given a simple undirected and connected graph $G=(V,E)$ and an
	integer $k$, list all subgraphs $G[V'] \in \setofisg_k(G)$.
\end{problem}

An algorithm for this problem is said to be \emph{optimally
output-sensitive} if the time taken to solve
Problem~\ref{prob:listinducedsubgraphs} is $O(m+\sum_{G[V'] \in
\setofisg_k(G)}{|E[V']|})$. Thus, an improvement of the result of Avis
and Fukuda \cite{avis}. This algorithm is optimal in the sense that it
takes time proportional to reading the input plus listing the edges in
each of the induced subgraphs with $k$ vertices.  Noting that the
input graph $G$ is connected, we have that $O(m+\sum_{G[V'] \in
\setofisg_k(G)}{|E[V']|})=O(\sum_{G[V'] \in \setofisg_k(G)}{|E[V']|})$
since each edge $e \in E$ belongs to at least one $k$-subgraph.

\medskip

We dedicate the rest of the paper to proving the main result of the
chapter.

\begin{theorem}
	\label{thm:main}
	Problem~\ref{prob:listinducedsubgraphs} can be solved in
	$O(\sum_{G[V'] \in \setofisg_k(G)}{|E[V']|})$ time.
\end{theorem}

\section{Top-level algorithm}

Given a graph $G = (V,E)$, we solve
Problem~\ref{prob:listinducedsubgraphs} by using the natural approach
of taking an ordering of the vertices $v_1, v_2, \ldots v_n$ and
listing all the connected induced subgraphs with $k$ vertices that
include a vertex $v_i$ but do not include any vertex $v_j$ with $j<i$.

Additionally, we use a \emph{certificate $C$}: an adequate data
structure that, given a $k'$-subgraph $G[S]$ with $k' \le k$
vertices, ensures that there exists a connected vertex set with $k$
vertices that includes $S$. We denote this relation as $S \sqsubseteq
C$. As a matter of convenience, we use $C$ to represent both the data
structure and the connected set of vertices of size $k$. If there does
not exist such a vertex set with size $k$, $C$ includes largest possible
connected vertex set that includes $S$ (and thus $|C|<k$).

\medskip

To efficiently implement this approach, we make use of the following
operations:

\begin{itemize}
\item $\vdel(u)$ deletes a vertex $u \in V$ and all its incident
	edges.
\item $\undel(u)$ restores $G$ to its previous state, before operation
	$\vdel(u)$.
\item $\certificate(S)$ returns a certificate $C$ of $S$ (as defined
	above).
\end{itemize}

\begin{algorithm}[t]
\begin{algorithmic}[1]
	\FOR{$i=1,2,\ldots,n-1$}	\label{line:numvertices}
		\STATE $S:= \langle v_i \rangle$ 	\label{line:initS}
		\STATE $C:= \certificate(S)$		\label{line:initC}
		\IF{$|C|=k$} 				\label{line:checkC}
			\STATE $\listsubgraphs(S,C)$	\label{line:list}
			\STATE $\vdel(v_i)$		\label{line:del}
		\ELSE					\label{line:else}
			\FOR{$u \in C$}			\label{line:uinC}
				\STATE $\vdel(u)$	\label{line:alldel}
			\ENDFOR				\label{line:deluinC}
		\ENDIF	 \label{line:endif}
	\ENDFOR		 \label{line:endfor}
\end{algorithmic}
\caption{$ListSubgraphs$( $G=(V,E)$, $k$ )\label{alg:ListSubgraphs}}
\end{algorithm}

Let us now introduce Algorithm~\ref{alg:ListSubgraphs}. We start by
taking the vertices in $V$ in order (line~\ref{line:numvertices}) and,
for each $v_i$ in $V$, we initialize an ordered set of nodes $S$ with
$v_i$ (line~\ref{line:initS}) and compute its certificate $C$
(line~\ref{line:initC}). If $C$ contains $k$ vertices
(lines~\ref{line:checkC}-\ref{line:else}) then there is at least a
$k$-subgraph that includes $v_i$. In this case, we list all the
$k$-subgraphs that include $v_i$ (line~\ref{line:list}) and remove
$v_i$ from $G$ (line~\ref{line:del}) to avoid listing multiple times
the same $k$-subgraph. If the certificate $C$ contains less than $k$
vertices (lines~\ref{line:else}-\ref{line:endif}) then there is no
$k$-subgraph that includes $v_i$. Furthermore, by the definition of
$C$, there is no $k$-subgraph that includes each vertex $u \in C$ and
thus $u$ can be deleted from $G$
(lines~\ref{line:uinC}-\ref{line:deluinC}).

In the next section we present a recursive algorithm that, given an
ordered set of vertices $S$ (initially $S= \langle v_i \rangle$) and a
certificate $C$ of $S$, lists all the $k$-subgraphs that include $S$
and thus implements $\listsubgraphs(S,C)$ (line~\ref{line:list}).

\begin{lemma}
	\label{lem:toplev_correctness}
	Algorithm~\ref{alg:ListSubgraphs} lists each $k$-subgraphs in
	$G$ once and only once.
\end{lemma}
\begin{proof}
	It follows directly from the definition of certificate that
	for each vertex $v_i \in V$, the algorithm calls
	$\listsubgraphs(\langle v_i \rangle ,C)$ if $v_i$ belongs to
	at least one $k$-subgraph. On line~\ref{line:del}, $v_i$ is
	removed and thus no other $k$-subgraph that includes $v_i$ is
	listed. Furthermore, if $v_i$ does not belong to any
	$k$-subgraph, then none of the vertices in its connected
	belong to a $k$-subgraph and thus can also be deleted
	(lines~\ref{line:uinC}-\ref{line:deluinC}).
\end{proof}

\section{Recursion}

Let us now focus on the problem of listing all the $k$-subgraphs that
include a vertex $v_i$. In this section we use the follow operations
as black boxes and we will detail them in the next sections where their
implementation is fundamental to the analysis of the time complexity.

\begin{itemize}
\item $\choosevertex(S,C)$ given a connected vertex set $S$ and
	certificate $C$, returns a vertex $v \in N^*(S)$. We will use
	the certificate $C$ to select a vertex $v$ that facilitates
	updating $C$.
\item $\cleft(C,v)$ updates the certificate $C$ such that $C$ is a
	certificate of $S \cup \{v\}$. Returns a data structure $I$ with
	the modifications made in $C$.
\item $\cright(C,v)$ updates the certificate $C$ in the case of
	removal of $v$ from graph $G$. Note that, if $v$ does not
	belong to the certificate C, no modifications are needed.
	Returns a data structure $I$ with the modifications made in
	$C$.
\item $\restore(C,I)$ restores the certificate $C$ to its previous
	state, prior to the modifications in $I$.
\end{itemize}

\begin{algorithm}[t]
\begin{algorithmic}[1]
		\IF{$|S|=k$}	\label{line:ifbase}
			\STATE $\print(E[S])$ \label{line:print}
			\STATE $\return$ \label{line:return}
		\ENDIF	\label{line:endifbase}
		\STATE $v := \choosevertex(S)$ \label{line:choose}
		\smallskip
		\STATE $I := \cleft(C,v)$ \label{line:cleft}
		\STATE $\listsubgraphs(S \cup \{v\}, C)$ \label{line:lrecursion}
		\STATE $\restore(C,I)$ \label{line:restore1}

		\smallskip
		\STATE $\vdel(v)$ \label{line:vdel}
		\STATE $I := \cright(C,v)$ \label{line:cright}
		\IF { $|C|=k$ } \label{line:isunary} 
			\STATE $\listsubgraphs(S, C)$ \label{line:rrecursion}
		\ENDIF
		\STATE $\undel(v)$ \label{line:undel}
		\STATE $\restore(C,I)$ \label{line:restore2}
\end{algorithmic}
\caption{$\listsubgraphs(S, C)$  \hfill Invariant: $S$ is connected
and $S \sqsubseteq C$ \label{alg:ListSubgraphs2}}
\end{algorithm}

Algorithm~\ref{alg:ListSubgraphs2} follows closely the binary
partition method while maintaining the invariant that $C$ is a
certificate of $S$. This is of central importance to reach an
efficient algorithm since it allows us to avoid recursive calls that
do not output any $k$-subgraphs.

Given a connected vertex set $S$, we take a vertex $v \in N^*(S)$
(line~\ref{line:choose}). We then perform \emph{leftward branching}
(lines~\ref{line:cleft}-\ref{line:restore1}), recursively listing all
the $k$-subgraphs in $G$ that include the vertices in $S$ and the
vertex $v$. Note that by the definition of $\choosevertex(S)$, $v$ is
connected to $S$ and thus we maintain the invariant that vertices in
$S$ are connected. Since $S$ is ordered, we define the operation $S \cup
\{v\}$ to append vertex $v$ to the end of $S$ and thus the vertices in
$S$ are kept in order of insertion. We use the operation $\cleft(S,v)$
(line~\ref{line:cleft}) to maintain the invariant $S \sqsubseteq C$ in
the recursive call (line~\ref{line:lrecursion}). In line
\ref{line:restore1} we restore $C$ to its state before the recursive
call.

Successively, we perform \emph{rightward branching}
(lines~\ref{line:vdel}-\ref{line:restore2}), listing all $k$-subgraphs
that include the vertices in $S$ but \emph{do not} include vertex $v$.
After removing $v$ from the graph (line~\ref{line:vdel}), we update the
certificate $C$ (line~\ref{line:cright}). If it was possible to update
the certificate (i.e. $|C|=k$), we recursively list the $k$-subgraphs
that include $S$ in the residual graph $G - \{v\}$
(line~\ref{line:rrecursion}). When returning from the recursive call,
we restore the graph (line~\ref{line:undel}) and the certificate
(line~\ref{line:restore2}) to its previous state.

When the vertex set $S$ has size $k$ we reach the \emph{base case} of
the recursion (lines~\ref{line:ifbase}-\ref{line:endifbase}). Since
$S$ is connected, it induces a $k$-sugraph $G[S]$ which we output
(line~\ref{line:print}). Note that in this case we do not proceed with
the recursion (line~\ref{line:return}).

\begin{lemma}
	\label{lem:leftbranch}
	Given the existence of a certificate
	$C$ of $S$ and $v = \choosevertex(S,C)$, there exists a
	$k$-subgraph that includes $S\cup\{v\}$ and thus a left branch
	always produces one or more $k$-subgraphs.
\end{lemma}
\begin{proof}
	It follows from the definition of certificate that vertices in
	$S$ are in a connected component of size at least $k$. Noting
	that in a left branch we do not remove a vertex, $S = S+\{v\}$
	is still in a connected component of size at least $k$. Thus,
	the recursive call on line $\ref{line:lrecursion}$ will list
	at least one $k$-subgraph.
\end{proof}

\begin{lemma}
	\label{lem:invariant}
	At each recursive call $\listsubgraphs(S,C)$, we maintain the
	following invariants: (i) $S$ is connected, (ii) $C$ is a
	valid certificate of $S$ and (iii) $\listsubgraphs(S,C)$ is
	only invoked if there exists a $k$-subgraph that includes $S$.
\end{lemma}
\begin{proof}
	(i)~By the definition of operation $\choosevertex(S,C)$, $v
	\in N^*(S)$ (line~\ref{line:choose}) and thus, on the
	recursive call of line~\ref{line:lrecursion} (corresponding to
	a left branch), $S \cup \{v\}$ is still connected. On the
	recursive call of line \ref{line:rrecursion} (corresponding to
	a right branch), $S$ is not altered. Thus we maintain the
	invariant that $S$ is connected.  (ii)~Before both recursive
	calls of lines~\ref{line:choose} and~\ref{line:lrecursion},
	the operations that maintain the certificate are called (lines
	\ref{line:cleft} and \ref{line:cright}). After returning from
	the recursion the certificate is restored to its previous
	state (lines \ref{line:restore1} and \ref{line:restore2}) and
	thus the invariant $S \sqsubseteq C$ is maintained.  (iii)~By
	Lemma~\ref{lem:leftbranch}, on the recursive call of
	line~\ref{line:lrecursion} the invariant is trivially
	maintained. On the case of right branching, recursive call of
	line~\ref{line:rrecursion}, we verify that $S$ still belongs
	to a connected component of size at least $k$
	(line~\ref{line:isunary}).
\end{proof}

\begin{lemma}
	\label{lem:correctness_sub}
	A call to $\listsubgraphs(\langle v_i \rangle, C)$ lists the
	$k$-subgraphs that include $v_i$.
\end{lemma}
\begin{proof}
	A call to $\listsubgraphs(S, C)$ recursively lists all
	$k$-subgraphs that include $v \in N^*(S)$
	(line~\ref{line:rrecursion}). By Lemma~\ref{lem:invariant} we
	maintain a certificate $C$ that allows us to verify if there
	exist any $k$-subgraphs that include $S$ in the residual graph
	$G - \{v\}$ (line~\ref{line:isunary}). If that is the case, we
	invoke the right branch of the recursion
	(line~\ref{line:rrecursion}). When $|S|=k$, we reach the base
	case of the recursion and $S$ induces a $k$-subgraph that we
	output (line~\ref{line:print}). Thus, we list all the
	$k$-subgraphs. Noting that we remove vertex $v$ on the right
	branch of the recursion and that the set of $k$-subgraphs that
	include $v$ is disjoint from the set of subgraphs that does
	not include $v$, the same $k$-subgraph is not listed twice.
\end{proof}

\section{Amortization strategy}

Let us consider the recursion tree $R_i$ of $\listsubgraphs(S, C)$. We
denote a node $r \in R_i$ by the arguments $\langle S,C \rangle$ of
$\listsubgraphs(S, C)$. At the root $r_i = \langle S,C \rangle$ of
$R_i$, $S=\langle v_i \rangle$.  A \emph{leaf} node in the recursion
tree corresponds to the base case of the recursion (and thus has no
children). We say that an internal node $r \in R$ is \emph{binary} if
and only if both the left and right branch recursive calls are
executed.  The node $r$ is unary otherwise (by
Lemma~\ref{lem:leftbranch} left branching is always performed). 

We propose a time cost charging scheme for each node $r = \langle S, C
\rangle \in R_i$. This cost scheme drives the definition of the
certificate and implementation of the update operations.

\begin{equation}
	\label{eq:charging}
T(r)=
\begin{cases}
	O\left(|E[S]|\right) & \text{if $r$ is $\mathit{leaf}$} \\
	O\left(1\right) & \text{otherwise} \\
\end{cases}
\end{equation}

From the correctness of the algorithm, it follows that the leaves of
the recursion tree $R_i$ are in one-to-one correspondence to the set of
$k$-subgraphs that include $v_i$. Additionally, note that $E[S]$ are
the edges in those $k$-subgraphs.

\begin{lemma}
	\label{lem:total_cost}
	The total time over all the nodes in the recursion trees $R_1,
	\ldots, R_n$ given the charging scheme in
	Eq.~\ref{eq:charging} is $O(\sum_{G[V'] \in
	\setofisg_k(G)}{|E[V']|})$
\end{lemma}
\begin{proof}
	Since there is a one-to-one correspondence from the leaves in
	$R_1, \ldots, R_n$ to the subgraphs in $\setofisg_k(G)$ then
	the sum of the cost over all the leaves in the recursion tree
	is \mbox{$O(\sum_{G[V'] \in \setofisg_k(G)}{|E[V']|})$}.
	Furthermore, the total number of binary nodes in $R$ is
	bounded by the number of leaves.  Additionally, when
	considering any root to leaf path in $R$ there are at most $k$
	left branches.  This implies that there exist $O(k
	|\setofisg_k(G)|)$ internal nodes in $R$.  Since any graph
	$G[V'] \in \setofisg_k(G)$ is connected, $|E[V']| \ge k-1$ and
	$O(k |\setofisg_k(G)| + \sum_{G[V'] \in
	\setofisg_k(G)}{|E[V']|}) = O(\sum_{G[V'] \in
	\setofisg_k(G)}{|E[V']|})$.
\end{proof}

This target cost will drive the design of the operations on the graph
and the certificate. Although, on internal nodes of $R$, we are not
able to reach $O(1)$ time  we will show that the time cost can be
amortized according to Eq.~\ref{eq:charging}.

\section{Certificate}

A certificate $C$ of a $k'$-subgraph $G[S]$ is a data structure that uses
a truncated multi-source depth-first-search tree and classifies a
vertex $v \in N^*(S)$ as \emph{internal} (to the certificate $C$) or
\emph{external}. We say that a vertex $v \in N^*(S)$ is internal if $v
\in C$ and external otherwise. 

Given a $k'$-subgraph $G[S]$ with $S = ⟨v_1,v_2,\ldots,v_{k'}⟩$, the
multi-source DFS tree $C$ contains the vertices in $S$, which are
conceptually treated as a collapsed vertex with the ordered cutlist
$\hat C(S)$ as adjacency list. Recall that $\hat C(S)$ contains the
edges in the cutset $C(S)$ ordered by the rank of their endpoints in
$S$. If two edges have the same endpoint vertex $v \in S$, we use the
order in which they appear in the adjacency list of $v$ to break the
tie.  Also, all the vertices in $S$ are conceptually marked as visited
at the beginning of the DFS, so $u_j$ is never part of the DFS tree
starting from $u_i$ for any two distinct $u_i,u_j \in S$. Hence the
adopted terminology of multi-source.

The multi-source DFS tree is truncated when it reaches $k$ vertices. In
this case we say that $|C|=k$ and know that there exists a
$k$-subgraph that that includes all vertices in $S$. This is the main
goal of maintaining the certificate.

The data structure that implements the certificate $C$ of $S$ is a
partition $C = F \cup B$ where $F$ is a forest storing the edges of
the multi-source DFS whose both endpoints are in $V[C]-V[S]$. It
represents the remaining forest of the multi-source DFS tree when the
collapsed vertices in $S$ are removed. Moreover, $B$ is a vector that
allows to efficiently test if a vertex $v$ belongs to $C$.

\begin{enumerate}
	\item[(i)] We store the forest as a sorted doubly-linked list
		of the roots of the trees in the forest. The order of
		this list is that induced by $\hat C(S)$: a root $r$
		precedes a root $t$ if the edge in $(x,r) \in \hat
		C(S)$ precedes $(y,t) \in \hat C(S)$. We also keep a
		pointer to the last leaf, in the order of the
		multi-source DFS visit, of the last tree in $F$.
	\item[(ii)] We maintain a vector $B$ where $B[v]=\true$ if and
		only if $v\in C$.
\end{enumerate}

Furthermore, the certificate implements the following operations:

\begin{itemize}
\item Let $r$ be the root of the last tree in $F$. Operation
	$\promote(C)$ removes $r$ from $F$ and appends the children of
	$r$ to $F$ (and thus the children of $r$ become roots of trees
	in $F$).
\item Let $l$ be the last leaf of the last tree $T$ in $F$ (in the
	order of the DFS visit). Operation $\removelastleaf(C)$
	removes $l$ from $T$. Note that $l$ can also the last root in
	$F$, in this case this root is removed since $T$ becomes
	empty.
\end{itemize}

\begin{lemma}
	\label{lem:op_promote_removelastleaf}
	Operations $\promote(C)$ and $\removelastleaf(C)$ can be
	performed in $O(1)$ time.
\end{lemma}
\begin{proof}
	Operation $\promote(C)$. By maintaining the trees in $F$ using
	the classic \emph{first-child next-children} encoding of
	trees, it is possible to take the linked list of children of
	$r$ and append it to $F$ in constant time.
	Operation $\removelastleaf(C)$ can also be implemented in
	constant time since we maintain a pointer to the last leaf $l$
	and only have to remove it from the certificate.
\end{proof}

As we will show in Section~\ref{subsec:maintain-certificate}, this
implementation of the certificate allows it to be efficiently updated.
Additionally, we show that it can be computed in time proportional to
the number of edges in the certificate $C$.

\begin{lemma}
	\label{lem:op_certificate}
	Operation $C = \certificate(\langle v_i \rangle)$ can be
	performed in $O(|E[C]|)$ time.
\end{lemma}
\begin{proof}
	We start performing a classic DFS visit starting from $v_i$
	and adding the tree edges visited to $F$.
	When $k$ vertices are collected, we stop the visit. On vertex
	$x$, we visit edges $e=(x,y)$ that either collect a vertex $y$
	not previously visited and we set $B[y]=\true$ or connect $x$
	to a vertex $y$ already visited ($e$ is a back edge). In both
	cases $e \in E[C]$ and we do not visit such edge $e$ more than
	twice (once of each endpoint).
\end{proof}

\usetikzlibrary{snakes}
\tikzstyle{vertex}=[circle,fill=black,minimum size=5pt,inner sep=0pt]
\tikzstyle{edge} = [draw,thin,-]
\tikzstyle{triangle} = [triangle]

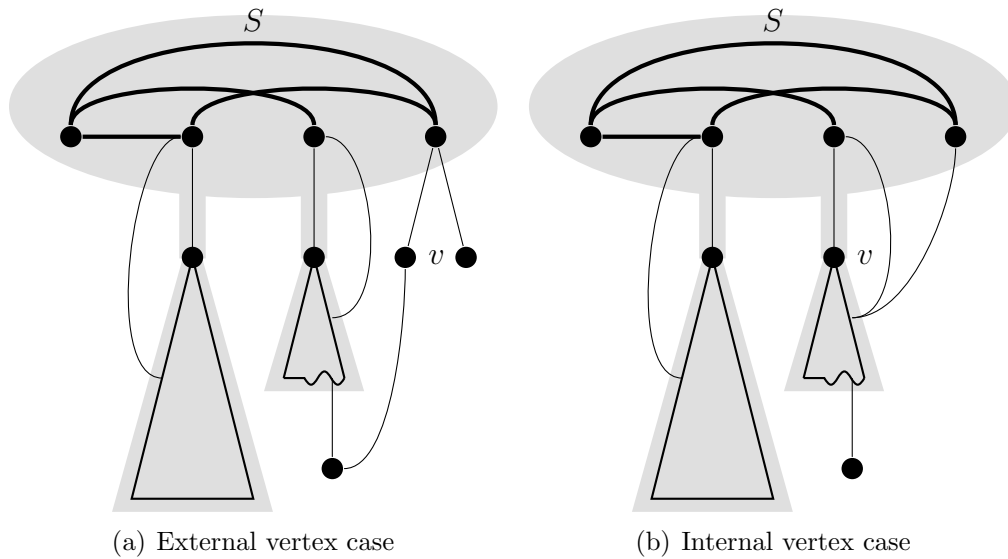
\begin{figure}[t]

\centering

\subfigure[External vertex case]{
\label{fig:external}

\begin{tikzpicture}[style=thick,scale=0.8] 

\draw[color=gray!25,fill=gray!25] (-3,-1.5) ellipse (4 and 1.5);
\draw[color=gray!25,fill=gray!25] (-4,-3.5) -- (-5.3,-8.2) -- (-2.7,-8.2) -- cycle;
\draw[color=gray!25,fill=gray!25] (-2,-3.5) -- (-2.8,-6.2) -- (-1.2,-6.2) -- cycle;
\draw[color=gray!25,fill=gray!25] (-4.2,-2) -- (-4.2,-4) -- (-3.8,-4) -- (-3.8,-2) -- cycle;
\draw[color=gray!25,fill=gray!25] (-2.2,-2) -- (-2.2,-4) -- (-1.8,-4) -- (-1.8,-2) -- cycle;

\node[vertex] (a) at (-6,-2) {$a$};
\node[vertex] (b) at (0,-2) {$a$};
\node[vertex] (c) at (-4,-2) {$a$};
\node[vertex] (g) at (-2,-2) {$a$};
\node[vertex] (k) at (-2,-4) {$a$};
\node[vertex] (m) at (-4,-4) {$a$};

\path[edge,ultra thick] (a) .. controls +(up:2cm) and +(up:2cm) .. node [auto] {$S$} (b);
\path[edge,ultra thick] (a) -- (c);
\path[edge,ultra thick] (b) .. controls +(up:1cm) and +(up:1cm) .. (c);
\path[edge,ultra thick] (a) .. controls +(up:1cm) and +(up:1cm) .. (g);

\draw (-4,-4) -- (-5,-8) -- (-3,-8) -- cycle;

\draw (-2,-4) -- (-2.5,-6);
\draw[snake=coil,segment aspect=0.0] (-1.5,-6) -- (-2.5,-6);
\draw (-1.5,-6) -- (-2,-4);

\node[vertex,label=right:$v$] (d) at (-0.5,-4) {$a$};
\node[vertex] (e) at (0.5,-4) {$a$};
\node[vertex] (f) at (-1.7,-7.5) {$a$};

\path[edge] (b) -- (d);
\path[edge] (b) -- node [right] {} (e);
\path[edge] (d) .. controls +(down:1cm) and +(right:1cm) .. (f);
\path[edge] (c) -- (-4,-4);
\path[edge] (g) -- (-2,-4);

\path[edge] (-1.7,-6) -- (f);
\path[edge] (-1.7,-5) .. controls +(right:1cm) and +(right:1cm) .. (g);
\path[edge] (-4.5,-6) .. controls +(left:1cm) and +(left:1cm) .. (c);
\end{tikzpicture} 
}
%
\subfigure[Internal vertex case]{
\label{fig:backedge}

\begin{tikzpicture}[style=thick,scale=0.8] 

\draw[color=gray!25,fill=gray!25] (-3,-1.5) ellipse (4 and 1.5);
\draw[color=gray!25,fill=gray!25] (-4,-3.5) -- (-5.3,-8.2) -- (-2.7,-8.2) -- cycle;
\draw[color=gray!25,fill=gray!25] (-2,-3.5) -- (-2.8,-6.2) -- (-1.2,-6.2) -- cycle;
\draw[color=gray!25,fill=gray!25] (-4.2,-2) -- (-4.2,-4) -- (-3.8,-4) -- (-3.8,-2) -- cycle;
\draw[color=gray!25,fill=gray!25] (-2.2,-2) -- (-2.2,-4) -- (-1.8,-4) -- (-1.8,-2) -- cycle;

\node[vertex] (a) at (-6,-2) {$a$};
\node[vertex] (b) at (0,-2) {$a$};
\node[vertex] (c) at (-4,-2) {$a$};
\node[vertex] (g) at (-2,-2) {$a$};
\node[vertex,label=right:$v$] (k) at (-2,-4) {$a$};
\node[vertex] (m) at (-4,-4) {$a$};

\path[edge,ultra thick] (a) .. controls +(up:2cm) and +(up:2cm) .. node [auto] {$S$} (b);
\path[edge,ultra thick] (a) -- (c);
\path[edge,ultra thick] (b) .. controls +(up:1cm) and +(up:1cm) .. (c);
\path[edge,ultra thick] (a) .. controls +(up:1cm) and +(up:1cm) .. (g);

\draw (-4,-4) -- (-5,-8) -- (-3,-8) -- cycle;

\draw (-2,-4) -- (-2.5,-6);
\draw[snake=coil,segment aspect=0.0] (-1.5,-6) -- (-2.5,-6);
\draw (-1.5,-6) -- (-2,-4);

\node[vertex] (f) at (-1.7,-7.5) {$a$};

\path[edge] (c) -- (-4,-4);
\path[edge] (g) -- (-2,-4);

\path[edge] (-1.7,-6) -- (f);
\path[edge] (-1.7,-5) .. controls +(right:1cm) and +(right:1cm) .. (g);
\path[edge] (-4.5,-6) .. controls +(left:1cm) and +(left:1cm) .. (c);
\path[edge] (-1.7,-5) .. controls +(right:1cm) and +(down:1cm) .. node [right] {} (b);

\end{tikzpicture} 
}

\caption{Choosing vertex $v \in N^*(S)$. The certificate $C$ is shadowed.}

\label{fig:choosevertex}
\end{figure}

Also with the goal of efficiently updating the certificate, we define
a specialization of the operation $v = \choosevertex(S,C)$
(illustrated in Figure~\ref{fig:choosevertex}). Recall that the only
requirement for the correctness of the algorithm is that $v \in
N^*(S)$.  By using the classification of $v$ as internal or external,
we are able to select a vertex that allows efficiently updating the
certificate.

\begin{itemize}
	\item $\choosevertex(S,C)$ returns a vertex $v \in N^*(S)$
		such that: $v$ is external if such vertex $v \notin C$
		exists; otherwise return the last root of $F$ (in this
		case $v$ is internal).
\end{itemize}

\subsection{Maintaining the certificate}
\label{subsec:maintain-certificate}

Let us now show how to maintain the certificate in the time bounds
defined in Eq.~\ref{eq:charging}. In the case of operation
$\cleft(C,v)$ with $v = \choosevertex(S,C)$, we can update the
certificate for the left branch in constant time. When $v$ is an
external vertex, it is sufficient to remove the last leaf (in DFS
order) of the last tree of $F$ while appending vertex $v$ to $S$. If
$v$ is internal, and thus already part of $C$, the vertices in the
certificate remain the same and it is enough to update the data
structure that represents $C$.

\begin{lemma}
	\label{lem:cleft_time}
	Operation $I = \cleft(C,v)$ can be performed in $O(1)$ time.
\end{lemma}
\begin{proof}
	Given a certificate $C$ of $S$, it is possible to update $C$
	so that $C$ becomes the certificate of $S \cup \{v\}$ in
	$O(1)$ time. If $v$ is an external vertex (i.e.~$v \notin C$),
	when we append $v$ to $S$, the data structure $C$ is still
	valid but has one vertex too many (i.e.~it is a certificate
	that there exists a $k+1$-subgraph) and thus we call
	$\removelastleaf(C)$ to decrease the size by one while
	maintaining the property that $C$ is still a multi-source DFS
	tree. In the case that $v$ is an internal vertex, the vertices
	in the certificate $C$ are still the same but we have to
	update $F$ to reflect that the size of $V[S+\{v\}]-V[C]$ is
	now smaller. Since $v$ is the root of the last tree in $F$ it
	is enough to call operation $\promote(C)$ while maintaining
	the invariant that $C$ is a multi-source truncated DFS from
	$S$. The modifications to $C$ are stored in $I$ in order to
	restore the certificate $C$ to its previous state using
	operation $\restore(C,I)$.
\end{proof}

Performing the operation $\cright(C,v)$ when $v$ is external to the
certificate is also easy as the certificate $C$ remains valid after
the removal of vertex $v$.

\begin{lemma}
	\label{lem:cright_external_time}
	Operation $I = \cright(C,v)$ can be performed in $O(1)$ time when
	$v$ is external.
\end{lemma}
\begin{proof}
	Since $v \notin C$, removing the vertex $v$ has no impact in
	the certificate and therefore $C$ is still a valid certificate
	of $S$ in the graph $G - \{v\}$. In this case $I$ is empty.
\end{proof}

To perform the operation $\cright(C,v)$ when $v$ is internal is more
complex since we are removing from the graph $G$ a vertex that is part
of the certificate. By the definition of operation $\choosevertex(C)$,
this operation invalidates the last tree of the forest $F$ (other
trees have been completely explored by the DFS and, since the graph is
undirected, there are no cross edges between the trees in $F$).
Nevertheless, noting that there are no vertices in $N^*(S)$ that do
not belong to the certificate, we prove that we can recompute part of
the multi-source DFS tree in the budget defined in
Eq.~\ref{eq:charging}. When this recomputation is able to update the
certificate $C$ with $|C|=k$, we prove that we can charge its cost on
the leaf that outputs the $k$-subgraph $G[C]$. When the recomputation
leads to a $C$ with $|C|<k$, we prove that there still exists a leaf
where we can charge the cost of recomputing the certificate.
Furthermore, we prove that the same leaf is not charged twice and thus
we respect the cost scheme defined in Eq.~\ref{eq:charging}.

\begin{lemma}
	\label{lem:cright_internal_time} Operation $I = \cright(C,v)$,
	when $v$ is internal, can be performed in the time defined by
	Eq.~\ref{eq:charging}.
\end{lemma}
\begin{proof}
	In this operation we have removed vertex $v$ which, by
	definition of $\choosevertex(C)$, is the last root of the
	last tree $T$ of $F$. We note that, when removing $v$, the
	tree $T$ has been split into several branches $T'_1, T'_2,
	\ldots, T'_i$, one for each of the $i$ children of $v$ in $T$.
	Additionally, by the definition of $\choosevertex(C)$, we
	remark that there are no edges $e=(x,y)$ where $x \in S$ and
	$y \notin V[C]-V[S]$ and thus every edge leaving from a vertex in
	$S$ leads to a vertex in the forest. Furthermore, by the
	properties of the multi-source DFS tree, we know that the
	trees $T'_1,\ldots, T'_{i-1}$ are closed while the tree $T'_i$
	is possibly open (i.e.~can lead to vertices not previously in
	the certificate).

	We update the certificate by performing the multi-source DFS
	tree much like in Lemma~\ref{lem:op_certificate} but with
	small twist. (i) If we reach a vertex in a tree of $F$
	different from $T$ we reuse the portion of the DFS visit
	already performed. This tree of $F$ is still valid as there
	exist no cross edges between tree (a key property and the main
	reason why we use a DFS). Thus, we proceed w.l.o.g.~assuming
	that $F$ is composed only by its last tree $T$; (ii) If we
	reach a vertex in the trees $T'_1, \ldots, T'_{i}$ we keep
	performing the visit until $k$ vertices are visited or we have
	no edges left to explore. The modifications performed are kept
	in $I$ and when a vertex $x$ is included in (resp.~removed
	from) the certificate $B[x]$ is set to $\true$ (resp.~$\false$).
	
	The existence of a certificate $C$ with $|C|=k$, depends of
	how many of the trees $T'_1, \ldots, T'_{i-1}$ we are able to
	``recapture''. Only if we are able to reach the tree $T'_{i}$,
	we will be able to include extra vertices not previously in
	the certificate (to account for the removal of $v$ and
	possible the disconnection of $T'_1, \ldots, T'_{i-1}$). For
	the purposes of this proof, let us consider two different
	cases: (a) when the visit reaches $k$ vertices and, (b) when
	the visit does not reach $k$ vertices.  Note that in case (b),
	the right branch of the recursion tree will not be performed.

	\usetikzlibrary{snakes}
\tikzstyle{vertex}=[circle,fill=black,minimum size=5pt,inner sep=0pt]
\tikzstyle{edge} = [draw,thin,-]
\tikzstyle{triangle} = [triangle]

\begin{figure}[t]

\centering

\subfigure[Certificate $C$ before update]{

\begin{tikzpicture}[style=thick,scale=0.8] 

\draw[color=gray!25,fill=gray!25] (-3,-1.5) ellipse (4 and 1.5);
\draw[color=gray!25,fill=gray!25] (-2,-3.5) -- (-3.5,-6.5) -- (-0.5,-6.5) -- cycle;
\draw[color=gray!25,fill=gray!25] (-2.2,-2) -- (-2.2,-4) -- (-1.8,-4) -- (-1.8,-2) -- cycle;

\node[vertex] (a) at (-6,-2) {$a$};
\node[vertex] (b) at (0,-2) {$a$};
\node[vertex] (c) at (-4,-2) {$a$};
\node[vertex] (g) at (-2,-2) {$a$};
\node[vertex,label=right:$v$] (k) at (-2,-4) {$a$};
\node[vertex] (m) at (-2.5,-5) {$a$};
\node[vertex] (n) at (-1.5,-5) {$a$};

\path[edge,ultra thick] (a) .. controls +(up:2cm) and +(up:2cm) .. node [auto] {$S$} (b);
\path[edge,ultra thick] (a) -- (c);
\path[edge,ultra thick] (b) .. controls +(up:1cm) and +(up:1cm) .. (c);
\path[edge,ultra thick] (a) .. controls +(up:1cm) and +(up:1cm) .. (g);

\draw (m) -- (-3.15,-6.3);
\draw[segment aspect=0.0] (-2,-6.3) -- (-3.15,-6.3);
\draw (-2,-6.3) -- (m);

\draw (n) -- (-1.8,-6);
\draw[snake=coil,segment aspect=0.0] (-1,-6) -- (-1.8,-6);
\draw (-1,-6) -- (n);

\node[vertex] (f) at (-1,-7) {$a$};
\node[vertex] (r) at (-1,-8) {$a$};

\path[edge] (g) -- (-2,-4);

\path[edge] (k) -- (m);
\path[edge] (k) -- (n);
\path[edge] (-1,-6) -- (f);
\path[edge] (r) -- (f);
\path[edge] (m) .. controls +(left:1cm) and +(left:1cm) .. (g);
\path[edge] (n) .. controls +(right:1cm) and +(down:1cm) .. node [right] {} (b);

\end{tikzpicture} 
}
%
\subfigure[Certificate $C$ after update]{

\begin{tikzpicture}[style=thick,scale=0.8] 

\draw[color=gray!25,fill=gray!25] (-3,-1.5) ellipse (4 and 1.5);
\draw[color=gray!25,fill=gray!25] (-2,-4) -- (-3.1,-6.5) -- (-1,-6.5) -- cycle;
\draw[color=gray!25,fill=gray!25] (-2.2,-2) -- (-2.2,-5) -- (-1.8,-5) -- (-1.8,-2) -- cycle;

\draw[color=gray!25,fill=gray!25] (0,-4) -- (-0.8,-6.2) -- (0.8,-6.2) -- cycle;
\draw[color=gray!25,fill=gray!25] (-0.2,-2) -- (-0.2,-5) -- (0.2,-5) -- (0.2,-2) -- cycle;
\draw[color=gray!25,fill=gray!25] (0.3,-6.2) -- (0.3,-7.2) -- (0.8,-7.2) -- (0.8,-6.2) -- cycle;

\node[vertex] (a) at (-6,-2) {$a$};
\node[vertex] (b) at (0,-2) {$a$};
\node[vertex] (c) at (-4,-2) {$a$};
\node[vertex] (g) at (-2,-2) {$a$};
\node[vertex,color=lightgray,label=right:$v$] (k) at (-1,-4) {$a$};
\node[vertex] (m) at (-2,-5) {$a$};
\node[vertex] (n) at (-0,-5) {$a$};

\path[edge,ultra thick] (a) .. controls +(up:2cm) and +(up:2cm) .. node [auto] {$S$} (b);
\path[edge,ultra thick] (a) -- (c);
\path[edge,ultra thick] (b) .. controls +(up:1cm) and +(up:1cm) .. (c);
\path[edge,ultra thick] (a) .. controls +(up:1cm) and +(up:1cm) .. (g);

\draw (m) -- (-2.65,-6.3);
\draw[segment aspect=0.0] (-1.5,-6.3) -- (-2.65,-6.3);
\draw (-1.5,-6.3) -- (m);

\draw (n) -- (0.5,-6);
\draw[snake=coil,segment aspect=0.0] (0.5,-6) -- (-0.5,-6);
\draw (-0.5,-6) -- (n);

\node[vertex] (f) at (0.5,-7) {$a$};
\node[vertex] (r) at (0.5,-8) {$a$};

\path[edge,color=lightgray] (g) -- (k);
\path[edge,color=lightgray] (k) -- (m);
\path[edge,color=lightgray] (k) -- (n);
\path[edge] (0.5,-6) -- (f);
\path[edge] (r) -- (f);
\path[edge] (m) .. controls +(up:1cm) and +(down:1cm) .. (g);
\path[edge] (n) .. controls +(up:1cm) and +(down:1cm) .. node [right] {} (b);

\end{tikzpicture} 
}

\caption{Case (a) of $\cright(C,v)$ when $v$ is internal}

\label{fig:case_a_internal}
\end{figure}

\begin{figure}[t]

\centering

\subfigure[Certificate $C$ before update]{

\begin{tikzpicture}[style=thick,scale=0.8] 

\draw[color=gray!25,fill=gray!25] (-3,-1.5) ellipse (4 and 1.5);
\draw[color=gray!25,fill=gray!25] (-2,-3.5) -- (-3.5,-6.5) -- (-0.5,-6.5) -- cycle;
\draw[color=gray!25,fill=gray!25] (-2.2,-2) -- (-2.2,-4) -- (-1.8,-4) -- (-1.8,-2) -- cycle;

\node[vertex] (a) at (-6,-2) {$a$};
\node[vertex] (b) at (0,-2) {$a$};
\node[vertex] (c) at (-4,-2) {$a$};
\node[vertex] (g) at (-2,-2) {$a$};
\node[vertex,label=right:$v$] (k) at (-2,-4) {$a$};
\node[vertex] (m) at (-2.5,-5) {$a$};
\node[vertex] (n) at (-1.5,-5) {$a$};

\path[edge,ultra thick] (a) .. controls +(up:2cm) and +(up:2cm) .. node [auto] {$S$} (b);
\path[edge,ultra thick] (a) -- (c);
\path[edge,ultra thick] (b) .. controls +(up:1cm) and +(up:1cm) .. (c);
\path[edge,ultra thick] (a) .. controls +(up:1cm) and +(up:1cm) .. (g);

\draw (m) -- (-3.15,-6.3);
\draw[segment aspect=0.0] (-2,-6.3) -- (-3.15,-6.3);
\draw (-2,-6.3) -- (m);

\draw (n) -- (-1.8,-6);
\draw[snake=coil,segment aspect=0.0] (-1,-6) -- (-1.8,-6);
\draw (-1,-6) -- (n);

\node[vertex] (f) at (-1,-7) {$a$};
\node[vertex] (r) at (-1,-8) {$a$};

\path[edge] (g) -- (-2,-4);

\path[edge] (k) -- (m);
\path[edge] (k) -- (n);
\path[edge] (-1,-6) -- (f);
\path[edge] (r) -- (f);
\path[edge] (n) .. controls +(right:1cm) and +(down:1cm) .. node [right] {} (b);

\end{tikzpicture} 
}
%
\subfigure[Certificate $C$ after update]{

\begin{tikzpicture}[style=thick,scale=0.8] 

\draw[color=gray!25,fill=gray!25] (-3,-1.5) ellipse (4 and 1.5);

\draw[color=gray!25,fill=gray!25] (0,-4) -- (-0.8,-6.2) -- (0.8,-6.2) -- cycle;
\draw[color=gray!25,fill=gray!25] (-0.2,-2) -- (-0.2,-5) -- (0.2,-5) -- (0.2,-2) -- cycle;
\draw[color=gray!25,fill=gray!25] (0.3,-6.2) -- (0.3,-8.2) -- (0.8,-8.2) -- (0.8,-6.2) -- cycle;

\node[vertex] (a) at (-6,-2) {$a$};
\node[vertex] (b) at (0,-2) {$a$};
\node[vertex] (c) at (-4,-2) {$a$};
\node[vertex] (g) at (-2,-2) {$a$};
\node[vertex,color=lightgray,label=right:$v$] (k) at (-1,-4) {$a$};
\node[vertex] (m) at (-2,-5) {$a$};
\node[vertex] (n) at (-0,-5) {$a$};

\path[edge,ultra thick] (a) .. controls +(up:2cm) and +(up:2cm) .. node [auto] {$S$} (b);
\path[edge,ultra thick] (a) -- (c);
\path[edge,ultra thick] (b) .. controls +(up:1cm) and +(up:1cm) .. (c);
\path[edge,ultra thick] (a) .. controls +(up:1cm) and +(up:1cm) .. (g);

\draw (m) -- (-2.65,-6.3);
\draw[segment aspect=0.0] (-1.5,-6.3) -- (-2.65,-6.3);
\draw (-1.5,-6.3) -- (m);

\draw (n) -- (0.5,-6);
\draw[snake=coil,segment aspect=0.0] (0.5,-6) -- (-0.5,-6);
\draw (-0.5,-6) -- (n);

\node[vertex] (f) at (0.5,-7) {$a$};
\node[vertex] (r) at (0.5,-8) {$a$};

\path[edge,color=lightgray] (g) -- (k);
\path[edge,color=lightgray] (k) -- (m);
\path[edge,color=lightgray] (k) -- (n);
\path[edge] (0.5,-6) -- (f);
\path[edge] (r) -- (f);
\path[edge] (n) .. controls +(up:1cm) and +(down:1cm) .. node [right] {} (b);

\end{tikzpicture} 
}

\caption{Case (b) of $\cright(C,v)$ when $v$ is internal}

\label{fig:case_b_internal}
\end{figure} 

	\begin{itemize}
	\item[(a)] Since the updated certificate $C$ includes $k$
		vertices and noting that we did not explore edges
		external to $C$, this operation takes $O(|E[C]|)$ time.
		In this case we perform a rightward branch. We charge
		this cost to the leaf corresponding to the output of
		$k$-subgraph $G[C]$. Given the order $\choosevertex(C)$
		selects the vertices to recurse, the path in the
		recursion tree $R$ from the current node of recursion
		$r_x \in R$ to the leaf corresponding to the output of
		$C$ is a sequences of left branches with $v$ being an
		internal vertex interleaved with right branches with
		$v$ being an external vertex. Noting that from $r_x$
		to this leaf we do not perform any right branch where
		$v$ is internal, any leaf is charged at most once.
		This case is illustrated in
		Figure~\ref{fig:case_a_internal}.

	\item[(b)] When $|C|<k$ the visit still takes $O(|E[C]|)$ time
		but, since $G[C]$ is not a $k$-subgraph, we are not able
		to amortize it in the same way as in (a). Let us
		assume that the vertices $\alpha$ of the trees $T'_1,
		\ldots T'_i$ were ``recaptured'', the vertices $\beta$
		of the trees $T'_1, \ldots T'_i$ were \emph{not}
		``recaptured'' and that the vertices $\gamma$ not
		previously in $C$ were visited.  Note that $|\alpha| +
		|\gamma| + |S| < k$.  We are able to charge the
		$O(|E[C]|)$ cost to a leaf (and thus $k$-subgraph) $X$
		that includes every edge in $E[C]$. The $k$-subgraph
		$X$ is formed in the following way:
		(i) $X$ includes the vertices in $S$;
		(ii) $X$ includes the vertex $v$;
		(iii) includes the vertices in $\alpha$ and $\gamma$;
		(iv) Since $v$ is part of $X$, all the vertices in $\beta$
		are connected to the vertices in $X$ and thus we
		include a sufficient number of those in order to have
		$|X|=k$. Note that $E[C] \subset E[X]$. When $\gamma
		\neq \emptyset$, the leaf corresponding to $X$ is in a
		right branch of one of the descendants of the left
		branch of the current node of the recursion and thus
		is not charged more than once. The limit case when
		$\beta=\emptyset$ implies that no new nodes were
		visited in the update of the certificate
		(i.e.~$\gamma=\emptyset$) and thus the certificate $C$
		previous to the update is the only leaf in the current
		branch of the recursion. In this limit case is easy to
		avoid recomputing the certificate in any descendant.
		Therefore, the cost of operation $\cright(C,v)$ can be
		charged on the leaves of the recursion according to
		the cost defined in Eq.~\ref{eq:charging}.
	\end{itemize}
\end{proof}

\begin{lemma}
	\label{lem:restore_time}
	Operation $\restore(C,I)$, can be
	performed in the time defined by Eq.~\ref{eq:charging}.
\end{lemma}
\begin{proof}
	We use standard data structures (i.e. linked lists) for the
	representation of certificate $C$. There exist persistent
	versions of these data structures that maintain a stack of
	modifications applied to them and that can restore its
	contents to their previous states. Given the modifications in
	$I$, these data structures take $O(|I|)$ time to restore the
	previous version of $C$. Given that the time taken to fill the
	modification in $I$ done on operations $\cleft(C,v)$ and
	$\cright(C,v)$ respects Eq.~\ref{eq:charging}, this implies
	that operation $\restore(C,I)$ can be performed in the same
	time bound.
\end{proof}

\section{Other operations}

In order to implement operation $\choosevertex(S,C)$ within the cost
established in Eq.~\ref{eq:charging}, we have to avoid visiting the
edges which are internal to the certificate $C$ multiple times. The
naive approach of finding the first edge that connects the vertices in
$S$ to an external vertex can take up to $O(k^2)$ time, as we possibly
have to visit all edges internal to $C$.

Let us consider the set of nodes $N$ in the recursion tree $R_i$ that
have the same vertices in the certificate as the node $r = (S,C) \in
R_i$.  The recursion nodes in $N$ correspond to a sequence of right
branches in external vertices and left branches on internal vertices.

We make use of a data structure which we call \emph{parking lists},
where we place an edge $e$ internal to certificate $C$ when visiting
it for the first time. Then, we remove $e$ from the graph $G$ to avoid
visiting it additional times. This allows us, for the recursion nodes
in $N$ whose vertices in the certificate $C$ are the same, to only
visit each edge in $E[C]$ once. We then prove that this cost can be
amortized according to Eq.~\ref{eq:charging}.  The \emph{parking
lists} are a set of lists $P_0, P_1, \ldots, P_n$ corresponding to
each $v_i \in V$. When implementing operation $\choosevertex(S,C)$, if
we visit an edge $e=(x,y)$ from the graph $G$ such that $x,y \in C$ we
append $e$ to $P_x$ and remove it from $\adj(x)$. Note that, for
the purpose of simplification, this modification is restricted to
operation $\choosevertex(S,C)$ and for other operations, such as the
ones necessary to maintain the certificate, internal edges are still
present (i.e.~the adjacency list of vertex $u \in V$ is $\adj(u) \cup
P_u$).

\begin{lemma}
	\label{lem:choose_time}
	Operation $\choosevertex(S,C)$ can be performed within the
	time cost defined in Eq~\ref{eq:charging}.
\end{lemma}
\begin{proof}
	For each vertex $x \in S$ we take each edge $e = (x,y) \in
	\adj(x)$. When vertex $y \notin C$, we return $y$. This can be
	tested in constant time by checking if vector $B[y]=\false$.
	When all edges are visited without finding such vertex $y$, we
	return the root of the last tree in forest $F$.

	During this process, when we take and edge $e=(x,y)$ with $y
	\in C$, we remove $e$ from $adj(x)$ and set $P_x := P_x \cup
	\{e\}$. Furthermore, on the implementation of operation
	$\removelastleaf(C)$, when removing vertex $l$ from the
	certificate we update $\adj(l) = \adj(l) \cup P_x$. This
	operation can be done in constant time.

	Consider the set of nodes $N$, forming a path in the recursion
	tree $R_i$, where the vertices in their certificates $C$ are
	the same. By using this strategy, we do not visit any edge in
	$E[C]$ more than twice (once for each endpoint) and we can
	charge this cost to the leaf corresponding to the output of
	$C$. Since there is only one such path for the set of vertices
	in certificate $C$, the cost of this operation can be charged
	according to Eq.~\ref{eq:charging}.
\end{proof}

We are now left with the analysis of operations~$\del(v)$ and
$\undel(v)$ in a node $\langle S, C \rangle$ of the recursion tree
$R_i$. One thing to note is that the adjacency list $\adj(v)$ can have
up to $n-1$ edges and, at first sight, we cannot afford to pay this
cost. A key insight is that each of the edges in $\adj(v)$ is
connected to the certificate $C$ and thus it implies that $e$ is part
of some $k$-subgraphs. We now prove that we can amortize the cost of
$\del(v)$ in those subgraphs.

\begin{lemma}
	\label{lem:del_time}
	Operations~$\del(v)$ and $\undel(v)$ can be performed within
	the time cost defined in Eq~\ref{eq:charging}.
\end{lemma}
\begin{proof}
	Consider the deletion of each edge incident in $v$,  $e =
	(v,w) \in \adj(v)$.  Note that any vertex $w$ is connected to
	$S$ though $v$ and thus there will be internal nodes $\langle
	S',C' \rangle$ of the recursion tree with $S' = S \cup \{v\}
	\cup \{w\}$. We charge the $O(1)$ cost of removing edge $e$ on
	the first internal node $\langle S',C' \rangle \in R_i$.
	Since, for a given $S$ we remove $(v,w)$ only once, this
	internal node is not charged multiple times. Operation
	$\undel(v)$ can be performed in the same time bound.
\end{proof}

\begin{lemma}
	\label{lem:complete_parking_lists_time}
	We can maintain the parking lists in such way that they
	include every edge internal to the certificate within the time
	cost defined in Eq~\ref{eq:charging}.
\end{lemma}
\begin{proof}
	When we add one external vertex $v$ to the certificate, we can
	visit every edge incident in $v$ and add it to $P_v$ if it is
	internal to the certificate. By a similar argument to the
	one used in Lemma~\ref{lem:del_time}, every edge $e = (v,w)
	\in \adj(v)$ can be visited. We can charge $O(1)$ cost of
	visiting $e$ to the first internal node $\langle S',C' \rangle
	\in R_i$ where internal node $\langle S',C' \rangle$ is the
	first node of the recursion tree with $S' = S \cup \{v\} \cup
	\{w\}$. Since $v$ is external and $v$ is removed on the right
	branch of the recursion, the edges incident in $v$ are not
	visited more than once.
	When $v$ is internal, every edge internal to the certificate
	has already been discovered by operation $\choosevertex(C)$.
\end{proof}

Having analyzed each operation performed in a node of the recursion
tree, we are able to prove the following lemma.

\begin{lemma}
	\label{lem:recursion_tree_node_time}
	The operations a node $r$ of the recursion tree $R$ can be
	perform in the time defined in Eq.~\ref{eq:charging}:
	$$
	T(r)=
		\begin{cases}
			O\left(|E[S]|\right) & \text{if $r$ is $leaf$} \\
			O\left(1\right) & \text{otherwise} \\
		\end{cases}
	$$
\end{lemma}
\begin{proof}
	Directly from Lemmas~\ref{lem:choose_time},
	\ref{lem:cleft_time}, \ref{lem:cright_external_time},
	\ref{lem:cright_internal_time}, \ref{lem:del_time} and
	\ref{lem:restore_time}, each operation performed on an
	internal node of the recursion tree can be performed within
	the time cost defined. Furthermore, on the leaves of the
	recursion tree, we only have to output the edges on $E[S]$,
	which can be recorded in the parking lists by
	Lemma~\ref{lem:complete_parking_lists_time}.
\end{proof}

With Lemma~\ref{lem:recursion_tree_node_time} and
Lemma~\ref{lem:total_cost} we are able to prove
Lemma~\ref{lem:algo_time} which completes the proof of
Theorem~\ref{prob:listinducedsubgraphs}.

\begin{lemma}
	\label{lem:algo_time}
	Algorithm~\ref{alg:ListSubgraphs} lists all $k$-subgraphs in
	$O(\sum_{G[V'] \in \setofisg_k(G)}{|E[V']|})$
\end{lemma}
\begin{proof}
	By Lemmas~\ref{lem:total_cost} and
	\ref{lem:recursion_tree_node_time}, the total time spent in
	$\listsubgraphs(S, C)$ (line~\ref{line:list}) is
	$O(\sum_{G[V'] \in \setofisg_k(G)}{|E[V']|})$. By
	Lemma~\ref{lem:op_certificate}, the cost of each call to $C =
	\certificate(S)$ (line~\ref{line:initC}) is $O(|E[C]|)$. In
	the case $|C|=k$ (lines~\ref{line:checkC}-\ref{line:del}),
	this cost can be amortized in the leaf corresponding to the
	output of $k$-subgraph $C$. When $|C|<k$
	(lines~\ref{line:else}-\ref{line:deluinC}), we remove every
	vertex in $C$ and incident edges, thus the total cost over
	these operations is $O(m)$. Noting that the input graph $G$ is
	connected, each edge belongs to at least a $k$-subgraph and
	the total time taken by Algorithm~\ref{alg:ListSubgraphs} is
	$O(\sum_{G[V'] \in \setofisg_k(G)}{|E[V']|})$

\end{proof}

\chapter{Listing cycles and st-paths}
\label{chapter:cycles}

Listing all the simple cycles (hereafter just called cycles) in a
graph is a classical problem whose efficient solutions date back to
the early 70s. For a graph with $n$ vertices and $m$ edges, containing
$\eta$ cycles, the best known solution in the literature is given by
Johnson's algorithm~\cite{Johnson1975} and takes $O((\eta+1)(m+n))$
time. This solution is surprisingly not optimal for undirected graphs:
to the best of our knowledge, no theoretically faster solutions have
been proposed in almost 40 years.

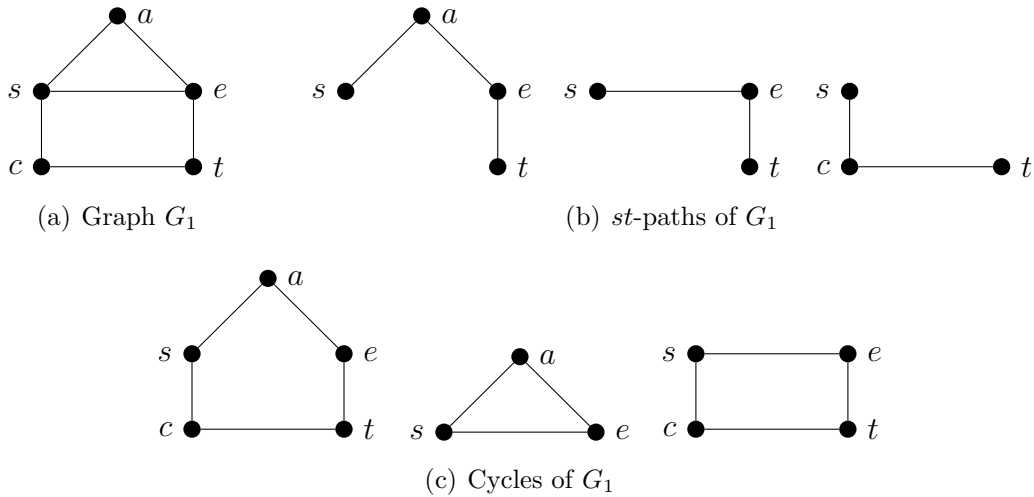
\begin{figure}[t]
\centering
\subfigure[Graph $G_1$]{
\begin{tikzpicture}[shorten >=1pt,->,scale=1.0]
  \tikzstyle{vertex}=[shape=circle,draw,thick,fill=black,minimum size=2pt,inner sep=2pt]
  \node[vertex,label=right:$a$] (A) at (0,0) {};
  \node[vertex,label=left:$s$] (B) at (-1,-1)   {};
  \node[vertex,label=left:$c$] (C) at (-1,-2)  {};
  \node[vertex,label=right:$t$] (D) at (1,-2)  {};
  \node[vertex,label=right:$e$] (E) at (1,-1)  {};
  \draw (A) -- (B) -- (C) -- (D) -- (E) -- (A) -- cycle;
  \draw (B) -- (E) -- cycle;
\end{tikzpicture}
} 
\quad
\subfigure[$st$-paths of $G_1$]{
\begin{tikzpicture}[shorten >=1pt,->,scale=1.0]
  \tikzstyle{vertex}=[shape=circle,draw,thick,fill=black,minimum size=2pt,inner sep=2pt]
  \node[vertex,label=right:$a$] (A) at (0,0) {};
  \node[vertex,label=left:$s$] (B) at (-1,-1)   {};
  \node[vertex,label=right:$t$] (D) at (1,-2)  {};
  \node[vertex,label=right:$e$] (E) at (1,-1)  {};
  \draw (B) -- (A) -- (E) -- (D) -- cycle;
\end{tikzpicture}
\begin{tikzpicture}[shorten >=1pt,->,scale=1.0]
  \tikzstyle{vertex}=[shape=circle,draw,thick,fill=black,minimum size=2pt,inner sep=2pt]
  \node[vertex,label=left:$s$] (B) at (-1,-1)   {};
  \node[vertex,label=right:$t$] (D) at (1,-2)  {};
  \node[vertex,label=right:$e$] (E) at (1,-1)  {};
  \draw (B) -- (E) -- (D) -- cycle;
\end{tikzpicture}
\begin{tikzpicture}[shorten >=1pt,->,scale=1.0]
  \tikzstyle{vertex}=[shape=circle,draw,thick,fill=black,minimum size=2pt,inner sep=2pt]
  \node[vertex,label=left:$s$] (B) at (-1,-1)   {};
  \node[vertex,label=left:$c$] (C) at (-1,-2)  {};
  \node[vertex,label=right:$t$] (D) at (1,-2)  {};
  \draw (B) -- (C) -- (D) -- cycle;
\end{tikzpicture}
}

\subfigure[Cycles of $G_1$]{
\begin{tikzpicture}[shorten >=1pt,->,scale=1.0]
  \tikzstyle{vertex}=[shape=circle,draw,thick,fill=black,minimum size=2pt,inner sep=2pt]
  \node[vertex,label=right:$a$] (A) at (0,0) {};
  \node[vertex,label=left:$s$] (B) at (-1,-1)   {};
  \node[vertex,label=left:$c$] (C) at (-1,-2)  {};
  \node[vertex,label=right:$t$] (D) at (1,-2)  {};
  \node[vertex,label=right:$e$] (E) at (1,-1)  {};
  \draw (A) -- (B) -- (C) -- (D) -- (E) -- (A) -- cycle;
\end{tikzpicture}
\begin{tikzpicture}[shorten >=1pt,->,scale=1.0]
  \tikzstyle{vertex}=[shape=circle,draw,thick,fill=black,minimum size=2pt,inner sep=2pt]
  \node[vertex,label=right:$a$] (A) at (0,0) {};
  \node[vertex,label=left:$s$] (B) at (-1,-1)   {};
  \node[vertex,label=right:$e$] (E) at (1,-1)  {};
  \draw (A) -- (B) -- (E) -- (A) -- cycle;
\end{tikzpicture}
\begin{tikzpicture}[shorten >=1pt,->,scale=1.0]
  \tikzstyle{vertex}=[shape=circle,draw,thick,fill=black,minimum size=2pt,inner sep=2pt]
  \node[vertex,label=left:$s$] (B) at (-1,-1)   {};
  \node[vertex,label=left:$c$] (C) at (-1,-2)  {};
  \node[vertex,label=right:$t$] (D) at (1,-2)  {};
  \node[vertex,label=right:$e$] (E) at (1,-1)  {};
  \draw (B) -- (C) -- (D) -- (E) -- cycle;
  \draw (B) -- (E) -- cycle;
\end{tikzpicture}
}

\caption{Example graph $G_1$, its st-paths and cycles}
\label{fig:cyclesstpathsexample}

\end{figure}

\subsubsection{Results}  

Originally introduced in \cite{ferreira2013}, we present the first
optimal solution to list all the cycles in an undirected graph~$G$.
Specifically, let $\setofcycles(G)$ denote the set of all these cycles
($|\setofcycles(G)| = \eta$).  For a cycle $c \in \setofcycles(G)$,
let $|c|$ denote the number of edges in~$c$. Our algorithm requires
$O(m + \sum_{c \in \setofcycles(G)}{|c|})$ time and is asymptotically
optimal: indeed, $\Omega(m)$ time is necessarily required to read $G$
as input, and $\Omega(\sum_{c \in \setofcycles(G)}{|c|})$ time is
necessarily required to list the output. Since $|c| \leq n$, the cost
of our algorithm never exceeds $O(m + (\eta+1) n)$ time.

Along the same lines, we also present the first optimal solution to
list all the simple paths from $s$ to $t$ (shortly, $st$-paths) in an
undirected graph $G$. Let $\setofpaths_{st}(G)$ denote the set of
$st$-paths in $G$ and, for an $st$-path $\pi \in \setofpaths_{st}(G)$,
let $|\pi|$ be the number of edges in $\pi$.  Our algorithm lists
all the $st$-paths in~$G$ optimally in $O(m + \sum_{\pi \in
  \setofpaths_{st}(G)}{|\pi|})$ time, observing that $\Omega(\sum_{\pi
  \in \setofpaths_{st}(G)}{|\pi|})$ time is necessarily required to
list the output.  

We prove the following reduction to relate
$\setofcycles(G)$ and $\setofpaths_{st}(G)$ for some suitable choices
of vertices $s,t$: if there exists an optimal algorithm to list the
$st$-paths in $G$, then there exists an optimal algorithm to list the
cycles in $G$.  Hence, we can focus on listing $st$-paths.

\subsubsection{Previous work}

The classical problem of listing all the cycles of a graph has been
extensively studied for its many applications in several fields,
ranging from the mechanical analysis of chemical
structures~\cite{Sussenguth65} to the design and analysis of reliable
communication networks, and the graph isomorphism
problem~\cite{Welch66}.
In particular, at the turn of the seventies several algorithms for
enumerating all cycles of an undirected graph have been proposed.
There is a vast body of work, and the majority of the algorithms
listing all the cycles can be divided into the following three classes
(see~\cite{Bezem87,Mateti76} for excellent surveys).

\begin{enumerate}
\item \textit{Search space algorithms.}
According to this approach, cycles are looked for in an appropriate
search space.  In the case of undirected graphs, the \emph{cycle
vector space} \cite{Diestel} turned out to be the most promising
choice: from a basis for this space, all vectors are computed and it
is tested whether they are a cycle. Since the algorithm introduced
in~\cite{Welch66}, many algorithms have been proposed: however, the
complexity of these algorithms turns out to be exponential in the
dimension of the vector space, and thus in $n$. For planar graphs, an
algorithm listing cycles in $O((\eta + 1)n)$ time was presented in
\cite{Syslo81}.

\item \textit{Backtrack algorithms.} 
By this approach, all paths are generated by backtrack and, for each
path, it is tested whether it is a cycle. One of the first algorithms
is the one proposed in~\cite{Tiernan70}, which is however exponential
in $\eta$. By adding a simple pruning strategy, this algorithm has
been successively modified in~\cite{Tarjan73}: it lists all the cycles
in $O(nm(\eta+1))$ time. Further improvements were proposed
in~\cite{Johnson1975,Szwarcfiter76,Read75}, leading to
$O((\eta+1)(m+n))$-time algorithms that work for both directed and
undirected graphs. 

\item \textit{Using the powers of the adjacency matrix.} 
This approach uses the so-called \emph{variable adjacency matrix},
that is, the formal sum of edges joining two vertices. A non-zero
element of the $p$-th power of this matrix is the sum of all walks of
length $p$: hence, to compute all cycles, we compute the $n$th power
of the variable adjacency matrix. This approach is not very efficient
because of the non-simple walks. Algorithms based on this approach
(e.g.\mbox{} \cite{Ponstein66,Yau67}) basically differ only on the way
they avoid to consider walks that are neither paths nor cycles.
\end{enumerate}

Almost 40 years after Johnson's algorithm~\cite{Johnson1975}, the
problem of efficiently listing all cycles of a graph is still an
active area of research
(e.g.~\cite{birmele2012,Halford04,Horvath04,Liu06,Sankar07,Wild08,Schott11}).  New
application areas have emerged in the last decade, such as
bioinformatics: for example, two algorithms for this problem have been
proposed in~\cite{Klamt06,Klamt09} while studying
biological interaction graphs. Nevertheless, no significant
improvement has been obtained from the theory standpoint: in
particular, Johnson's algorithm is still the theoretically most
efficient. His $O((\eta+1)(m+n))$-time solution is surprisingly not
optimal for undirected graphs as we show in this chapter.

\subsubsection{Difficult graphs for Johnson's algorithm}

\begin{figure}[t]
\centering
\definecolor{cqcqcq}{rgb}{0,0,0}
\begin{tikzpicture}[scale=1.0,line cap=round,line join=round,>=triangle 45,x=1.0cm,y=1.0cm]
\draw (-2,2)-- (-1,3);
\draw [dash pattern=on 5pt off 5pt] (-2,2)-- (-1,2.52);
\draw [dash pattern=on 5pt off 5pt] (-2,2)-- (-1,2);
\draw (-1,3)-- (0,2);
\draw [dash pattern=on 5pt off 5pt] (-1,2.52)-- (0,2);
\draw [dash pattern=on 5pt off 5pt] (-1,2)-- (0,2);
\draw (-2,2)-- (-1,1);
\draw (-1,1)-- (0,2);
\draw (0,2)-- (1,3);
\draw [dash pattern=on 5pt off 5pt] (0,2)-- (1.04,2.48);
\draw [dash pattern=on 5pt off 5pt] (0,2)-- (1,2);
\draw (0,2)-- (1,1);
\draw (1,3)-- (2,2);
\draw [dash pattern=on 5pt off 5pt] (1.04,2.48)-- (2,2);
\draw [dash pattern=on 5pt off 5pt] (1,2)-- (2,2);
\draw (1,1)-- (2,2);
\draw [dash pattern=on 5pt off 5pt] (-2,2)-- (-1,1.48);
\draw [dash pattern=on 5pt off 5pt] (-1,1.48)-- (0,2);
\draw [dash pattern=on 5pt off 5pt] (0,2)-- (1,1.44);
\draw [dash pattern=on 5pt off 5pt] (1,1.44)-- (2,2);
\draw [shift={(0,2)}] plot[domain=0:3.14,variable=\t]({1*2*cos(\t r)+0*2*sin(\t r)},{0*2*cos(\t r)+1*2*sin(\t r)});
\begin{footnotesize}
\fill [color=black] (-2,2) circle (1.5pt);
\draw[color=black] (-2,2) node[left] {$a$};
\fill [color=black] (-1,3) circle (1.5pt);
\draw[color=black] (-1,3) node[above] {$v_1$};
\fill [color=black] (-1,2) circle (1.5pt);
\fill [color=black] (-1,1) circle (1.5pt);
\draw[color=black] (-1,1) node[below] {$v_k$};
\fill [color=black] (-1,1.48) circle (1.5pt);
\fill [color=black] (-1,2.52) circle (1.5pt);
\fill [color=black] (0,2) circle (1.5pt);
\draw[color=black] (0,2) node[above] {$b$};
\fill [color=black] (1,3) circle (1.5pt);
\draw[color=black] (1,3) node[above] {$u_1$};
\fill [color=black] (1,2) circle (1.5pt);
\fill [color=black] (1.04,2.48) circle (1.5pt);
\fill [color=black] (1,1) circle (1.5pt);
\draw[color=black] (1,1) node[below] {$u_k$};
\fill [color=black] (2,2) circle (1.5pt);
\draw[color=black] (2,2) node[right] {$c$};
\fill [color=black] (1,1.44) circle (1.5pt);
\end{footnotesize}
\end{tikzpicture}
\caption{Diamond graph.}
\label{fig:johnsoncounter}
\end{figure}
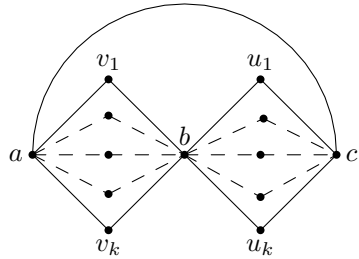

It is worth observing that the analysis of the time complexity of Johnson's algorithm
is not pessimistic and cannot match the one of our algorithm for
listing cycles.  For example, consider the sparse ``diamond'' graph
$D_n = (V, E)$ in Fig.~\ref{fig:johnsoncounter} with $n=2k+3$
vertices in $V = \{a,b,c, v_1, \ldots, v_k, u_1, \ldots, u_k\}$. There
are $m = \Theta(n)$ edges in $E = \{ (a,c)$, $(a,v_i)$, $(v_i,b)$,
$(b,u_i)$, $(u_i,c)$, for $1 \leq i \leq k\}$, and three kinds of
(simple) cycles:
(1)~$(a, v_i), (v_i, b), (b, u_j), (u_j, c), (c, a)$ for $1 \leq i, j
\leq k$;
(2)~$(a, v_i), (v_i, b), (b, v_j), (v_j, a)$ for $1 \leq i < j \leq
k$;
(3)~$(b, u_i), (u_i, c), (c, u_j), (u_j, b)$ for $1 \leq i < j \leq
k$,
totalizing $\eta = \Theta(n^2)$ cycles.
Our algorithm takes $\Theta(n + k^2) = \Theta(\eta) = \Theta(n^2)$
time to list these cycles.  On the other hand, Johnson's algorithm
takes $\Theta(n^3)$ time, and the discovery of the $\Theta(n^2)$ cycles
in~(1) costs $\Theta(k) = \Theta(n)$ time each: the backtracking
procedure in Johnson's algorithm starting at $a$, and passing through
$v_i$, $b$ and $u_j$ for some $i,j$, arrives at $c$: at that point, it
explores all the vertices $u_l$ $(l \neq i)$ even if they do not lead
to cycles when coupled with $a$, $v_i$, $b$, $u_j$, and $c$.

\section{Preliminaries}

Let $G=(V,E)$ be an undirected connected graph with $n=|V|$ vertices
and $m=|E|$ edges, without self-loops or parallel edges. For a vertex
$u \in V$, we denote by $N(u)$ the neighborhood of $u$ and by
$d(u)=|N(u)|$ its degree.  $G[V']$ denotes the subgraph \emph{induced}
by $V' \subseteq V$, and $G - u$ is the induced subgraph $G[ V
\setminus \{u\}]$ for $u \in V$. Likewise for edge $e \in E$, we adopt
the notation $G-e = (V,E \setminus \{e\})$. For a vertex $v \in V$,
the \emph{postorder} DFS number of $v$ is the relative time in which
$v$ was \emph{last} visited in a DFS traversal, i.e. the position of
$v$ in the vertex list ordered by the last visiting time of each
vertex in the DFS.

Paths are simple in $G$ by definition: we refer to a path $\pi$ by its
natural sequence of vertices or edges.  A path $\pi$ from $s$ to $t$,
or $st$-\emph{path}, is denoted by $\pi = s \leadsto t$. Additionally,
$\setofpaths(G)$ is the set of all paths in $G$ and
$\setofpaths_{s,t}(G)$ is the set of all $st$-paths in $G$.  When
$s=t$ we have cycles, and $\setofcycles(G)$ denotes the set of all
cycles in $G$. We denote the number of edges in a path $\pi$ by
$|\pi|$ and in a cycle~$c$ by $|c|$. In this chapter, we consider the following
problems.

\begin{problem}[Listing st-Paths]
	\label{prob:liststpaths}
	Given an undirected graph $G=(V,E)$ and two distinct vertices
	$s,t \in V$, output all the paths $\pi \in
	\setofpaths_{s,t}(G)$.
\end{problem} 

\begin{problem}[Listing Cycles]
	\label{prob:listcycles}
	Given an undirected graph $G=(V,E)$, output all the cycles $c
	\in \setofcycles(G)$.
\end{problem} 

\begin{figure}[t]
\centering
\begin{tikzpicture}
[nodeDecorate/.style={shape=circle,inner sep=1pt,draw,thick,fill=black},%
  lineDecorate/.style={-,dashed},%
  elipseDecorate/.style={color=gray!30},
  scale=0.6]
\fill [elipseDecorate] (5,10) circle (2);
\fill [elipseDecorate] (9,10) circle (2);
\fill [elipseDecorate] (2,10) circle (1);
\fill [elipseDecorate] (-1,10) circle (2);
\fill [elipseDecorate,rotate around={-55:(-1,8)}] (-1,7) circle (1);

\draw (5,10) circle (2);
\draw (9,10) circle (2);
\draw (2,10) circle (1);
\draw (-1,10) circle (2);
\draw (5,7) circle (1);
\draw (9,7) circle (1);
\draw [rotate around={55:(-1,8)}] (-1,7) circle (1);
\draw [rotate around={-55:(-1,8)}] (-1,7) circle (1);
\draw (13,10) circle (2);
\draw (13,7) circle (1);
\draw (-4,10) circle (1);
\begin{footnotesize}
\node (7) at (-2,7) [nodeDecorate,label=above:$s$] {};
\node (14) at (9,11) [nodeDecorate,label=above:$t$] {};
\end{footnotesize}
\foreach \nodename/\x/\y in {
  0/7/10, 1/5/11,
  2/3/10, 3/1/10, 4/-3/10, 5/-1/10, 6/-1/8, 7/-2/7, 8/-1/11,
  9/0/7,
  11/5/10, 12/4/9, 13/5/8, 14/9/11, 15/11/10, 16/9/9,
  17/9/7 , 18/9/8, 50/5/7, 51/13/11, 52/13/8, 53/13/7, 54/-4/10}
{
  \node (\nodename) at (\x,\y) [nodeDecorate] {};
}

\path
\foreach \startnode/\endnode in {6/7, 6/9, 5/6, 5/3, 5/8, 4/5, 3/2,
2/11, 1/11, 11/12, 11/0, 11/13, 13/50, 0/14, 14/15, 15/16, 16/18,
18/17, 15/51, 15/52, 51/52, 52/53, 54/4}
{
  (\startnode) edge[lineDecorate] node {} (\endnode)
};

\path
\foreach \startnode/\endnode/\bend in { 8/3/20, 6/3/20, 4/8/10,
12/13/20, 1/0/20, 2/1/20, 13/0/20, 0/18/10}
{
  (\startnode) edge[lineDecorate, bend left=\bend] node {} (\endnode)
};

\end{tikzpicture}
\caption{Block tree of $G$ with bead string $\sbeadstring$ in gray.}
\label{fig:beadstring}
\end{figure}

Our algorithms assume without loss of generality that the input graph
$G$ is connected, hence $m \ge n-1$, and use the decomposition of $G$
into biconnected components. Recall that an \emph{articulation point}
(or cut-vertex) is a vertex $u \in V$ such that the number of
connected components in $G$ increases when $u$ is removed. $G$ is
\emph{biconnected} if it has no articulation points. Otherwise, $G$
can always be decomposed into a tree of biconnected components, called
the \emph{block tree}, where each biconnected component is a maximal
biconnected subgraph of $G$ (see Fig.~\ref{fig:beadstring}), and
two biconnected components are adjacent if and only if they share an
articulation point.

\section{Overview and main ideas}
\label{sec:overview}

  While the basic approach is simple (see the binary partition in
  point~\ref{item:abstract:3}), we use a number of non-trivial ideas to
  obtain our optimal algorithm for an undirected (connected) graph $G$
  as summarized in the steps below.
  \begin{enumerate}
  \item Prove the following reduction. If there exists an optimal
    algorithm to list the $st$-paths in $G$, there exists an optimal
    algorithm to list the cycles in $G$. This relates
    $\setofcycles(G)$ and $\setofpaths_{st}(G)$ for some choices $s,t$.


  \item Focus on listing the $st$-paths. Consider the decomposition of
    the graph into biconnected components ({\bcc}s), thus forming a
    tree $T$ where two {\bcc}s are adjacent in $T$ iff they share an
    articulation point. Exploit (and prove) the property that if $s$
    and $t$ belong to distinct {\bcc}s, then $(i)$ there is a unique
    \emph{sequence} $\sbeadstring$ of adjacent {\bcc}s in $T$ through
    which each $st$-path must necessarily pass, and $(ii)$ each
    $st$-path is the concatenation of paths connecting the
    articulation points of these {\bcc}s in $\sbeadstring$.

  \item \label{item:abstract:3} Recursively list the $st$-paths in
    $\sbeadstring$ using the classical binary partition (i.e.\mbox{}
    given an edge $e$ in $G$, list all the cycles containing
    $e$, and then all the cycles not containing~$e$): now it suffices to
    work on the \emph{first} \bcc\ in $\sbeadstring$, and efficiently
    maintain it when deleting an edge $e$, as required by the binary
    partition.

  \item Use a notion of \emph{certificate} to avoid recursive calls
    (in the binary partition) that do not list new $st$-paths.  This
    certificate is maintained dynamically as a data structure
    representing the first \bcc\ in $\sbeadstring$, which guarantees
    that there exists at least one \emph{new} solution in the current
    $\sbeadstring$.

  \item Consider the binary recursion tree corresponding to the binary
    partition.  Divide this tree into \emph{spines}: a spine
    corresponds to the recursive calls generated by the edges $e$
    belonging to the same adjacency list in $\sbeadstring$.  The
    amortized cost for each listed $st$-path $\pi$ is $O(|\pi|)$ when
    there is a guarantee that the amortized cost in each spine $S$ is
    $O(\mu)$, where $\mu$ is a lower bound on the number of $st$-paths
    that will be listed from the recursive calls belonging to~$S$. The
    (unknown) parameter~$\mu$, which is different for each spine~$S$, and the
    corresponding cost $O(\mu)$, will drive the design of the proposed
    algorithms.
  \end{enumerate}

\subsection{Reduction to $st$-paths}
\label{sub:reduction-paths}

We now show that listing cycles reduces to listing $st$-paths while
preserving the optimal complexity.  

\begin{lemma}
  \label{lemma:reduction}
  Given an algorithm that solves Problem~\ref{prob:liststpaths} in
  \mbox{$O(m + \sum_{\pi \in \setofpaths_{s,t}(G)}{|\pi|})$}
  time, there exists an algorithm that solves
  Problem~\ref{prob:listcycles} in \mbox{$O(m + \sum_{c \in
  \setofcycles(G)}{|c|})$} time.
\end{lemma}
\begin{proof}
  Compute the biconnected components of $G$ and keep them in a list
  $L$. Each (simple) cycle is contained in one of the biconnected
  components and therefore we can treat each biconnected component
  individually as follows. While $L$ is not empty, extract a biconnected
  component $B=(V_{B},E_{B})$ from $L$ and repeat the following three
  steps: $(i)$ compute a DFS traversal of $B$ and take any back edge
  $b=(s,t)$ in $B$; $(ii)$ list all $st$-paths in $B-b$, i.e.~the
  cycles in $B$ that include edge~$b$; $(iii)$ remove edge $b$ from
  $B$, compute the new biconnected components thus created by removing
  edge~$b$, and append them to $L$. When $L$ becomes empty, all the
  cycles in $G$ have been listed.

  Creating $L$ takes $O(m)$ time. For every $B \in L$, steps $(i)$ and
  $(iii)$ take $O(|E_B|)$ time.  Note that step $(ii)$ always outputs
  distinct cycles in $B$ (i.e.~$st$-paths in $B-b$) in
  $O(|E_{B}|+\sum_{\pi \in \setofpaths_{s,t}(B-b)}{|\pi|})$ time.
  However, $B-b$ is then decomposed into biconnected components whose
  edges are traversed again. We can pay for the latter cost: for any
  edge $e \neq b$ in a biconnected component $B$, there is always a
  cycle in $B$ that contains both $b$ and $e$ (i.e.\mbox{} it is an
  $st$-path in $B-b$), hence $\sum_{\pi \in
    \setofpaths_{s,t}(B-b)}{|\pi|}$ dominates the term $|E_{B}|$,
  i.e.~$\sum_{\pi \in \setofpaths_{s,t}(B-b)}{|\pi|}=
  \Omega(|E_{B}|)$.  Therefore steps $(i)$--$(iii)$ take $O(\sum_{\pi
    \in \setofpaths_{s,t}(B-b)}{|\pi|})$ time. When $L$ becomes empty,
  the whole task has taken $O(m + \sum_{c \in \setofcycles(G)}{|c|})$
  time.
\end{proof}

\subsection{Decomposition in biconnected components}
\label{sec:decomposition}

We now focus on listing $st$-paths
(Problem~\ref{prob:liststpaths}). We use the decomposition of $G$ into
a block tree of biconnected components.  Given vertices $s,t$, define
its \emph{bead string}, denoted by $\sbeadstring$, as the unique
sequence of one or more adjacent biconnected components (the
\emph{beads}) in the block tree, such that the first one contains $s$
and the last one contains $t$ (see Fig.~\ref{fig:beadstring}): these
biconnected components are connected through articulation points,
which must belong to all the paths to be listed.

\begin{lemma}
  \label{lemma:beadstring}
  All the $st$-paths in $\setofpaths_{s,t}(G)$ are contained in the
  induced subgraph $G[\sbeadstring]$ for the bead string
  $\sbeadstring$. Moreover, all the articulation points in
  $G[\sbeadstring]$ are traversed by each of these paths.
\end{lemma}
\begin{proof}
  Consider an edge $e = (u,v)$ in $G$ such that $u \in \sbeadstring$
  and $v \notin \sbeadstring$. Since the biconnected components of a
  graph form a tree and the bead string $\sbeadstring$ is a path in
  this tree, there are no paths $v \leadsto w$ in $G-e$ for any $w \in
  \sbeadstring$ because the biconnected components in $G$ are maximal
  and there would be a larger one (a contradiction).
  Moreover, let $B_1, B_2, \ldots, B_r$ be the biconnected components
  composing $\sbeadstring$, where $s \in B_1$ and $t \in B_r$. If
  there is only one biconnected component in the path (i.e.~$r=1$),
  there are no articulation points in $\sbeadstring$.  Otherwise, all
  of the $r-1$ articulation points in $\sbeadstring$ are traversed by
  each path $\pi \in \setofpaths_{s,t}(G)$: indeed, the articulation
  point between adjacent biconnected components $B_i$ and $B_{i+1}$ is
  their only vertex in common and there are no edges linking $B_i$ and
  $B_{i+1}$.
\end{proof}

We thus restrict the problem of listing the paths in
$\setofpaths_{s,t}(G)$ to the induced subgraph $G[\sbeadstring]$,
conceptually isolating it from the rest of $G$. For the sake of
description, we will use interchangeably $\sbeadstring$ and
$G[\sbeadstring]$ in the rest of the chapter.

\subsection{Binary partition scheme}
\label{sec:basic-scheme}

We list the set of $st$-paths in $\sbeadstring$, denoted by
$\setofpaths_{s,t}(\sbeadstring)$, by applying the binary partition
method (where $\setofpaths_{s,t}(G) = \setofpaths_{s,t}(\sbeadstring)$
by Lemma~\ref{lemma:beadstring}): we choose an edge $e = (s,v)$
incident to~$s$ and then list all the $st$-paths that include $e$ and
then all the $st$-paths that do not include $e$. Since we delete some
vertices and some edges during the recursive calls, we proceed as follows.

\smallskip

\noindent{\it Invariant:} At a generic recursive step on vertex $u$
(initially, $u:=s$), let $\pi_s = s \leadsto u$ be the path discovered
so far (initially, $\pi_s$ is empty $\{\}$). Let $\beadstring$ be the
current bead string (initially, $\beadstring :=
\sbeadstring$). More precisely, $\beadstring$ is defined as follows:
$(i)$~remove from $\sbeadstring$ all the vertices in $\pi_s$ but $u$, and
the edges incident to $u$ and discarded so far; $(ii)$~recompute the
block tree on the resulting graph; $(iii)$~$\beadstring$ is the unique
bead string that connects $u$ to $t$ in the recomputed block tree.

\smallskip

\noindent{\it Base case:} When $u=t$, output the $st$-path $\pi_s$.

\smallskip

\noindent{\it Recursive rule:} Let $\setofpaths(\pi_s, u,
		\beadstring)$ denote the set of $st$-paths to be
		listed by the current recursive call. Then, it is the
		union of the following two disjoint sets, for an edge
		$e=(u,v)$ incident to~$u$:
\begin{itemize}
\item \emph{Left branching:} the $st$-paths in $\setofpaths(\pi_s \cdot e,
  v, \vbeadstring)$ that use $e$, where $\vbeadstring$ is the unique
  bead string connecting $v$ to $t$ in the block tree resulting from
  the deletion of vertex $u$ from $\beadstring$.
\item \emph{Right branching:} the $st$-paths in $\setofpaths(\pi_s, u,
  \beadstring')$ that do \emph{not} use~$e$, where $\beadstring'$ is
  the unique bead string connecting $u$ to $t$ in the block tree
  resulting from the deletion of edge $e$ from $\beadstring$.
\end{itemize}

\noindent
Hence, $\setofpaths_{s,t}(\sbeadstring)$ (and so
$\setofpaths_{s,t}(G)$) can be computed by invoking $\setofpaths(\{\}, s,
\sbeadstring)$. The correctness and completeness of the above approach
is discussed in Section~\ref{sec:intro-cert}.

At this point, it should be clear why we introduce the notion of bead
strings in the binary partition. The existence of the partial path
$\pi_s$ and the bead string $\beadstring$ guarantees that there surely
exists at least one $st$-path. But there are two sides of the coin
when using $\beadstring$.

\begin{enumerate}
\item One advantage is that we can avoid useless recursive calls:
If vertex $u$ has only one incident edge $e$, we just perform the left
branching; otherwise, we can safely perform both the left and right
branching since the \emph{first} bead in $\beadstring$ is always a
biconnected component by definition (thus there exists both an
$st$-path that traverses $e$ and one that does not).

\item \label{side_coin:2} The other side of the coin is that we have to maintain the
bead string $\beadstring$ as $\vbeadstring$ in the left branching and
as $\beadstring'$ in the right branching by
Lemma~\ref{lemma:beadstring}. Note that these bead strings are surely
non-empty since $\beadstring$ is non-empty by induction (we only
perform either left or left/right branching when there are solutions by
item~1).
\end{enumerate}

To efficiently address point~\ref{side_coin:2}, we need to introduce the notion of
certificate as described next.

\subsection{Introducing the certificate}
\label{sec:intro-cert}

Given the bead string $\beadstring$, we call the \emph{head} of
$\beadstring$, denoted by $\head$, the first biconnected component in
$\beadstring$, where $u \in \head$. Consider a DFS tree of
$\beadstring$ rooted at $u$ that changes along with $\beadstring$, and
classify the edges in $\beadstring$ as tree edges or back edges (there
are no cross edges since the graph is undirected).

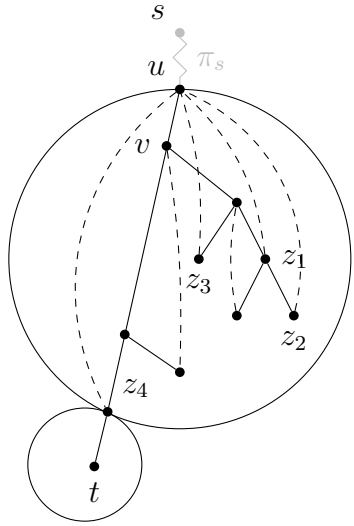
\begin{figure}[t]
	\centering
\begin{tikzpicture}
[nodeDecorate/.style={shape=circle,inner sep=1pt,draw,thick,fill=black},%
  lineDecorate/.style={-,dashed},%
  elipseDecorate/.style={color=gray!30},
  scale=0.25]
\draw (10,22) circle (9);
\draw (5,11.1) circle (3);

\node (s) at (10,34) [nodeDecorate,color=lightgray,label=above left:$s$] {};
\node (u) at (10,31) [nodeDecorate,label=above left:{ $u$}] {};

\node (tp) at (6.2,13.9) [nodeDecorate,label=above right:{ $z_4$}] {};
\node (t) at (5.5,11) [nodeDecorate,label=below:$t$] {};

\path {
	(s) edge[snake,-,color=lightgray] node {\quad\quad$\pi_s$} (u)
	(u) edge node {} (tp)
	(tp) edge node {} (t)
};

\node (a) at (9.3,28) [nodeDecorate,label=left:$v$] {};
\node (b) at (7.1,18) [nodeDecorate,] {};
\node (c) at (13,25) [nodeDecorate,] {};
\node (d) at (14.5,22) [nodeDecorate,label=right:$z_1$] {};
\node (e) at (11,22) [nodeDecorate,label=below:$z_3$] {};
\node (f) at (16,19) [nodeDecorate,label=below:$z_2$] {};
\node (g) at (13,19) [nodeDecorate,] {};
\node (h) at (10,16) [nodeDecorate,] {};

\path {
	(a) edge node {} (c)
	(c) edge node {} (d)
	(d) edge node {} (f)
	(c) edge node {} (e)
	(d) edge node {} (g)
	(b) edge node {} (h)
};

\path {
	(u) edge[dashed,bend left=-40] node {} (tp)
	(f) edge[dashed,bend left=-40] node {} (u)
	(g) edge[dashed,bend left=10] node {} (c)
	(e) edge[dashed,bend left=-10] node {} (u)
	(d) edge[dashed,bend left=-20] node {} (u)
	(a) edge[dashed,bend left=5] node {} (h)
};

\end{tikzpicture}
\caption{Example certificate of $B_{u,t}$ \label{fig:Certificate}}
\end{figure}

To maintain $\beadstring$ (and so $\head$) during the recursive calls,
we introduce a \emph{certificate} $C$ (see Fig.~\ref{fig:Certificate}): It is a suitable data structure that uses the above
classification of the edges in $\beadstring$, and supports the
following operations, required by the binary partition scheme.
\begin{itemize}
\item $\chooseedge(C,u)$: returns an edge $e = (u,v)$ with $v \in
	\head$ such that $\pi_s \cdot (u,v) \cdot u \leadsto t$ is an
	$st$-path such that $u \leadsto t$ is inside $\beadstring$.
	Note that $e$ always exists since $\head$ is biconnected.
	Also, the chosen $v$ is the last one in DFS postorder among the
	neighbors of $u$: in this way, the (only) tree edge $e$ is
	returned when there are no back edges leaving from~$u$.  (As
	it will be clear in Sections~\ref{sec:recursion-amortization}
	and~\ref{sec:certificate}, this order facilitates the analysis
	and the implementation of the certificate.)
\item $\oracleleft(C,e)$: for the given $e=(u,v)$, it obtains
	$\vbeadstring$ from $\beadstring$ as discussed in
	Section~\ref{sec:basic-scheme}. This implies updating also
	$\head$, $C$, and the block tree, since the recursion
	continues on~$v$. It returns bookkeeping information $I$ for
	what is updated, so that it is possible to revert to
	$\beadstring$, $\head$, $C$, and the block tree, to their
	status before this operation.
\item $\oracleright(C,e)$: for the given $e=(u,v)$, it obtains
	$\beadstring'$ from $\beadstring$ as discussed in
	Section~\ref{sec:basic-scheme}, which implies updating also
	$\head$, $C$, and the block tree. It returns bookkeeping
	information $I$ as in the case of $\oracleleft(C,e)$.
\item $\undooracle(C,I)$: reverts the bead string to $\beadstring$,
	the head $\head$, the certificate $C$, and the block tree, to
	their status before operation $I := \oracleleft(C,e)$ or $I :=
	\oracleright(C,e)$ was issued (in the same recursive call).
\end{itemize}

Note that a notion of certificate in listing problems has been
introduced in~\cite{Ferreira11}, but it cannot be directly applied to
our case due to the different nature of the problems and our use of more
complex structures such as biconnected components. 

Using our certificate
and its operations, we can now formalize the binary partition and its
recursive calls $\setofpaths(\pi_s, u, \beadstring)$ described in
Section~\ref{sec:basic-scheme} as Algorithm~\ref{alg:liststpaths},
where $\beadstring$ is replaced by its certificate $C$.

\begin{algorithm}[t]
	\caption{\label{alg:liststpaths} $\liststpaths(\pi_s,\,u,\,C)$}
\begin{algorithmic}[1]
	\IF{$u=t$}
		\STATE $\routput(\pi_s)$ \label{code:base}
		\STATE $\return$ \label{code:returnbase}
	\ENDIF
	\STATE $e = (u,v) := \chooseedge( C, u )$ \label{code:choose}
	\IF{ $e \text{ is back edge}$ \label{code:if_back}}
		\STATE $I := \oracleright(C,e)$  \label{code:right_update}
		\STATE $\liststpaths(\pi_s,\, u,\,C)$ \label{code:right_branch}
		\STATE $\undooracle(C, I)$ \label{code:right_undo}
	\ENDIF
        \STATE $I := \oracleleft(C,e)$ \label{code:left_update}
        \STATE $\liststpaths( \pi_s \cdot (u,v),\, v,\, C)$ \label{code:left_branch}
        \STATE $\undooracle(C, I)$ \label{code:left_undo}
\end{algorithmic}
\end{algorithm}

The base case ($u=t$) corresponds to lines~1--4 of
Algorithm~\ref{alg:liststpaths}. During recursion, the left branching
corresponds to lines~5 and~11-13, while the right branching to
lines~5--10. Note that we perform only the left branching when there is
only one incident edge in $u$, which is a tree edge by definition of
$\chooseedge$. Also, lines~9 and~13 are needed to restore the
parameters to their values when returning from the recursive
calls.

\begin{lemma}
  \label{lemma:correctness_algo_listpaths}
  Algorithm~\ref{alg:liststpaths} correctly lists all the $st$-paths in
  $\setofpaths_{s,t}(G)$.
\end{lemma}
\begin{proof}
  For a given vertex $u$ the function $\chooseedge(C, u)$ returns an
  edge $e$ incident to $u$. We maintain the invariant that $\pi_s$ is
  a path $s \leadsto u$, since at the point of the recursive call in
  line~\ref{code:left_branch}: (i) is connected as we append edge
  $(u,v)$ to $\pi_s$ and; (ii) it is simple as vertex $u$ is removed
  from the graph $G$ in the call to $\oracleleft(C,e)$ in
  line~\ref{code:left_update}. In the case of recursive call in
  line~\ref{code:right_branch} the invariant is trivially maintained
  as $\pi_s$ does not change.
  The algorithm only outputs $st$-paths since $\pi_s$ is
  a $s \leadsto u$ path and $u=t$ when the algorithm outputs, in
  line~\ref{code:base}. 

  The paths with prefix $\pi_s$ that do not use $e$ are listed by
  the recursive call in line~\ref{code:right_branch}. This is done by
  removing~$e$ from the graph (line~\ref{code:right_update}) and thus
  no path can include $e$. Paths that use $e$ are listed in
  line~\ref{code:left_branch} since in the recursive call $e$ is added
  to $\pi_s$. Given that the tree edge incident to $u$ is the last one
  to be returned by $\chooseedge(C,u)$, there is no path that does not
  use this edge, therefore it is not necessary to call
  line~\ref{code:right_branch} for this edge.
\end{proof}

A natural question is what is the time complexity: we must account for
the cost of maintaining~$C$ and for the cost of the recursive calls of
Algorithm~\ref{alg:liststpaths}. Since we cannot always maintain the
certificate in $O(1)$ time, the ideal situation for attaining an
optimal cost is taking $O(\mu)$ time if at least $\mu$ $st$-paths are
listed in the current call (and its nested calls). Unfortunately, we
cannot estimate~$\mu$ efficiently and cannot design
Algorithm~\ref{alg:liststpaths} so that it takes $O(\mu)$ adaptively.
We circumvent this by using a different cost scheme in
Section~\ref{sub:recursion-tree-cost} that is based on the recursion
tree induced by Algorithm~\ref{alg:liststpaths}.
Section~\ref{sec:certificate} is devoted to the efficient
implementation of the above certificate operations according to the
cost scheme that we discuss next.

\subsection{Recursion tree and cost amortization}
\label{sub:recursion-tree-cost}

We now show how to distribute the costs among the several recursive
calls of Algorithm~\ref{alg:liststpaths} so that optimality is
achieved. Consider a generic execution on the bead string
$\beadstring$. We trace this execution by using a binary recursion
tree $R$. The nodes of $R$ are labeled by the arguments of 
Algorithm~\ref{alg:liststpaths}: specifically, we denote a node
in $R$ by the triple $x = \langle \pi_s, u, C \rangle$ iff it
represents the call with arguments $\pi_s$, $u$, and~$C$.\footnote{For
clarity, we use ``nodes'' when referring to $R$ and ``vertices'' when
referring to $\beadstring$.}  The left branching is represented by the
left child, and the right branching (if any) by the right child of the
current node. 

\begin{lemma}
\label{lem:properties_recursion}
The binary recursion tree $R$ for $\beadstring$ has the following properties: 
\begin{enumerate}
	\setlength{\itemsep}{0pt} 
\item \label{item:R1} There is a one-to-one correspondence between the
  paths in $\setofpaths_{s,t}(\beadstring)$ and the leaves in the recursion
  tree rooted at node $\langle \pi_s, u, C \rangle$.
\item \label{item:R2} Consider any leaf and its corresponding $st$-path
  $\pi$: there are $|\pi|$ left branches in
  the corresponding root-to-leaf trace.
\item \label{item:R3} Consider the instruction $e:=\chooseedge(C,u)$
  in Algorithm~\ref{alg:liststpaths}: unary (i.e.\mbox{} single-child)
  nodes correspond to left branches ($e$ is a tree edge) while binary
  nodes correspond to left and right branches ($e$ is a back
  edge).
\item \label{item:R4} The number of binary nodes is
  $|\setofpaths_{s,t}(\beadstring)| - 1$.
\end{enumerate}
\end{lemma}
\begin{proof}
We proceed in order as follows.
\begin{enumerate}
\item We only output a solution in a leaf and we only do recursive
  calls that lead us to a solution. Moreover every node partitions the
  set of solutions in the ones that use an edge and the ones that do
  not use it. This guarantees that the leaves in the left subtree of
  the node corresponding to the recursive call and the leaves in the
  right subtree do not intersect. This implies that different leaves
  correspond to different paths from $s$ to $t$, and that for each
  path there is a corresponding leaf.
\item Each left branch corresponds to the inclusion of an edge in the
  path $\pi$.
\item Since we are in a biconnected component, there is always a left
  branch. There can be no unary node as a right branch: indeed for any
  edge of $\beadstring$ there exists always a path from $s$ to $t$
  passing through that edge. Since the tree edge is always the last
  one to be chosen, unary nodes cannot correspond to back edges and
  binary nodes are always back edges.
\item It follows from point \ref{item:R1} and from the fact that the recursion tree is
  a binary tree. (In any binary tree, the number of binary nodes is
  equal to the number of leaves minus 1.)
\end{enumerate}
\end{proof}

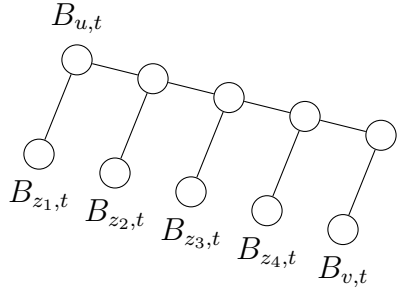
\begin{figure}[t]
	\centering
\begin{tikzpicture}
[nodeDecorate/.style={shape=circle,inner sep=4pt,draw,fill=white},%
  lineDecorate/.style={-,dashed},%
  elipseDecorate/.style={color=gray!30},
  scale=0.25]

  \node (a) at (0,0) [nodeDecorate,label=above:$B_{u,t}$] {};
  \node (b) at (4,-1) [nodeDecorate] {};
  \node (c) at (8,-2) [nodeDecorate] {};
  \node (d) at (12,-3) [nodeDecorate] {};
  \node (e) at (16,-4) [nodeDecorate] {};

  \node (a2) at (-2,-5) [nodeDecorate,label=below:$B_{z_1,t}$] {};
  \node (b2) at (2,-6) [nodeDecorate,label=below:$B_{z_2,t}$] {};
  \node (c2) at (6,-7) [nodeDecorate,label=below:$B_{z_3,t}$] {};
  \node (d2) at (10,-8) [nodeDecorate,label=below:$B_{z_4,t}$] {};
  \node (e2) at (14,-9) [nodeDecorate,label=below:$B_{v,t}$] {};

\path {
	(a) edge node {} (b)
	(b) edge node {} (c)
	(c) edge node {} (d)
	(d) edge node {} (e)
	(a) edge node {} (a2)
	(b) edge node {} (b2)
	(c) edge node {} (c2)
	(d) edge node {} (d2)
	(e) edge node {} (e2)
};
\end{tikzpicture}
\caption{Spine of the recursion tree \label{fig:spine}}
\end{figure}
We define a \emph{spine} of $R$ to be a subset of $R$'s nodes linked
as follows: the first node is a node $x$ that is either the left child
of its parent or the root of $R$, and the other nodes are those
reachable from $x$ by right branching in $R$. Let $x = \langle \pi_s,
u, C \rangle$ be the first node in a spine $S$. The nodes in $S$
correspond to the edges that are incident to vertex $u$ in
$\beadstring$: hence their number equals the degree $d(u)$ of $u$ in
$\beadstring$, and the deepest (last) node in $S$ is always a tree
edge in $\beadstring$ while the others are back edges.
Fig.~\ref{fig:spine} shows the spine corresponding to $B_{u,t}$ in
Fig.~\ref{fig:Certificate}. Summing up, $R$ can be seen as composed by
spines, unary nodes, and leaves where each spine has a unary node as
deepest node. This gives a global picture of $R$ that we now exploit
for the analysis.


We define the \emph{compact head}, denoted by \mbox{$\chead = (V_X,
E_X)$}, as the (multi)graph obtained by compacting the maximal chains
of degree-2 vertices, except $u$, $t$, and the vertices that are the
leaves of its DFS tree rooted at~$u$.

The rationale behind the above definition is that the costs defined in
terms of $\chead$ amortize well, as the size of $\chead$ and the
number of $st$-paths in the subtree of $R$ rooted at node $x = \langle
\pi_s, u, C \rangle$ are intimately related (see
Lemma~\ref{lemma:lower_bound_paths_beadstring} in
Section~\ref{sec:recursion-amortization}) while this is not
necessarily true for $\head$.

Recall that each leaf corresponds to a path $\pi$ and each spine
corresponds to a
compact head $\chead=(V_X,E_X)$. We now define the following abstract
cost for spines, unary nodes, and leaves of $R$, for a sufficiently
large constant $c_0 > 0$, that Algorithm~\ref{alg:liststpaths} must
fulfill:
\begin{equation}
  \label{eq:abstrac_cost}
  T(r) =
  \left\{\begin{array}{ll}
      c_0 & \mbox{if $r$ is unary}\\ 
      c_0 |\pi| & \mbox{if $r$ is a leaf}\\ 
      c_0 (|V_X|+|E_X|) \quad & \mbox{if $r$ is a spine}
\end{array}\right. 
\end{equation}

\begin{lemma}
  \label{lemma:total_cost_recursion_tree}
  The sum of the costs in the nodes of the recursion tree $\sum_{r \in R} T(r) = O(\sum_{\pi \in \setofpaths_{s,t}(\beadstring)}{|\pi|})$.
\end{lemma}

Section~\ref{sec:recursion-amortization} contains the proof of
Lemma~\ref{lemma:total_cost_recursion_tree} and related properties.
Setting $u:=s$, we obtain that the cost in
Lemma~\ref{lemma:total_cost_recursion_tree} is optimal, by
Lemma~\ref{lemma:beadstring}.

\begin{theorem}
  \label{theorem:optimal_paths}
  Algorithm~\ref{alg:liststpaths} solves problem
  Problem~\ref{prob:liststpaths} in $O(m + \sum_{\pi \in
  \setofpaths_{s,t}(G)}{|\pi|})$ time.
\end{theorem}

By Lemma~\ref{lemma:reduction}, we obtain an optimal result for
listing cycles.

\begin{theorem}
  \label{theorem:optimal_cycles}
  Problem~\ref{prob:listcycles} can be optimally solved in \mbox{$O(m +
  \sum_{c \in \setofcycles(G)}{|c|})$} time. 
\end{theorem}

\section{Amortization strategy}
\label{sec:recursion-amortization}

We devote this section to prove
Lemma~\ref{lemma:total_cost_recursion_tree}. Let us split the sum
in Eq.~\eqref{eq:abstrac_cost} in three parts, and bound each part
individually, as
\begin{equation}
  \label{eq:sum_R}
  \sum_{r \in R} T(r) \leq \sum_{r:\,\mathrm{unary}} T(r) + \sum_{r:\,\mathrm{leaf}} T(r) + \sum_{r:\,\mathrm{spine}} T(r).
\end{equation}

We have that $\sum_{r:\,\mathrm{unary}} T(r) = O(\sum_{\pi \in
  \setofpaths_{s,t}(G)}{|\pi|})$, since there are
$|\setofpaths_{s,t}(G)|$ leaves, and the root-to-leaf trace leading to
the leaf for $\pi$ contains at most $|\pi|$ unary nodes by
Lemma~\ref{lem:properties_recursion}, where each unary node has cost
$O(1)$ by Eq.~\eqref{eq:abstrac_cost}.

Also, $\sum_{r:\,\mathrm{leaf}} T(r) = O(\sum_{\pi \in
  \setofpaths_{s,t}(G)}{|\pi|})$, since the leaf $r$ for $\pi$ has cost
$O(|\pi|)$ by Eq.~\eqref{eq:abstrac_cost}.

It remains to bound $\sum_{r\,\mathrm{spine}} T(r)$. By
Eq.~\eqref{eq:abstrac_cost}, we can rewrite this cost as
$\sum_{\chead} c_0 (|V_X| + |E_X|)$, where the sum ranges over the
compacted heads $\chead$ associated with the spines $r$. We use the
following lemma to provide a lower bound on the number of $st$-paths
descending from $r$.

\begin{lemma}
  \label{lemma:lower_bound_paths_beadstring}
  Given a spine $r$, and its bead string $\beadstring$ with head
  $\head$, there are at least $|E_X| - |V_X| + 1$ $st$-paths in $G$
  that have prefix $\pi_s = s \leadsto u$ and suffix $u \leadsto t$
  internal to $\beadstring$, where the compacted head is $\chead =
  (V_X, E_X)$.
\end{lemma}
\begin{proof}
   $\chead$ is biconnected. In any biconnected graph $B = (V_B, E_B)$
   there are at least $|E_B| - |V_B| + 1$ $xy$-paths for any $x,y \in
   V_B$. Find an ear decomposition \cite{Diestel} of $B$ and consider
   the process of forming $B$ by adding ears one at the time, starting
   from a single cycle including $x$ and $y$. Initially 
   $|V_B|=|E_B|$ and there are 2 $xy$-paths. Each new ear forms a path
   connecting two vertices that are part of a $xy$-path, increasing
   the number of paths by at least 1. If the ear has $k$ edges, its
   addition increases $V$ by $k-1$, $E$ by $k$, and the number of
   $xy$-paths by at least 1. The result follows by induction.
\end{proof}

The implication of Lemma~\ref{lemma:lower_bound_paths_beadstring} is
that there are at least $|E_X| - |V_X| + 1$ leaves descending from the
given spine $r$. Hence, we can charge to each of them a cost of
$\frac{c_0 (|V_X| + |E_X|)}{|E_X| - |V_X| +
  1}$. Lemma~\ref{lemma:density} allows us to prove that the latter
cost is $O(1)$ when $\head$ is different from a single edge or a
cycle. (If $\head$ is a single edge or a cycle, $\chead$ is a
  single or double edge, and the cost is trivially a constant.)
\begin{lemma}
  \label{lemma:density}
  For a compacted head $\chead = (V_X, E_X)$, its density is
  $\frac{|E_X|}{|V_X|} \geq \frac{11}{10}$.
\end{lemma}
\begin{proof}
	Consider the following partition $V_X = \{r\} \cup V_2 \cup
	V_3$ where: $r$ is the root; $V_2$ is the set of vertices with
	degree 2 and; $V_3$, the vertices with degree $\geq 3$.  Since
	$H_X$ is compacted DFS tree of a biconnected graph, we have
	that $V_2$ is a \emph{subset} of the leaves and $V_3$ contains
	the set of internal vertices (except $r$). There are no vertices
	with degree 1 and $d(r) \geq 2$. Let $x = \sum_{v \in V_3}
	d(v)$ and $y = \sum_{v \in V_2} d(v)$.  We can write the
	density as a function of $x$ and $y$, namely,
	$$\frac{|E_X|}{|V_X|} = \frac{x + y + d(r)}{2 (|V_3| + |V_2| + 1)} $$

	Note that $|V_3| \le \frac{x}{3}$ as the vertices in $V_3$
	have at least degree 3, $|V_2| = \frac{y}{2}$ as vertices in
	$V_2$ have degree exactly 2. Since $d(r) \ge 2$, we derive the
	following bound
	$$\frac{|E_X|}{|V_X|} \ge \frac{x + y + 2}{\frac{2}{3}x + y + 2} $$

	Consider any graph with $|V_X|>3$ and its DFS tree rooted at
        $r$. Note that: (i) there are no tree edges between any two
        leaves, (ii) every vertex in $V_2$ is a leaf and (iii) no leaf
        is a child of $r$.  Therefore, every tree edge incident in a
        vertex of $V_2$ is also incident in a vertex of $V_3$. Since
        exactly half the incident edges to $V_2$ are tree edges (the
        other half are back edges) we get that $y \le 2x$.

	With $|V_X| \ge 3$ there exists at least one internal vertex in
	the DFS tree and therefore $x \ge 3$.
 \begin{alignat*}{2}
	 \text{minimize }\quad   & \frac{x + y + 2}{\frac{2}{3}x + y + 2}\ \\
	 \text{subject to }\quad  & 0 \le y \le 2x, \\
	 		    & x \ge 3.
  \end{alignat*}

 Since for any $x$ the function is minimized by the maximum $y$ s.t.
 $y \le 2x$ and for any $y$ by the minimum $x$, we get
 $$\frac{|E_X|}{|V_X|} \ge \frac{9x + 6}{8x + 6} \ge
 \frac{11}{10}.$$ 
\end{proof}

Specifically, let $\alpha = \frac{11}{10}$ and write $\alpha = 1 +
2/\beta$ for a constant $\beta$: we have that $|E_X| + |V_X| = (|E_X|
- |V_X|) + 2 |V_X| \leq (|E_X| - |V_X|) + \beta (|E_X| - |V_X|) =
\frac{\alpha+1}{\alpha-1} (|E_X| - |V_X|)$. Thus, we can charge each
leaf with a cost of $\frac{c_0 (|V_X| + |E_X|)}{|E_X| - |V_X| + 1}
\leq c_0 \frac{\alpha+1}{\alpha-1} = O(1)$. This motivates the
definition of $\chead$, since Lemma~\ref{lemma:density} does not
necessarily hold for the head $\head$ (due to the unary nodes in its
DFS tree).

One last step to bound $\sum_{\chead} c_0 (|V_X| + |E_X|)$: as noted
before, a root-to-leaf trace for the string storing $\pi$ has $|\pi|$
left branches by Lemma~\ref{lem:properties_recursion}, and as many
spines, each spine charging $c_0 \frac{\alpha+1}{\alpha-1} = O(1)$ to
the leaf at hand. This means that each of the $|\setofpaths_{s,t}(G)|$
leaves is charged for a cost of $O(|\pi|)$, thus bounding the sum as
$\sum_{r\,\mathrm{spine}} T(r) = \sum_{\chead} c_0 (|V_X| + |E_X|) =
O(\sum_{\pi \in \setofpaths_{s,t}(G)}{|\pi|})$. This completes the
proof of Lemma~\ref{lemma:total_cost_recursion_tree}. As a corollary,
we obtain the following result.

\begin{lemma}
  \label{lemma:amortized_cost_per_path}
  The recursion tree $R$ with cost as in Eq.~\eqref{eq:abstrac_cost}
  induces an $O(|\pi|)$ amortized cost for each $st$-path $\pi$.
\end{lemma}

\section{Certificate implementation and maintenance}
\label{sec:certificate}

We show how to represent and update the certificate $C$ so that the
time taken by Algorithm~\ref{alg:liststpaths} in the recursion tree
can be distributed among the nodes as in Eq.~\eqref{eq:abstrac_cost}:
namely, (i)~$O(1)$ time in unary nodes; (ii)~$O(|\pi|)$ time in the
leaf corresponding to path $\pi$; (iii)~$O(|V_X| + |E_X|)$ in each
spine.

The certificate $C$ associated with a node $\langle \pi_s, u, C
\rangle$ in the recursion tree is a compacted and augmented DFS tree
of bead string $\beadstring$, rooted at vertex~$u$. The DFS tree
changes over time along with $\beadstring$, and is
maintained in such a way that $t$ is in the leftmost path of the tree.

We augment the DFS tree with sufficient information to implicitly
represent articulation points and biconnected components.  
Recall that, by the properties of a DFS tree of an undirected graph,
there are no cross edges (edges between different branches of the
tree).

We compact the DFS tree by contracting the vertices that have
degree~2, except $u$, $t$, and the leaves (the latter surely have
incident back edges). Maintaining this compacted representation is not
a difficult data-structure problem. From now on we can assume
w.l.o.g.\mbox{} that $C$ is an augmented DFS tree rooted at $u$ where
internal nodes of the DFS tree have degree $\ge 3$, and each vertex
$v$ has associated the following information.

\begin{enumerate}
	\setlength{\itemsep}{0pt} 
	\item A doubly-linked list $lb(v)$ of back edges linking $v$
          to its descendants $w$ sorted by postorder DFS numbering.
	\item A doubly-linked list $ab(v)$ of back edges linking $v$
          to its ancestors $w$ sorted by preorder DFS numbering.
        \item \label{item:point3} An integer $\gamma(v)$, such that if $v$ is an
                ancestor of $w$ then $\gamma(v) < \gamma(w)$.
	\item \label{item:point4} The smallest $\gamma(w)$ over all
          $w$, such that $(h,w)$ is a back edge and $h$ is in the
          subtree of $v$, denoted by $\mathit{lowpoint}(v)$.
\end{enumerate}

Given three vertices $v,w,x \in C$ such that $v$ is the parent of $w$
and $x$ is not in the subtree\footnote{The second condition is always
  satisfied when $w$ is not in the leftmost path, since $t$ is not in
  the subtree of $w$.} of $w$, we can efficiently test if $v$ is an
articulation point, i.e.\mbox{} $\mathit{lowpoint}(w) \leq
\gamma(v)$. (Note that we adopt a variant of $\mathit{lowpoint}$ using
$\gamma(v)$ in place of the preorder numbering~\cite{Tarjan72}: it has
the same effect whereas using $\gamma(v)$ is preferable since it is
easier to dynamically maintain.)

\begin{lemma}
  \label{lem:certificate_scratch}
 The certificate associated with the root of the recursion can be
 computed in $O(m)$ time.
\end{lemma}

\begin{proof}
	In order to set $t$ to be in the leftmost path, we perform a
        DFS traversal of graph $G$ starting from $s$ and stop when we
        reach vertex $t$. We then compute the DFS tree, traversing the
        path $s \leadsto t$ first. When visiting vertex $v$, we set
        $\gamma(v)$ to depth of $v$ in the DFS. Before going up on the
        traversal, we compute the lowpoints using the lowpoints of the
        children. Let $z$ be the parent of $v$. If
        $\mathit{lowpoint}(v) \leq \gamma(z)$ and $v$ is not in the
        leftmost path in the DFS, we cut the subtree of $v$ as it does
        not belong to $B_{s,t}$.  When first exploring the
        neighborhood of $v$, if $w$ was already visited,
        i.e. $e=(u,w)$ is a back edge, and $w$ is a descendant of $v$;
        we add $e$ to $ab(w)$. This maintains the DFS preordering in
        the ancestor back edge list. Now, after the first scan of
        $N(v)$ is over and all the recursive calls returned (all the
        children were explored), we re-scan the neighborhood of
        $v$. If $e=(v,w)$ is a back edge and $w$ is an ancestor of
        $v$, we add $e$ to $lb(w)$. This maintains the DFS
        postorder in the descendant back edge list. This procedure
        takes at most two DFS traversals in $O(m)$ time.  This DFS
        tree can be compacted in the same time bound.  
\end{proof}

\begin{lemma}
	\label{lem:choose}
	Operation $\chooseedge(C,u)$ can be implemented in $O(1)$ time.
\end{lemma}
\begin{proof}
If the list $lb(v)$ is empty, return the tree edge $e=(u,v)$ linking $u$
to its only child $v$ (there are no other children).  Else,
return the last edge in $lb(v)$. 
\end{proof}

We analyze the cost of updating and restoring the
certificate $C$. We can reuse parts of~$C$,
namely, those corresponding to the vertices that are not in the
compacted head $H_X = (V_X,E_X)$ as defined in
Section~\ref{sub:recursion-tree-cost}.
%
%
%
We prove that, given a unary node $u$ and its tree edge $e=(u,v)$, the
subtree of $v$ in~$C$ can be easily made a certificate for the left
branch of the recursion.

\begin{lemma}
	\label{lem:unary_left}
	On a unary node, $\oracleleft(C,e)$ takes $O(1)$
	time.
\end{lemma}
\begin{proof}
	Take edge $e=(u,v)$. Remove edge $e$ and set $v$ as the root
	of the certificate. Since $e$ is the only edge incident in
	$v$, the subtree $v$ is still a DFS tree. Cut the list of children
	of $v$ keeping only the first child. (The other children are no
	longer in the bead string and become part of $I$.) There is
	no need to update $\gamma(v)$. 
\end{proof}

We now devote the rest of this section to show how to efficiently
maintain $C$ on a spine.  Consider removing a back edge $e$ from $u$:
the compacted head $H_X=(V_X,E_X)$ of the bead string can be divided
into smaller biconnected components.  Many of those can be excluded
from the certificate (i.e. they are no longer in the new bead string,
and so they are bookkept in $I$) and additionally we have to update
the lowpoints that change. We prove that this operation can be
performed in $O(|V_X|)$ total time on a spine of the recursion tree.

\begin{lemma}
	\label{lem:removebackedge}
	The total cost of all the operations $\oracleright(C,e)$ in a
        spine is $O(|V_X|)$ time.
\end{lemma}
\begin{proof}
	In the right branches along a spine, we remove all back edges
        in $lb(u)$. This is done by starting from the last edge in
        $lb(u)$, i.e.~proceeding in reverse DFS postorder.
	For back edge $b_i = (z_i,u)$, we traverse the vertices in the
	path from $z_i$ towards the root $u$, as these are the only lowpoints
	that can change.
	While moving upwards on the tree, on each vertex $w$, we
        update $\mathit{lowpoint}(w)$. This is done by taking the
        endpoint $y$ of the first edge in $ab(w)$ (the back edge that
        goes the topmost in the tree) and choosing the minimum between
        $\gamma(y)$ and the lowpoint of each child\footnote{If
          $\mathit{lowpoint}(w)$ does not change we cannot pay to
          explore its children. For each vertex we dynamically
          maintain a list $l(w)$ of its children that have lowpoint
          equal to $\gamma(u)$. Then, we can test in constant time if
          $l(w) \neq \emptyset$ and $y$ is not the root $u$. If both
          conditions are true $\mathit{lowpoint}(w)$ changes,
          otherwise it remains equal to $\gamma(u)$ and we stop.} of
        $w$. We stop when the updated $\mathit{lowpoint}(w) =
        \gamma(u)$ since it implies that the lowpoint of the vertex
        can not be further reduced.  Note that we stop before $u$,
        except when removing the last back edge in $lb(u)$.

	To prune the branches of the DFS tree that are no longer in
        $B_{u,t}$, consider again each vertex $w$ in the path from
        $z_i$ towards the root $u$ and its parent $y$. We check if the
        updated $\mathit{lowpoint}(w) \leq \gamma(y)$ and $w$ is not
        in the leftmost path of the DFS. If both conditions are
        satisfied, we have that $w \notin B_{u,t}$, and therefore we
        cut the subtree of $w$ and keep it in $I$ to restore later. We
        use the same halting criterion as in the previous paragraph.
	
	The cost of removing all back edges in the spine is
	$O(|V_X|)$: there are $O(|V_X|)$ tree edges and, in the paths
	from $z_i$ to $u$, we do not traverse the same tree edge twice
	since the process described stops at the first common ancestor
	of endpoints of back edges $b_i$. Additionally, we take $O(1)$
	time to cut a subtree of an articulation point in the DFS tree. 
\end{proof}

To compute $\oracleleft(C,e)$ in the binary nodes of a spine, we use
the fact that in every left branching from that spine, the graph is
the same (in a spine we only remove edges incident to $u$ and on a
left branch from the spine we remove the vertex $u$) and therefore its
block tree is also the same. However, the certificates on these nodes
are not the same, as they are rooted at different vertices. Using
the reverse DFS postorder of the edges, we are able to traverse
each edge in $H_X$ only a constant number of times in the spine.

\begin{lemma}
	\label{lem:promotebackedge}
	The total cost of all operations $\oracleleft(C,e)$ in a
        spine is amortized $O(|E_X|)$.
\end{lemma}
\begin{proof}
	Let $t'$ be the last vertex in the path $u \leadsto t$ s.t.
	$t' \in V_X$. Since $t'$ is an articulation point, the subtree
        of the DFS tree rooted in $t'$ is maintained in the
	case of removal of vertex $u$. Therefore the only modifications
	of the DFS tree occur in the compacted head $H_X$ of $B_{u,t}$.
	Let us compute the certificate $C_i$: this is the certificate
	of the left branch of the $i$th node of the spine where we
	augment the path with the back edge $b_i = (z_i,u)$ of
	$lb(u)$ in the order defined by $\chooseedge(C,u)$.

	For the case of $C_1$, we remove $u$ and rebuild the
	certificate starting form $z_1$ (the last edge in $lb(u)$)
	using the algorithm from Lemma~\ref{lem:certificate_scratch}
	restricted to $H_X$ and using $t'$ as target and $\gamma(t')$
	as a baseline to $\gamma$ (instead of the depth). This takes
	$O(|E_X|)$ time.  

	For the general case of $C_i$ with $i>1$ we also rebuild
	(part) of the certificate starting from $z_i$ using the
	procedure from Lemma~\ref{lem:certificate_scratch} but we use
	information gathered in $C_{i-1}$ to avoid exploring useless
	branches of the DFS tree. The key point is that, when we reach
	the first bead in common to both $B_{z_i,t}$ and
	$B_{z_{i-1},t}$, we only explore edges internal to this bead.
	If an edge $e$ leaving the bead leads to $t$, we can reuse
	a subtree of $C_{i-1}$. If $e$ does not lead to $t$, then it
	has already been explored (and cut) in $C_{i-1}$ and there is
	no need to explore it again since it will be discarded.
	Given the order we take $b_i$, each bead is
	not added more than once, and the total cost over the
	spine is $O(|E_X|)$.

	Nevertheless, the internal edges $E_X'$ of the first bead in
	common between $B_{z_i,t}$ and $B_{z_{i-1},t}$ can be explored
	several times during this procedure.\footnote{Consider the
	case where $z_i, \ldots, z_j$ are all in the same bead after
	the removal of $u$. The bead strings are the same, but the
	roots $z_i, \ldots, z_j$ are different, so we have to compute
	the corresponding DFS of the first component $|j-i|$ times.}
	We can charge the cost $O(|E'_X|)$ of exploring those edges to
	another node in the recursion tree, since this common bead is
	the head of at least one certificate in the recursion subtree
	of the left child of the $i$th node of the spine.
	Specifically, we charge the first node in the \emph{leftmost}
	path of the $i$th node of the spine that has exactly the
	edges $E'_X$ as head of its bead string: (i) if $|E'_X| \le 1$
	it corresponds to a unary node or a leaf in the recursion tree
	and therefore we can charge it with $O(1)$ cost; (ii)
	otherwise it corresponds to a first node of a spine and
	therefore we can also charge it with $O(|E'_X|)$. We use this
	charging scheme when $i \neq 1$ and the cost is always charged
	in the leftmost recursion path of $i$th node of the spine.
	Consequently, we never charge a node in the recursion tree more
	than once.  
\end{proof}

\begin{lemma}
	\label{lem:restore} On each node of the recursion tree,
	$\undooracle(C,I)$ takes time proportional to the size of the
	modifications kept in $I$.
\end{lemma}
\begin{proof}
	We use standard data structures (i.e.~linked lists) for the
        representation of certificate $C$.  Persistent versions of
        these data structures exist that maintain a stack of
        modifications applied to them and that can restore its
        contents to their previous states.  Given the modifications in
        $I$, these data structures take $O(|I|)$ time to restore the
        previous version of $C$.

        Let us consider the case of performing $\oracleleft(C,e)$. We
        cut at most $O(|V_X|)$ edges from $C$. Note that, although we
        conceptually remove whole branches of the DFS tree, we only
        remove edges that attach those branches to the DFS tree. The
        other vertices and edges are left in the certificate but, as
        they no longer remain attached to $B_{u,t}$, they will never
        be reached or explored. In the case of $\oracleright(C,e)$, we
        have a similar situation, with at most $O(|E_X|$) edges being
        modified along the spine of the recursion tree. 
\end{proof}

From Lemmas~\ref{lem:choose} and
\ref{lem:removebackedge}--\ref{lem:restore}, it follows that on a
spine of the recursion tree we have the costs:
$\chooseedge(u)$ on each node which is bounded by $O(|V_X|)$ time as there
are at most $|V_X|$ back edges in $u$; $\oracleright(C,e)$,
$\undooracle(C,I)$ take $O(|V_X|)$ time; $\oracleleft(C,e)$ and
$\undooracle(C,I)$ are charged $O(|V_X|+|E_X|)$ time.  We thus have the following
result,  completing the proof of Theorem~\ref{theorem:optimal_paths}.

\begin{lemma}
	\label{lem:algo_cost}
	Algorithm~$\ref{alg:liststpaths}$ can be implemented with a
        cost fulfilling Eq.~\eqref{eq:abstrac_cost}, thus it takes
        total \mbox{$O(m+\sum_{r \in R} T(r)) = O(m+\sum_{\pi \in
	\setofpaths_{s,t}(\beadstring)}{|\pi|})$} time.
\end{lemma}

\section{Extended analysis of operations}
\label{app:extend-analys-oper}

In this suplementary section, we present all details and illustrate
with figures the operations $\oracleright(C,e)$ and $\oracleleft(C,e)$
that are performed along a spine of the recursion tree. In order to
better detail the procedures in Lemma~\ref{lem:removebackedge} and
Lemma~\ref{lem:promotebackedge}, we divide them in smaller parts.  We
use bead string $B_{u,t}$ from Fig.~\ref{fig:Certificate} and the
respective spine from Fig.~\ref{fig:spine} as the base for the
examples. This spine contains four binary nodes corresponding to the
back edges in $lb(u)$ and an unary node corresponding to the tree edge
$(u,v)$. Note that edges are taken in order of the endpoints
$z_1,z_2,z_3,z_4,v$ as defined in operation $\chooseedge(C,u)$.

By Lemma~\ref{lemma:beadstring}, the impact of operations
$\oracleright(C,e)$ and $\oracleleft(C,e)$ in the certificate is
restricted to the biconnected component of $u$. Thus we mainly focus
on maintaining the compacted head $H_X = (V_X,E_X)$ of the bead
string~$B_{u,t}$.

\subsection{Operation right\_update(C,e)} 

\begin{figure*}[!t]
\centering
\subfigure[Step 1]{
\begin{tikzpicture}
[nodeDecorate/.style={shape=circle,inner sep=1pt,draw,thick,fill=black},%
  lineDecorate/.style={-,dashed},%
  elipseDecorate/.style={color=gray!30},
  scale=0.18]

\draw (10,22) circle (9);
\draw (5,11.1) circle (3);

\node (s) at (10,34) [nodeDecorate,color=lightgray,label=above left:$s$] {};
\node (u) at (10,31) [nodeDecorate,label=above left:{ $u$}] {};

\node (tp) at (6.2,13.9) [nodeDecorate,label=right:{ $z_4$}] {};
\node (t) at (5.5,11) [nodeDecorate,label=left:$t$] {};

\path {
	(s) edge[snake,-,color=lightgray] node {\quad\quad$\pi_s$} (u)
	(u) edge node {} (tp)
	(tp) edge node {} (t)
};

\node (a) at (9.3,28) [nodeDecorate,label=left:$v$] {};
\node (b) at (7.1,18) [nodeDecorate,] {};
\node (c) at (13,25) [nodeDecorate,] {};
\node (d) at (14.5,22) [nodeDecorate,label=right:$z_1$] {};
\node (e) at (11,22) [nodeDecorate,label=below:$z_3$] {};
\node (f) at (16,19) [nodeDecorate,label=below:$z_2$] {};
\node (g) at (13,19) [nodeDecorate,] {};
\node (h) at (10,16) [nodeDecorate,] {};

\path {
	(a) edge node {} (c)
	(c) edge node {} (d)
	(d) edge node {} (f)
	(c) edge node {} (e)
	(d) edge node {} (g)
	(b) edge node {} (h)
};

\path {
	(u) edge[dashed,bend left=-40] node {} (tp)
	(f) edge[dashed,bend left=-40] node {} (u)
	(g) edge[dashed,bend left=10] node {} (c)
	(e) edge[dashed,bend left=-10] node {} (u)
	(a) edge[dashed,bend left=5] node {} (h)
};
\end{tikzpicture}
}
\subfigure[Step 2]{
\begin{tikzpicture}
[nodeDecorate/.style={shape=circle,inner sep=1pt,draw,thick,fill=black},%
  lineDecorate/.style={-,dashed},%
  elipseDecorate/.style={color=gray!30},
  scale=0.18]
\draw (10,22) circle (9);
\draw (5,11.1) circle (3);

\node (s) at (10,34) [nodeDecorate,color=lightgray,label=above left:$s$] {};
\node (u) at (10,31) [nodeDecorate,label=above left:{ $u$}] {};

\node (tp) at (6.2,13.9) [nodeDecorate,label=right:{ $z_4$}] {};
\node (t) at (5.5,11) [nodeDecorate,label=left:$t$] {};

\path {
	(s) edge[snake,-,color=lightgray] node {\quad\quad$\pi_s$} (u)
	(u) edge node {} (tp)
	(tp) edge node {} (t)
};

\node (a) at (9.3,28) [nodeDecorate,label=left:$v$] {};
\node (b) at (7.1,18) [nodeDecorate,] {};
\node (c) at (13,25) [nodeDecorate,] {};
\node (e) at (11,22) [nodeDecorate,label=below:$z_3$] {};
\node (h) at (10,16) [nodeDecorate,] {};

\path {
	(a) edge node {} (c)
	(c) edge node {} (e)
	(b) edge node {} (h)
};

\path {
	(u) edge[dashed,bend left=-40] node {} (tp)
	(e) edge[dashed,bend left=-10] node {} (u)
	(a) edge[dashed,bend left=5] node {} (h)
};

\end{tikzpicture}
}
\subfigure[Step 3]{
\begin{tikzpicture}
[nodeDecorate/.style={shape=circle,inner sep=1pt,draw,thick,fill=black},%
  lineDecorate/.style={-,dashed},%
  elipseDecorate/.style={color=gray!30},
  scale=0.18]
\draw (10,22) circle (9);
\draw (5,11.1) circle (3);

\node (s) at (10,34) [nodeDecorate,color=lightgray,label=above left:$s$] {};
\node (u) at (10,31) [nodeDecorate,label=above left:{ $u$}] {};

\node (tp) at (6.2,13.9) [nodeDecorate,label=right:{ $z_4$}] {};
\node (t) at (5.5,11) [nodeDecorate,label=left:$t$] {};

\path {
	(s) edge[snake,-,color=lightgray] node {\quad\quad$\pi_s$} (u)
	(u) edge node {} (tp)
	(tp) edge node {} (t)
};

\node (a) at (9.3,28) [nodeDecorate,label=left:$v$] {};
\node (b) at (7.1,18) [nodeDecorate,] {};
\node (h) at (10,16) [nodeDecorate,] {};

\path {
	(b) edge node {} (h)
};

\path {
	(u) edge[dashed,bend left=-40] node {} (tp)
	(a) edge[dashed,bend left=5] node {} (h)
};

\end{tikzpicture}
}
\subfigure[Step 4 (final)]{

\begin{tikzpicture}
[nodeDecorate/.style={shape=circle,inner sep=1pt,draw,thick,fill=black},%
  lineDecorate/.style={-,dashed},%
  elipseDecorate/.style={color=gray!30},
  scale=0.18]
\draw (10,22) circle (9) [color=gray!30];
\draw (5,11.1) circle (3);
\draw (9.6,29.5) circle (1.5);
\draw (8,23) circle (5);
\draw (6.7,15.9) circle (2.1);

\node (s) at (10,34) [nodeDecorate,color=lightgray,label=above left:$s$] {};
\node (u) at (10,31) [nodeDecorate,label=above left:{ $u$}] {};

\node (tp) at (6.2,13.9) [nodeDecorate,label=right:{ $z_4$}] {};
\node (t) at (5.5,11) [nodeDecorate,label=left:$t$] {};

\path {
	(s) edge[snake,-,color=lightgray] node {\quad\quad$\pi_s$} (u)
	(u) edge node {} (tp)
	(tp) edge node {} (t)
};

\node (a) at (9.3,28) [nodeDecorate,label=below left:$v$] {};
\node (b) at (7.1,18) [nodeDecorate,] {};
\node (h) at (10,19) [nodeDecorate,] {};

\path {
	(b) edge node {} (h)
};

\path {
	(a) edge[dashed,bend left=5] node {} (h)
};

\end{tikzpicture}
}

\caption{Application of $\oracleright(C,e)$ on a spine of the recursion tree}
\label{fig:oracleleftexample}
\end{figure*}
\begin{lemma}
	\emph{(Lemma~\ref{lem:removebackedge} restated)} In a spine of
	the recursion tree, operations $\oracleright(C,e)$ can be
	implemented in $O(|V_X|)$ total time.
\end{lemma}

	In the right branches along a spine, we remove all back edges
	in $lb(u)$. This is done by starting from the last edge in
	$lb(u)$, i.e. proceeding in reverse DFS postorder. In the
	example from Fig.~\ref{fig:Certificate}, we remove the back
	edges $(z_1,u) \ldots (z_4,u)$. To update the certificate
	corresponding to $B_{u,t}$, we have to (i) update the
	lowpoints in each vertex of $H_X$; (ii) prune vertices that
	cease to be in $B_{u,t}$ after removing a back edge. For a
	vertex $w$ in the tree, there is no need to update~$\gamma(w)$.

	Consider the update of lowpoints in the DFS tree.  For
	a back edge $b_i = (z_i,u)$, we traverse the vertices in the
	path from $z_i$ towards the root $u$. By definition of
	lowpoint, these are the only lowpoints that can change.
	Suppose that we remove back edge $(z_4,u)$ in the example from
	Fig.~\ref{fig:Certificate}, only the lowpoints of the vertices
	in the path from $z_4$ towards the root $u$ change.
	Furthermore, consider a vertex $w$ in the tree that is an
	ancestor of at least two endpoints $z_i, z_j$ of back edges
	$b_i$, $b_j$. The lowpoint of $w$ does not change when we
	remove $b_i$.  These observations lead us to the following
	lemma.

\begin{lemma}
	In a spine of the recursion tree, the update of lowpoints in
	the certificate by operation $\oracleright(C,e)$ can be done
	in $O(|V_X|)$ total time.  \label{lem:lowpoints}
\end{lemma}
\begin{proof}
	Take each back edge $b_i = (z_i,u)$ in the order defined by
        $\chooseedge(C,u)$. Remove $b_i$ from $lb(u)$ and $ab(z_i)$.
        Starting from $z_i$, consider each vertex $w$ in the path from
        $z_i$ towards the root $u$.  On vertex $w$, we update
        $\mathit{lowpoint}(w)$ using the standard procedure: take the
        endpoint $y$ of the first edge in $ab(w)$ (the back edge that
        goes the nearest to the root of the tree) and choosing the
        minimum between $\gamma(y)$ and the lowpoint of each child of
        $w$.  When the updated $\mathit{lowpoint}(w) = \gamma(u)$, we
        stop examining the path from $z_i$ to $u$ since it implies
        that the lowpoint of the vertex can not be further reduced
        (i.e. $w$ is both an ancestor to both $z_i$ and
        $z_{i+1}$).

	If $\mathit{lowpoint}(w)$ does not change we cannot pay to
	explore its children. In order to get around this, for each
	vertex we dynamically maintain, throughout the spine, a list
	$l(w)$ of its children that have lowpoint equal to
	$\gamma(u)$. Then, we can test in constant time if $l(w) \neq
	\emptyset$ and $y$ (the endpoint of the first edge in $ab(w)$)
	is not the root $u$. If both conditions are satisfied
	$\mathit{lowpoint}(w)$ changes, otherwise it remains equal to
	$\gamma(u)$ and we stop. The total time to create the lists is
	$O(|V_X|)$ and the time to update is bounded by the number of
	tree edges traversed, shown to be $O(|V_X|)$ in the next
	paragraph.

	The cost of updating the lowpoints when removing all
	back edges $b_i$ is $O(|V_X|)$: there are $O(|V_X|)$ tree
	edges and we do not traverse the same tree edge twice since
	the process described stops at the first common ancestor of
	endpoints of back edges $b_i$ and $b_{i+1}$. By contradiction:
	if a tree edge $(x,y)$ would be traversed twice when removing
	back edges $b_i$ and $b_{i+1}$, it would imply that both $x$
	and $y$ are ancestors of $z_i$ and $z_{i+1}$ (as edge $(x,y)$
	is both in the path $z_i$ to $u$ and the path $z_{i+1}$ to
	$u$) but we stop at the first ancestor of $z_i$ and
	$z_{i+1}$.
\end{proof}

Let us now consider the removal of vertices that
are no longer in $B_{u,t}$ as consequence of operation
$\oracleright(C,e)$ in a spine of the recursion tree. By removing a
back edge $b_i = (z_i,u)$, it is possible that a vertex $w$ previously
in $H_X$ is no longer in the bead string $B_{u,t}$ (e.g. $w$ is no
longer biconnected to $u$ and thus there is no simple path $u \leadsto
w \leadsto t$).

\begin{lemma}
	In a spine of the recursion tree, the branches of the DFS that
	are no longer in $B_{u,t}$ due to operation
	$\oracleright(C,e)$ can be removed from the certificate in
	$O(|V_X|)$ total time.
	\label{lem:cutbranches}
\end{lemma}
\begin{proof}
	To prune the branches of the DFS tree that are no longer in
        $H_X$, consider again each vertex $w$ in the path from $z_i$
        towards the root $u$ and the vertex $y$, parent of $w$. It is
        easy to check if $y$ is an articulation point by verifying if
        the updated $\mathit{lowpoint}(w) \leq \gamma(y)$ and there
        exists $x$ not in the subtree of $w$. If $w$ is not in the
        leftmost path, then $t$ is not in the subtree of $w$. If that
        is the case, we have that $w \notin B_{u,t}$, and therefore we
        cut the subtree of $w$ and bookkeep it in $I$ to restore
        later. Like in the update the lowpoints, we stop examining the
        path $z_i$ towards $u$ in a vertex $w$ when
        $\mathit{lowpoint}(w) = \gamma(u)$ (the lowpoints and
        biconnected components in the path from $w$ to $u$ do not
        change).  When cutting the subtree of $w$, note that there are
        no back edges connecting it to $B_{u,t}$ ($w$ is an
        articulation point) and therefore there are no updates to the
        lists $lb$ and $ab$ of the vertices in $B_{u,t}$.  Like in the
        case of updating the lowpoints, we do not traverse the same
        tree edge twice (we use the same halting criterion). 
\end{proof}

With Lemma~\ref{lem:lowpoints} and Lemma~\ref{lem:cutbranches} we
finalize the proof of Lemma~\ref{lem:removebackedge}. 
Fig.~\ref{fig:oracleleftexample} shows the changes the bead string $B_{u,t}$
from Fig.~\ref{fig:Certificate} goes through in the corresponding spine
of the recursion tree.

\subsection{Operation~left\_update(C,e)} 
In the binary nodes of a spine, we use
the fact that in every left branching from that spine the graph is
the same (in a spine we only remove edges incident to $u$ and on a
left branch from the spine we remove the vertex $u$) and therefore its
block tree is also the same. In Fig.~\ref{fig:blocktree_without_u},
we show the resulting block tree of the graph from
Fig.~\ref{fig:Certificate} after having removed vertex $u$. However,
the certificates on these left branches are not the same, as they are
rooted at different vertices. In the example we must compute the
certificates $C_1 \ldots C_4$ corresponding to bead strings $B_{z_1,t}
\ldots B_{z_4,t}$. We do not account for the cost of the left
branch on the last node of spine (corresponding to $B_{v,t}$) as the
node is unary and we have shown in Lemma~\ref{lem:unary_left} how to
maintain the certificate in $O(1)$ time.

By using the reverse DFS postorder of the back edges, we are able
to traverse each edge in $H_X$ only an amortized constant number of
times in the spine.
\begin{lemma}
	\emph{(Lemma~\ref{lem:promotebackedge} restated)} The calls to
	operation $\oracleleft(C,e)$ in a spine of the recursion tree
	can be charged with a time cost of $O(|E_X|)$ to that spine. 
\end{lemma}

To achieve this time cost, for each back edge $b_i = (z_i,u)$, we
compute the certificate corresponding to $B_{z_i,t}$ based on the
certificate of $B_{z_{i-1},t}$. Consider the compacted head $H_X =
(V_X , E_X )$ of the bead string $B_{u,t}$. We use $O(|E_X|)$ time to
compute the first certificate $C_1$ corresponding to bead string
$B_{z_1,t}$. Fig.~\ref{fig:spine-left-update} shows bead string
$B_{z_1,t}$ from the example of Fig.~\ref{fig:Certificate}.

\begin{lemma}
	The certificate $C_1$, corresponding to bead string
	$B_{z_1,t}$, can be computed in $O(|E_X|)$ time.
\end{lemma}
\begin{proof}
	Let $t'$ be the last vertex in the path $u \leadsto t$ s.t.
	$t' \in V_X$. Since $t'$ is an articulation point, the subtree
	of the DFS tree rooted in $t'$ is maintained in the case of
	removal vertex $u$. Therefore the only modifications of the
	DFS tree occur in head $H_X$ of $B_{u,t}$.

	To compute $C_1$, we remove $u$ and rebuild the certificate
	starting form $z_1$ using the algorithm from
	Lemma~\ref{lem:certificate_scratch} restricted to $H_X$ and
	using $t'$ as target and $\gamma(t')$ as a baseline to
	$\gamma$ (instead of the depth). In particular we do
	the following.
	To set $t'$ to be in the leftmost path, we perform a
	DFS traversal of graph $H_X$ starting from $z_1$ and stop when
	we reach vertex $t'$. Then compute the DFS tree, traversing
	the path $z_1 \leadsto t'$ first.
	
	{\it Update of $\gamma$.} For each tree edge $(v,w)$ in the
	$t' \leadsto z_1$ path, we set $\gamma(v)=\gamma(w)-1$, using
	$\gamma(t')$ as a baseline.  During the rest of the traversal,
	when visiting vertex $v$, let $w$ be the parent of $v$ in the
	DFS tree. We set $\gamma(v)=\gamma(w)+1$. This maintains the
	property that $\gamma(v)>\gamma(w)$ for any $w$ ancestor of
	$v$.

	{\it Lowpoints and pruning the tree.}  Bottom-up in
	the DFS-tree, compute the lowpoints using the
	lowpoints of the children.  For $z$ the parent of $v$, if
	$\mathit{lowpoint}(v) \leq \gamma(z)$ and $v$ is not in the
	leftmost path in the DFS, cut the subtree of $v$ as it does
	not belong to~$B_{z_1,t}$.

	{\it Computing $lb$ and $ab$.} In the traversal, when finding
	a back edge $e=(v,w)$, if $w$ is a descendant of $v$ we append
	$e$ to $ab(w)$. This maintains the DFS preorder in the
	ancestor back edge list. After the first scan of $N(v)$ is
	over and all the recursive calls returned, re-scan the
	neighborhood of $v$. If $e=(v,w)$ is a back edge and $w$ is an
	ancestor of $v$, we add $e$ to $lb(w)$. This maintains the DFS
	postorder in the descendant back edge list.  This procedure
	takes $O(|E_X|)$ time.
\end{proof}

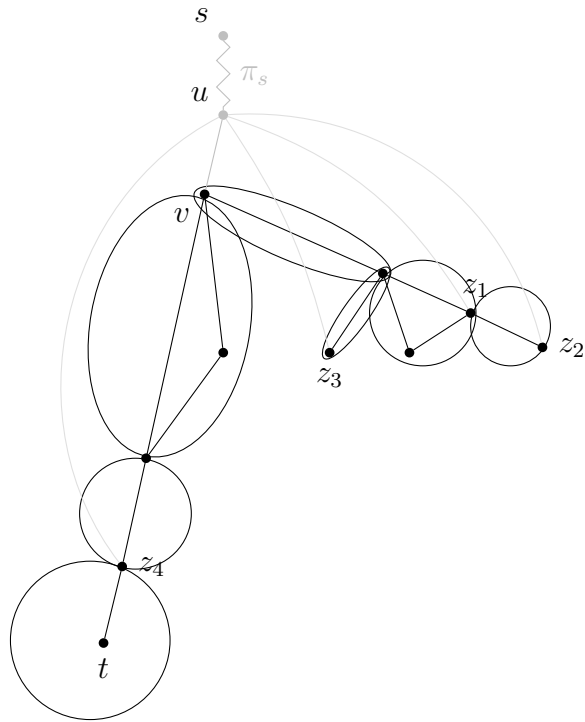
\begin{figure}[t!]
\centering
\begin{tikzpicture}
[nodeDecorate/.style={shape=circle,inner sep=1pt,draw,thick,fill=black},%
  lineDecorate/.style={-,dashed},%
  elipseDecorate/.style={color=gray!30},
  scale=0.35]
\draw (5,11.1) circle (3);
\draw[rotate around={-10:(8,23)}] (8,23) ellipse (3 and 5);
\draw[rotate around={-23:(12.6,26.5)}] (12.6,26.5) ellipse (4 and 1);
\draw[rotate around={55:(15,23.5)}] (15,23.5) ellipse (2.1 and 0.5);
\draw (6.7,15.9) circle (2.1);
\draw (17.5,23.5) circle (2.0);
\draw (20.8,23) circle (1.5);

\node (s) at (10,34) [nodeDecorate,color=lightgray,label=above left:$s$] {};
\node (u) at (10,31) [nodeDecorate,color=lightgray,label=above left:{ $u$}] {};

\node (tp) at (6.2,13.9) [nodeDecorate,label=right:{ $z_4$}] {};
\node (t) at (5.5,11) [nodeDecorate,label=below:$t$] {};
\node (a) at (9.3,28) [nodeDecorate,label=below left:$v$] {};

\path {
	(s) edge[snake,-,color=lightgray] node {\quad\quad$\pi_s$} (u)
	(u) edge[color=lightgray] node {} (a)
	(a) edge node {} (tp)
	(tp) edge node {} (t)
};

\node (b) at (7.1,18) [nodeDecorate,] {};
\node (c) at (16,25) [nodeDecorate,] {};
\node (d) at (19.3,23.5) [nodeDecorate,label=above:$~z_1$] {};
\node (e) at (14,22) [nodeDecorate,label=below:$z_3$] {};
\node (f) at (22,22.2) [nodeDecorate,label=right:$z_2$] {};
\node (g) at (17,22) [nodeDecorate,] {};
\node (h) at (10,22) [nodeDecorate,] {};

\path {
	(a) edge node {} (c)
	(c) edge node {} (d)
	(d) edge node {} (f)
	(c) edge node {} (e)
	(d) edge node {} (g)
	(b) edge node {} (h)
};

\path {
	(g) edge node {} (c)
	(a) edge node {} (h)


	(u) edge[bend left=-50,color=lightgray!50] node {} (tp)
	(f) edge[bend left=-40,color=lightgray!50] node {} (u)
	(e) edge[bend left=-10,color=lightgray!50] node {} (u)
	(d) edge[bend left=-20,color=lightgray!50] node {} (u)
};

\end{tikzpicture}
\caption{Block tree after removing vertex $u$}
\label{fig:blocktree_without_u}
\end{figure}

To compute each certificate $C_i$, corresponding to bead string
$B_{z_i,t}$, we are able to avoid visiting most of the edges that
belong $B_{z_{i-1},t}$. Since we take $z_i$ in reverse DFS
postorder, on the spine of the recursion we visit $O(|E_X|)$ edges
plus a term that can be amortized.

\begin{lemma}
	For each back edge $b_i = (z_i,u)$ with $i>1$, let ${E_X}_i'$
	be the edges in the first bead in common between $B_{z_i,t}$
	and $B_{z_{i-1},t}$. The total cost of computing all
	certificates $B_{z_i,t}$ in a spine of the recursion tree is:
	$O(|E_X| + \sum_{i>1}{|{E_X}_i'|})$.
	\label{lem:cost-spine-not-amortized}
\end{lemma}
\begin{proof}
	Let us compute the certificate $C_i$: the certificate
	of the left branch of the $i$th node of the spine where we
	augment the path with back edge $b_i = (z_i,u)$ of
	$lb(u)$.

	For the general case of $C_i$ with $i>1$ we also rebuild
	(part) of the certificate starting from $z_i$ using the
	procedure from Lemma~\ref{lem:certificate_scratch} but we use
	information gathered in $C_{i-1}$ to avoid exploring useless
	branches of the DFS tree. The key point is that, when we reach
	the first bead in common to both $B_{z_i,t}$ and
	$B_{z_{i-1},t}$, we only explore edges internal to this bead.
	If an edge $e$ that leaves the bead leads to $t$, we can reuse
	a subtree of $C_{i-1}$. If $e$ does not lead to $t$, then it
	has already been explored (and cut) in $C_{i-1}$ and there is
	no need to explore it again since it is going to be discarded.

	In detail, we start computing a DFS from $z_i$ in $B_{u,t}$
	until we reach a vertex $t' \in B_{z_{i-1},t}$. Note that the
	bead of $t'$ has one entry point and one exit point in
	$C_{i-1}$. After reaching $t'$ we proceed with the traversal
	using only edges already in $C_{i-1}$. When arriving at a
	vertex $w$ that is not in the same bead of $t'$, we stop the
	traversal. If $w$ is in a bead towards $t$, we reuse the
	subtree of $w$ and use $\gamma(w)$ as a baseline of the
	numbering $\gamma$. Otherwise $w$ is in a bead towards
	$z_{i-1}$ and we cut this branch of the certificate. When all
	edges in the bead of $t'$ are traversed, we proceed with visit
	in the standard way.

	Given the order we take $b_i$, each bead is not added more
	than once to a certificate $C_i$, therefore the total cost
	over the spine is $O(|E_X|)$.
	Nevertheless, the internal edges ${E_X}_i'$ of the first bead
	in common between $B_{z_i,t}$ and $B_{z_{i-1},t}$ are explored
	for each back edge $b_i$.
\end{proof}

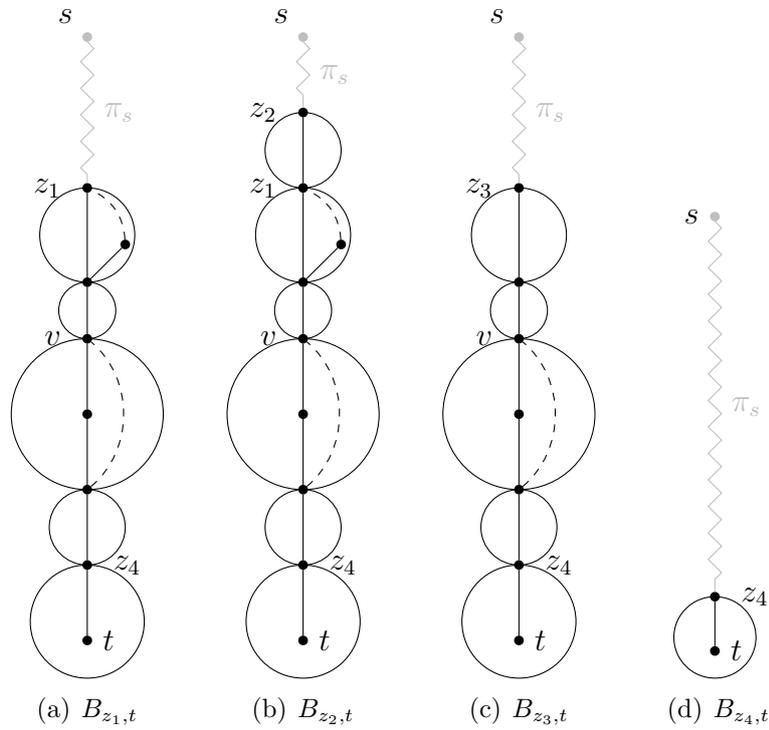
\begin{figure}[t]
\centering
\subfigure[$B_{z_1,t}$] {

\begin{tikzpicture}
[nodeDecorate/.style={shape=circle,inner sep=1pt,draw,thick,fill=black},%
  lineDecorate/.style={-,dashed},%
  elipseDecorate/.style={color=gray!30},
  scale=0.25]
\draw (10,23.5) circle (2.5);
\draw (10,19.5) circle (1.5);
\draw (10,14) circle (4);
\draw (10,8) circle (2);
\draw (10,3) circle (3);

\node (s) at (10,34) [nodeDecorate,color=lightgray,label=above left:$s$] {};

\node (z2) at (10,26) [nodeDecorate,label=left:$z_1~$] {};
\node (a) at (10,21) [nodeDecorate] {};
\node (b) at (12,23) [nodeDecorate] {};
\node (v) at (10,18) [nodeDecorate,label=left:$v~$] {};
\node (c) at (10,14) [nodeDecorate] {};
\node (d) at (10,10) [nodeDecorate] {};
\node (z4) at (10,6) [nodeDecorate,label=right:$~z_4$] {};
\node (t) at (10,2) [nodeDecorate,label=right:$t$] {};

\path {
	(s) edge[snake,-,color=lightgray] node {\quad\quad$\pi_s$} (z2)
	(z2) edge node {} (a)
	(a) edge node {} (b)
	(a) edge node {} (v)
	(v) edge node {} (c)
	(c) edge node {} (d)
	(d) edge node {} (z4)
	(z4) edge node {} (t)
};

\path {
	(d) edge[dashed,bend left=-50] node {} (v)
	(b) edge[dashed,bend left=-30] node {} (z2)
};
\end{tikzpicture}
}
\quad
\subfigure[$B_{z_2,t}$] {

\begin{tikzpicture}
[nodeDecorate/.style={shape=circle,inner sep=1pt,draw,thick,fill=black},%
  lineDecorate/.style={-,dashed},%
  elipseDecorate/.style={color=gray!30},
  scale=0.25]
\draw (10,28) circle (2);
\draw (10,23.5) circle (2.5);
\draw (10,19.5) circle (1.5);
\draw (10,14) circle (4);
\draw (10,8) circle (2);
\draw (10,3) circle (3);

\node (s) at (10,34) [nodeDecorate,color=lightgray,label=above left:$s$] {};

\node (z1) at (10,30) [nodeDecorate,label=left:$z_2~$] {};
\node (z2) at (10,26) [nodeDecorate,label=left:$z_1~$] {};
\node (a) at (10,21) [nodeDecorate] {};
\node (b) at (12,23) [nodeDecorate] {};
\node (v) at (10,18) [nodeDecorate,label=left:$v~$] {};
\node (c) at (10,14) [nodeDecorate] {};
\node (d) at (10,10) [nodeDecorate] {};
\node (z4) at (10,6) [nodeDecorate,label=right:$~z_4$] {};
\node (t) at (10,2) [nodeDecorate,label=right:$t$] {};

\path {
	(s) edge[snake,-,color=lightgray] node {\quad\quad$\pi_s$} (z1)
	(z1) edge node {} (z2)
	(z2) edge node {} (a)
	(a) edge node {} (b)
	(a) edge node {} (v)
	(v) edge node {} (c)
	(c) edge node {} (d)
	(d) edge node {} (z4)
	(z4) edge node {} (t)
};

\path {
	(d) edge[dashed,bend left=-50] node {} (v)
	(b) edge[dashed,bend left=-30] node {} (z2)
};
\end{tikzpicture}
}
\quad
\subfigure[$B_{z_3,t}$] {

\begin{tikzpicture}
[nodeDecorate/.style={shape=circle,inner sep=1pt,draw,thick,fill=black},%
  lineDecorate/.style={-,dashed},%
  elipseDecorate/.style={color=gray!30},
  scale=0.25]
\draw (10,23.5) circle (2.5);
\draw (10,19.5) circle (1.5);
\draw (10,14) circle (4);
\draw (10,8) circle (2);
\draw (10,3) circle (3);

\node (s) at (10,34) [nodeDecorate,color=lightgray,label=above left:$s$] {};

\node (z3) at (10,26) [nodeDecorate,label=left:$z_3~$] {};
\node (a) at (10,21) [nodeDecorate] {};
\node (v) at (10,18) [nodeDecorate,label=left:$v~$] {};
\node (c) at (10,14) [nodeDecorate] {};
\node (d) at (10,10) [nodeDecorate] {};
\node (z4) at (10,6) [nodeDecorate,label=right:$~z_4$] {};
\node (t) at (10,2) [nodeDecorate,label=right:$t$] {};

\path {
	(s) edge[snake,-,color=lightgray] node {\quad\quad$\pi_s$} (z3)
	(z3) edge node {} (a)
	(a) edge node {} (v)
	(v) edge node {} (c)
	(c) edge node {} (d)
	(d) edge node {} (z4)
	(z4) edge node {} (t)
};

\path {
	(d) edge[dashed,bend left=-50] node {} (v)
};
\end{tikzpicture}
}
\quad
\subfigure[$B_{z_4,t}$] {

\begin{tikzpicture}
[nodeDecorate/.style={shape=circle,inner sep=1pt,draw,thick,fill=black},%
  lineDecorate/.style={-,dashed},%
  elipseDecorate/.style={color=gray!30},
  scale=0.18]
\draw (10,3) circle (3);

\node (s) at (10,34) [nodeDecorate,color=lightgray,label=left:$s$] {};

\node (z4) at (10,6) [nodeDecorate,label=right:$~z_4$] {};
\node (t) at (10,2) [nodeDecorate,label=right:$t$] {};

\path {
	(s) edge[snake,-,color=lightgray] node {\quad\quad$\pi_s$} (z4)
	(z4) edge node {} (t)
};

\end{tikzpicture}
}
\caption{Certificates of the left branches of a spine}
\label{fig:spine-left-update}
\end{figure}

Although the edges in ${E_X}_i'$ are in a common bead between
$B_{z_i,t}$ and $B_{z_{i-1},t}$, these edges must be visited. The
entry point in the common bead can be different for $z_i$ and
$z_{i-1}$, the DFS tree of that bead can also be different. For an
example, consider the case where $z_i, \ldots, z_j$ are all in the
same bead after the removal of $u$. The bead strings $B_{z_i,t} \ldots
B_{z_j,t}$ are the same, but the roots $z_i, \ldots, z_j$ of the
certificate are different, so we have to compute the corresponding DFS
of the first bead $|j-i|$ times. Note that this is not the case for
the other beads in common: the entry point is always the same.

\begin{lemma}
	The cost $O(|E_X| + \sum_{i>1}{|{E_X}_i'}|)$ on a spine of the
	recursion tree can be amortized to $O(|E_X|)$.
	\label{lem:left-amortize}
\end{lemma}
\begin{proof}
	We can charge the cost $O(|{E_X}_i'|)$ of exploring the edges
	in the first bead in common between $B_{z_i,t}$ and
	$B_{z_{i-1},t}$ to another node in the recursion tree. Since
	this common bead is the head of at least one certificate in
	the recursion subtree of the left child of the $i$th node of
	the spine.  Specifically, we charge the first and only node in
	the \emph{leftmost} path of the $i$th child of the spine that
	has exactly the edges ${E_X}_i'$ as head of its bead string:
	(i) if $|{E_X}_i'| \le 1$ it corresponds to a unary node or a
	leaf in the recursion tree and therefore we can charge it with
	$O(1)$ cost; (ii) otherwise it corresponds to a first node of
	a spine and therefore we can also charge it with
	$O(|{E_X}_i'|)$. We use this charging scheme when $i \neq 1$
	and the cost is always charged in the leftmost recursion path
	of $i$th node of the spine, consequently we never charge a
	node in the recursion tree more than once. 
\end{proof}

Lemmas~\ref{lem:cost-spine-not-amortized} and~\ref{lem:left-amortize}
finalize the proof of Lemma~\ref{lem:promotebackedge}. 
Fig.~\ref{fig:spine-left-update} shows the certificates of bead strings
$B_{z_i,t}$ on the left branches of the spine from
Figure~\ref{fig:spine}.

\chapter{Conclusion and future work}
\label{chapter:Conclusion}

This thesis described optimal algorithms for several relevant problems
to the topic of listing combinatorial patterns in graphs. These
algorithms share an approach that appears to be general and applicable
to diverse incarnations of the problem.

\medskip

We finish this exposition by considering future directions for the
work performed in this thesis.

\begin{enumerate}
	\item The results obtained can have practical implications in
		several domains. An \emph{experimental analysis} of
		the performance of the algorithms would be beneficial
		continuation of the work performed. We highlight the
		problem of listing cycles as the first candidate for
		experimentations.
	\item Defining a \emph{model of dynamic data structures} that
		support non-arbitrary operations has the potential to
		become an important tool in this and other topics.
		One possible starting point could be a stacked
		fully-dynamic data-structure model. In this model, to
		apply an operation on an edge $e$ (or vertex $v$)
		would imply the undoing of every operation done since
		the last operation on edge $e$ (resp.~vertex $v$). 
	\item A \emph{general formulation} of the technique, defining
		the requirements of the pattern and setting of the
		problem, would allow the application of the technique
		as a black box. This would likely require some work on
		the previous point. A possible difficulty in doing so
		is that the amortized analysis is specific to on the
		pattern being listed. This difficulty could be
		circumvented by parameterizing the time complexity of
		maintaining the dynamic data structure in function of
		the number of patterns it guarantees to exist in the
		input.
	\item Application to \emph{additional problems} of listing,
		such as bubbles, induced paths and holes described in
		Section~\ref{section:CombinatorialPatternsInGraphs},
		would likely achieve positive results.
	\item Another idea that likely would lead to positive results
		is an extension of the analysis of the algorithms and
		their time complexity to take into account
		\emph{succinct encodings of the output}. One clear
		point where this would be useful, would be the on the
		problem of listing $st$-paths and cycles. In this
		case, we believe that the algorithm is already
		optimally sensitive on the size of the front coding of
		the output.
	\item Further investigation of the applications of the
		technique to \emph{searching and indexing}
		of patterns in graphs is recommended.  Connected to
		the previous point, succinct representations of the
		output can be used as an index.
\end{enumerate}

\bibliographystyle{plain}
\bibliography{thesisbib}

\printindex
\end{document}